\documentclass[10pt]{amsart}
\usepackage[utf8]{inputenc}
\usepackage[a4paper,width=170mm,top=15mm,bottom=15mm,left=15mm, right=15mm]{geometry}
\usepackage{etex}
\usepackage[all]{xy}
\usepackage{etoolbox}
\usepackage{lmodern}
\usepackage{array}
\usepackage{xstring}
\usepackage{longtable}
\usepackage[listings]{tcolorbox}
\tcbuselibrary{breakable}
\tcbuselibrary{skins}
\usepackage[pdftex]{hyperref}
\usepackage{amsmath}
\usepackage{tikz,tikz-cd,stmaryrd}
\usepackage{fancyhdr}
\usepackage{amssymb}
\usepackage{amsthm}
\usepackage{amsmath}
\usepackage{amsfonts}
\usepackage{amsthm}
\usepackage{pifont}
\usepackage{amsmath}
\usepackage{epsfig}
\usepackage{latexsym}
\usepackage{amssymb}
\usepackage{color}
\usepackage{amscd}
\usepackage{multirow}
\usepackage{graphicx}
\usepackage{lscape}
 \usepackage{tabularx}
\usepackage{float}
\usepackage{xcolor}
\usepackage{bm}
\usepackage{tikz-cd}
 \makeatletter
    \let\stdchapter\section
    \renewcommand*\section{%
    \@ifstar{\starchapter}{\@dblarg\nostarchapter}}
    \newcommand*\starchapter[1]{%
        \stdchapter*{#1}
        \thispagestyle{fancy}
        \markboth{\MakeUppercase{#1}}{}
    }
    \def\nostarchapter[#1]#2{%
        \stdchapter[{#1}]{#2}
        \thispagestyle{fancy}
    }
\makeatother

\newtcolorbox{boxA}{
    fontupper = \bf,
    boxrule = 1.5pt,
    colframe = black % frame color
}

\newtheorem{theorem}{Theorem}[section]
\newtheorem*{theorem*}{Theorem}
\newtheorem{lemma}[theorem]{Lemma}
\newtheorem{proposition}[theorem]{Proposition}
\theoremstyle{definition}
\newtheorem{definition}[theorem]{Definition}
\theoremstyle{corollary}
\newtheorem{corollary}[theorem]{Corollary}
\theoremstyle{remark}
\newtheorem{remark}[theorem]{Remark}
\newtheorem{example}[theorem]{Example}
\theoremstyle{conclusion}

\usepackage{ulem}

\usepackage{amsmath}
\usepackage{amssymb}
\usepackage{amsfonts}
\usepackage{esvect}
\usepackage{mathrsfs}
\usepackage{authblk}
\usepackage{geometry}
\usepackage{abstract}
\usepackage{xcolor}

 \makeatletter
\@addtoreset{equation}{section} %\@addtoreset{equation}{Subsection}
\makeatother

  \usepackage{color}

\begin{document}

%\maketitle
\thispagestyle{empty}
\begin{center}
	\Large{{\bf     On the construction of polynomial Poisson algebras: a novel grading approach}  }
\end{center}
\vskip 0.5cm
\begin{center}
	\textsc{Rutwig Campoamor-Stursberg$^{1,\star}$, Danilo Latini$^{2,*}$, Ian Marquette$^{3,\bullet}$, \\
    Junze Zhang$^{4,\dagger}$ and Yao-Zhong Zhang$^{4,\ddagger}$}
\end{center}
\begin{center}
	$^1$ Instituto de Matem\'{a}tica Interdisciplinar and Dpto. Geometr\'{i}a y Topolog\'{i}a, UCM, E-28040 Madrid, Spain
\end{center}

\begin{center}
	$^2$ Dipartimento di Matematica “Federigo Enriques”, Università degli Studi di Milano, Via C. Saldini 50, 20133 Milano, Italy \& INFN Sezione di Milano, Via G. Celoria 16, 20133 Milano, Italy
\end{center}

\begin{center}
	$^3$ Department of Mathematical and Physical Sciences, La Trobe University, Bendigo, VIC 3552, Australia
\end{center}

\begin{center}
	$^4$ School of Mathematics and Physics, The University of Queensland, Brisbane, QLD 4072, Australia
\end{center}
\begin{center}
	\footnotesize{$^\star$\textsf{rutwig@ucm.es} \hskip 0.25cm$^*$\textsf{danilo.latini@unimi.it} \hskip 0.25cm $^\bullet$\textsf{i.marquette@latrobe.edu.au} \hskip 0.25cm
 $^\dagger$\textsf{junze.zhang@uq.net.au} \hskip 0.25cm
 $^\ddagger$\textsf{yzz@maths.uq.edu.au}}
\end{center}
\vskip  1cm
\hrule

%\begin{center}
%	{\bf Abstract}
%\end{center}
\begin{abstract}
\vskip 0.15cm
\noindent 
 In this work, we refine recent results on the explicit construction of polynomial algebras associated with commutants of subalgebras in enveloping algebras of Lie algebras by considering an additional grading with respect to the subalgebra. It is shown that such an approach simplifies and systematizes the explicit derivation of the Lie--Poisson brackets of elements in the commutant, and several fundamental properties of the grading are given. The procedure is illustrated by revisiting three relevant reduction chains associated with the rank-two complex simple Lie algebra $\mathfrak{sl}(3,\mathbb{C})$. Specifically, we analyze the reduction chains $\mathfrak{so}(3) \subset \mathfrak{su}(3)$, corresponding to the Elliott model in nuclear physics, the chain $\mathfrak{o}(3) \subset \mathfrak{sl}(3,\mathbb{C})$ associated with the decomposition of the enveloping algebra of $\mathfrak{sl}(3,\mathbb{C})$ as a sum of modules, and the reduction chain $\mathfrak{h} \subset \mathfrak{sl}(3,\mathbb{C})$ connected to the Racah algebra $R(3)$. In addition, a description of the classification of the centralizer with respect to the Cartan subalgebra $\mathfrak{h}$ associated with the classical series $A_n$ in connection with its root system is reconsidered. As an illustration of the procedure, the case of $S(A_3)^\mathfrak{h}$ is considered in detail, which is connected with the rank-two Racah algebra for specific realizations of the generators as vector fields. This case has attracted interest with regard to orthogonal polynomials.

\end{abstract}
\vskip 0.35cm
\hrule

\tableofcontents

\section{Introduction}

The interplay between algebraic structures and dynamical systems has long been a topic of utmost importance in mathematical physics, providing a deep understanding of the behavior of complex physical systems \cite{MR1184379,MR1255424}. Among these algebraic structures, finitely generated polynomial algebras stand out for their versatility and robustness in modeling various physical phenomena \cite{MR1173277,MR1306244,MR1806263,MR2804560}. In this context, the study of superintegrable systems (i.e., systems possessing more integrals of motion than degrees of freedom) has garnered significant attention, due to their rich mathematical properties and physical applications \cite{MR3493688,MR3797912,Correa2020}. For example, quadratic, cubic and higher-order polynomial deformations of Lie algebras have played a significant role in the construction of superintegrable systems \cite{MR1806263,MR3227700,MR3119484,MR2492581,MR2566882,MR4256391,MR4131326, Latini_2019, LATINI2021168397, Latini_2021}, and have shown a deep connection with the theory of orthogonal polynomials and generalized special functions \cite{marquette2023infinite, bernard2023bethe}.

Recently, it was pointed out how the general labeling problem related to subalgebra chains is of significance in several applications, such as nuclear physics, the decomposition of enveloping algebras, or the embedding problem of higher-rank Racah algebras. All these examples are, in fact, connected with polynomial Poisson algebras \cite{MR0411412,MR2276736,MR4710584,campoamor2023algebraic}.  It was shown that the missing-label problem, \cite{MR4411095} when considering labeling operators generically, leads to non-Abelian algebras taking the form of finitely generated polynomial Poisson algebras. This construction method is based on the centralizer subalgebra (i.e. commutant) in the symmetric algebra  associated with a Lie algebra. The computation of indecomposable polynomials generating the commutant is carried out in the context of Poisson brackets.  Although the construction has been known since the 1950s in nuclear physics models and Lie algebra representations, the underlying algebraic structure remains largely unexplored. Recent studies reveal that these structures correspond to finitely generated polynomial algebras. These algebras can be structures with three generators or more complex, often labeled higher-rank algebras \cite{MR4660510}. It has been shown that computational methods can be made easier using Poisson algebras and dual space variables. This simplification hinges on the robust link between the universal enveloping algebra and the symmetric space associated with dual-space variables within the Poisson-Lie framework. The symmetrization map transitions between these are akin to the classical/quantum correspondence. Despite this, it remains challenging to close these algebraic relations. As the dimension of the Lie algebra grows, so does the degree of labeling operators, complicating the construction of all monomials.

%This allowed us to discuss several examples that were used in the context of nuclear physics.

 In this paper, we re-examine the subalgebra chains that are studied from $3$ distinct examples, which are from physical models, the decomposition of the enveloping algebra, and analyze recent examples pertaining to the Cartan commutant. Specifically,  we show the closure of these polynomial Poisson algebras through the grading of their generators. The structure of the paper is as follows: in Section $\ref{2}$, we discuss the properties and construction of commutants in the context of Lie algebras and their corresponding symmetric algebras, and also illustrate a way to find the linearly independent and indecomposable generators of the centralizer. Next, in Section $\ref{3}$, we focus on the construction of polynomial algebras using subalgebra centralizers. It emphasizes the terminology and conditions necessary to predict allowed polynomials in non-trivial brackets. The construction relies on grading by degrees and identifying specific forms of monomials, which facilitates determining the total count of possible terms in non-zero commutator relations. With the general terminologies of Section $\ref{3}$, in Section $\ref{4}$  we apply the grading method to identify potentially permissible monomials in the expansion of a non-trivial bracket. To some extent, by determining the grading of each generator in polynomial algebras, the construction of the compact form is simplified. We consider cases related to Cartan generators, simple roots, and $\mathfrak{so}(3)$ tensor operators.
 % We also highlight the use of appropriate bases to restrict the number of monomials and facilitate the construction of commutator relations, leveraging symbolic computing packages to find linearly independent and indecomposable polynomials. %Additionally, the grading method is used to constrain permissible monomials and determine the finiteness of the generated algebra. This involves calculating gradings and applying commutator relations to achieve a compact form.

Finally, in Section \ref{5}, we undertake a comprehensive and detailed analysis of the centralizer concerning the Cartan subalgebra within the symmetric algebra. This examination includes a classification that meticulously outlines the allowable monomials constrained by the presence of non-trivial brackets. A thorough exploration is conducted into both non-decomposable and decomposable terms, distinguishing cases where Cartan elements are involved from those where they are not, within non-trivial brackets. The study provides explicit forms for the expansion in each scenario. %By emphasizing the intricate connection between root system structures and grading frameworks, this discussion ultimately links the findings to the conjugation classes associated with the Weyl group as well as to the equivalences found among permutation groups.

%    Examples involving complex semisimple Lie algebras and triangular decomposition are provided. The text concludes by proposing a generic method to find the grading in non-trivial bracket relations of polynomial algebras, highlighting the relationship between grading, commutator relations, and Lie algebra decomposition.
\section{Preliminaries}
\label{2}

Throughout this work, let $\mathbb{F}$ be either $\mathbb{R}$ or $\mathbb{C},$ and let $\mathbb{N}$ be the ring of integers.

\medskip

Let $\mathfrak{g}$ be a $n$-dimensional semi-simple or reductive Lie algebra over a field $\mathbb{F}$ with a non-trivial commutator $[\cdot,\cdot]$ satisfying the relations \begin{equation}
[X_i,X_j]= \sum_{k=1}^n C_{ij}^k X_k  \qquad 1 \leq i,j \leq n .
\end{equation} Here $C_{ij}^k \in \mathbb{F} $ are the structure constants of $\mathfrak{g}.$ Let $\left(U(\mathfrak{g}),[\cdot,\cdot]\right)$ be the enveloping algebra of $\mathfrak{g}.$ Suppose that $\beta_\mathfrak{g} = \{X_1,\ldots,X_n\}$ is an ordered basis for $\mathfrak{g}$.   The dual space $\mathfrak{g}^*$ admits a Poisson-Lie bracket in terms of linear coordinates $\{x_1,...,x_n\}$ determined by a dual basis of $\beta_\mathfrak{g}$, for which one has:   \begin{equation}
\{x_i,x_j\}= \sum_{k=1}^n C_{ij}^k x_k \qquad 1 \leq i,j \leq n .
\end{equation} The symmetric algebra $S(\mathfrak{g}) \cong \mathbb{F}[\mathfrak{g}^*]$ becomes a Poisson polynomial algebra under the same Lie-Poisson bracket $\{\cdot,\cdot\}$ defined on $\mathfrak{g}^*$. That is, for any $p,q \in S(\mathfrak{g}),$ a Poisson-Lie bracket $\{\cdot,\cdot\}:S(\mathfrak{g}) \times S(\mathfrak{g}) \rightarrow S(\mathfrak{g})$ is defined by \begin{equation}
\{p,q\}= \sum_{  i,j,k = 1}^n C_{ij}^k x_k \dfrac{\partial p }{\partial x_i }\dfrac{\partial q}{\partial x_j } .
\end{equation}

 We now construct the symmetric algebra for an arbitrary subalgebra $\mathfrak{a}$ of $\mathfrak{g}$.  Without loss of generality, we denote the generators of $\mathfrak{a}$ by $X_{\ell_1},\ldots,X_{\ell_s}$, where $s = \dim \mathfrak{a}$ and $\{\ell_1,\ldots,\ell_s\} \subset \{1,\ldots,n\}$. We now look at the adjoint action of the subalgebra $\mathfrak{a}$ on the enveloping algebra and the symmetric algebra, respectively. From the universal property of $U(\mathfrak{g}),$ the adjoint action of $\mathfrak{a}$ on $U(\mathfrak{g})$ preserves the same commutator defined on $\mathfrak{g}.$ For the symmetric algebra, we can then show the following statement.

 \begin{proposition}
 \label{p2.1}
     The coadjoint action of $\mathfrak{a}$ on $S(\mathfrak{g})$ preserves the Poisson-Lie bracket $\{\cdot,\cdot\}$ defined on $S(\mathfrak{g}).$ That is, for all $X_m \in \mathfrak{a},$ and $p,q \in S(\mathfrak{g}), $ \begin{align*}
         \mathrm{ad}_{X_m}^* (\{p,q\}) = \{ \mathrm{ad}_{X_m}^*(p),q\} + \{p, \mathrm{ad}_{X_m}^*(q)\} \, , \text{ } m = \ell_1,\ldots,\ell_s.
     \end{align*}
 \end{proposition}

The proposition $ \ref{p2.1}$ motivates the following definition.

\begin{definition}
\label{2.3} \cite{MR0760556}
%  Let $\mathfrak{g}$ be a finite-dimensional Lie algebra and $\mathfrak{g}^*$ be its dual. Let $(U(\mathfrak{g}),[\cdot,\cdot])$ and $(S(\mathfrak{g}),\{\cdot,\cdot\})$ be the universal enveloping algebra and associated symmetry algebra of $\mathfrak{g}$, respectively.
The coadjoint orbits of $\mathfrak{a}$ on the symmetric algebra $S(\mathfrak{g})$ are given by \begin{align}
 %   P\left( X_1,\ldots,X_n\right) \in U(\mathfrak{g}) &\mapsto  [X_j,P] \in U(\mathfrak{g}) , \\
      p(x_1,\ldots,x_n) \in S(\mathfrak{g})   \mapsto \{x_m,p  \} = \tilde{X}_m(p ) = \sum_{  k,l =1}^n C_{mk}^lx_l \dfrac{\partial p}{\partial x_k} \in S(\mathfrak{g}), \text{ } m = \ell_1,\ldots,\ell_s \label{eq:dual}
  \end{align} where $\tilde{X}_m = \sum_{  k,l =1}^n C_{mk}^l x_l \dfrac{\partial}{\partial x_k}$ are vector field realizations of the generators of $\mathfrak{a}.$
\end{definition}

According to Definition $\ref{2.3},$ our primary focus is to explore the kernel of the coadjoint action of $\mathfrak{a}$ on $S(\mathfrak{g})$.

\begin{definition}
\label{2.1}
 The $\textit{commutant}$ (or $\textit{centralizer subalgebra}$) $ S(\mathfrak{g})^\mathfrak{a}$ is defined as the centralizer of $\mathfrak{a}^*$ in $S(\mathfrak{g})$:
\begin{align*}
  %   U(\mathfrak{g})^{\mathfrak{a}}  = &\left\{ P \in U(\mathfrak{g}): \text{ } [X,P] = 0  \quad \forall X \in \mathfrak{a}\right\},  \\
      S(\mathfrak{g})^\mathfrak{a}   =& \left\{p \in S(\mathfrak{g}): \text{ } \{x,p\} = 0  \quad \forall x \in \mathfrak{a}^*\right\},
  \end{align*} where $p$ is polynomial in terms of $x_j$ for all $1 \leq j \leq n$.
\end{definition}
\begin{remark}
\label{2.2}
 (i) The $\textit{Poisson center}$ of $\left(S(\mathfrak{g}),\{\cdot,\cdot\}\right)$ is the set of all $\mathfrak{g}$-invariant polynomials, i.e., \begin{align*}
    S(\mathfrak{g})^\mathfrak{g}  =\left\{ p \in S(\mathfrak{g}) : \text{ } \{p,x\} = 0 \quad \forall x \in \mathfrak{g}^*\right\}.
\end{align*} These elements can be identified with the (polynomial) solutions of the differential operators in \eqref{eq:dual}  (see \cite{MR4355741,MR1451138} for details).

 (ii) We note that for any finite-dimensional Lie algebras, a linear basis in $S(\mathfrak{g})^\mathfrak{a}$ is not necessarily standard. In other terms, one cannot guarantee that the Poisson algebra is finitely-generated. However,  if $\mathfrak{g}$ is semisimple or reductive, it can be shown that $S(\mathfrak{g})$ is Noetherian, from which we can conclude that the centralizer $S(\mathfrak{g})^\mathfrak{a}$ is finitely-generated \cite[Chapter 2]{MR1451138}. This implies that, once a maximal set of indecomposable polynomials $\left\{p_{k_1},\dots ,p_{k_n}\right\}$ has been found, there always exists some integer $\zeta \in \mathbb{N}$ such that $p_{\zeta+k_j}$ is decomposable for all $j \geq 1$. By saying that a polynomial $p \in S(\mathfrak{g})$ is decomposable, we indicate that there exists another polynomial $p' \in S(\mathfrak{g})$ of a lower degree such that $p \equiv 0 \mod p'$, which means that $p'$ is a divisor of $p$. Additionally, it is important to note that the elements in the generating set within the centralizer of a subalgebra do not necessarily imply their algebraic independence.
\end{remark}
 %Hence, we always assume that the system $\eqref{eq:func}$ admits an integrity basis formed by polynomials.

For any $h \in \mathbb{N}_0$, we define
\begin{equation*}
U_h(\mathfrak{g}) = \mathrm{span} \{X_1^{i_1} \ldots X_n^{i_n}: i_1+ \ldots + i_n \leq  h\}
\end{equation*}
as the linear subspace of $U(\mathfrak{g})$ spanned by polynomials of degree at most $h$ in the (noncommutative) generators of $\mathfrak{g}$. The degree $\delta$ of an arbitrary element $P \in U(\mathfrak{g})$ is defined as $\delta := {\rm inf}\{k: P \in U_{k}(\mathfrak{g})\}$.  Furthermore, there is a natural filtration in $U(\mathfrak{g})$ given by the following relations \cite{MR1451138}
\begin{align}
    U_{0}(\mathfrak{g}) = \mathbb{F}, \quad U_\ell(\mathfrak{g})U_k(\mathfrak{g}) \subset U_{\ell+k}(\mathfrak{g}), \quad U_\ell(\mathfrak{g}) \subset U_{\ell+k}(\mathfrak{g}),\quad \forall k, \ell\geq 1.
\end{align}
Clearly, for each $k\geq 1$ we can define the blocks $ U^0(\mathfrak{g}) = \mathbb{F}$ and  $U^k(\mathfrak{g})=U_k(\mathfrak{g})/U_{k-1}(\mathfrak{g})$,  from which it follows that we have a graded algebra $\mathrm{gr} \, U(\mathfrak{g})  := \bigoplus_{k\geq 0} U^k(\mathfrak{g})$. Here we set $U^0 (\mathfrak{g}) = \{0\}$. From the Poincar\'e-Birkhoff-Witt (PBW in short) theorem, it can be easily deduced that the dimension of each gradation block is
\begin{align}
    \dim  U^k(\mathfrak{g}) = \dim \frac{U_k(\mathfrak{g})}{ U_{k-1}(\mathfrak{g})} = \binom{\dim \mathfrak{g} +k-1}{k}.
\end{align} Now, back to the symmetric algebra $S(\mathfrak{g})$. By definition, we deduce the decomposition
\begin{align}
   S(\mathfrak{g}) = \bigoplus_{k \geq 0} S^k(\mathfrak{g})  ,
\end{align} where \begin{align*}
     S^k(\mathfrak{g}) : = \mathrm{span} \left\{x_1^{a_1} \cdots x_n^{a_n}: a_1 + \ldots + a_n = k, \quad a_j \in \mathbb{N}_0 : =\mathbb{N} \cup \{0\}\right\} 
 \end{align*} is a subalgebra of $S(\mathfrak{g})$ consisting of all the degree $k$ polynomials. It follows that, for any $p \in S(\mathfrak{g})$, the polynomial decomposes as $p = \sum_{k \geq 0} p^{(k)}$, where $p^{(k)} \in S^k(\mathfrak{g})$ for all $k \geq 0$.  A linear isomorphism $\Lambda:S(\mathfrak{g})\rightarrow U(\mathfrak{g})$ that commutes with the adjoint action is easily obtained through the symmetrization map
\begin{equation}
\Lambda\left(x_{j_1}\cdots x_{j_k}\right)=\frac{1}{k!} \sum_{\sigma\in S_{k}} X_{j_{\sigma(1)}}\cdots X_{j_{\sigma(k)}},\label{eq:syma}
\end{equation}
where $S_k$ denotes the symmetric group of order $k!$, where $k! $ is the factorial of $k \in \mathbb{N}.$ Note that $\Lambda$ defines a vector space isomorphism. In particular, for any $p \in S^{k}(\mathfrak{g}) $ and $q \in S^{\ell}(\mathfrak{g}),$ we have \begin{align*}
    \deg \left( \Lambda(pq)\right) = k + \ell , \qquad   \Lambda(pq) -  \Lambda(p) \Lambda(q) \in U_{k+ \ell-1}(\mathfrak{g}).
\end{align*}  Note that \eqref{eq:syma} induces an algebra isomorphism $\tilde{\Lambda}: S(\mathfrak{g}) \rightarrow \mathrm{gr} \, U(\mathfrak{g})$. It follows that $U^k(\mathfrak{g})= \tilde{\Lambda}_k\left(S^k(\mathfrak{g})\right)$ is an algebra isomorphism,  where $\tilde{\Lambda}_k := \tilde{\Lambda}\vert_{S^k(\mathfrak{g})}$.

We now focus on $\mathfrak{a}$-invariant homogeneous polynomial spaces within $\left(S(\mathfrak{g}),\{\cdot,\cdot\}\right)$.   For similar constructions, we refer to \cite{MR191995,MR1520346, marquette2023algebraic} and the citations therein.
Define the vector space of $\mathfrak{a}$-invariant $k$-homogeneous polynomials as
\begin{align*}
 S^k(\mathfrak{g})^{\mathfrak{a}}  = \left\{p^{(k)} \in S^k(\mathfrak{g}): \{x,p^{(k)}\} = 0  \quad \forall x \in \mathfrak{a}^*\right\},
\end{align*}
where $p^{(k)}(x_1,\ldots,x_n)$ is a homogeneous polynomial of degree $k \geq 0$ with the generic form
\begin{align}
    p^{(k)}(x_1,\ldots,x_n) = \sum_{i_1 + \cdots + i_n = k} \Gamma_{i_1,\ldots, i_n}\, x_1^{i_1} \cdots x_n^{i_n}, \qquad  \Gamma_{i_1,\ldots,  i_n} \in \mathbb{F}. \label{eq:ci}
\end{align} %On the other hands, define \begin{align}
 %   U^k(\mathfrak{g})^{\mathfrak{a}}  &= \left\{ Y \in U^k(\mathfrak{g}): [X,Y] = 0  \quad \forall X \in \mathfrak{a}\right\}. \label{eq:sin}
%\end{align}
%By construction, it follows that $  U (\mathfrak{g})^\mathfrak{a} = \bigoplus_{k \geq 0}   U^k(\mathfrak{g})^{\mathfrak{a}}  $ and $ S (\mathfrak{g})^\mathfrak{a} = \bigoplus_{k \geq 0}  S^k(\mathfrak{g})^{\mathfrak{a}}  . $ In particular, $\Lambda_k$ induces an algebra isomorphism between $ U^k (\mathfrak{g})^\mathfrak{a} $ and $S^k (\mathfrak{g})^\mathfrak{a}$.
By definition, in order to find a finite generating set for centralizer subalgebras, all $\mathfrak{a}$-invariant linearly independent and indecomposable homogeneous polynomial solutions of the system of partial differential equations (PDEs)
\begin{align}
    \tilde{X}_m\left(p^{(k)}\right)(x_1,\ldots,x_n) = \left\{x_m,p^{(k)}\right\}  = \sum_{1 \leq l,i \leq n} C_{mi}^lx_l \dfrac{\partial p^{(k)}}{\partial x_i} = 0, \text{ }\quad m = \ell_1,\ldots,\ell_s  \label{eq:func}
\end{align}
must be found.  Note that it is shown that, if $\mathfrak{a} = \mathfrak{g},$ the maximal number of functionally independent solutions of $\eqref{eq:func}$ is known to be given by \cite{MR0411412,MR0204094}
\begin{align}
    \mathcal{N} (\mathfrak{g}) = \dim  \mathfrak{g} - \mathrm{rank}(A_{ij} ), \text{ } 1 \leq  i,j  \leq n,
\end{align} where $A_{ij} :=\sum_{l=1}^n C_{ij}^l x_l$ represents the matrix of the commutator table of the Lie algebra $\mathfrak{g}$ over the given basis. In this context, as $\mathfrak{a}$ is a subalgebra of $\mathfrak{g}$, we will consider the labeling problem where the functions (not necessarily polynomials) satisfy the system of PDEs $\eqref{eq:ci}$. It can be shown that the number of functionally independent solutions in the system $\eqref{eq:ci}$ is exactly
\begin{align}
    \mathcal{N}(\mathfrak{a}) = \dim \mathfrak{g} - \dim \mathfrak{a} + \ell_0.
\end{align}
Here $\ell_0$ is the number of $\mathfrak{g}$-invariant polynomials in $S(\mathfrak{a} )$  (for more details, see \cite{MR2276736,MR4660510,MR2515551} and \cite[Chapter 12, Section 12.1.5]{gtp} and references therein).

Typically, finding a polynomial for the centralizer in relation to a subalgebra can be approached by two methods: solving systems of  PDEs $\eqref{eq:func}$ directly, using the method of characteristics, or employing a polynomial ansatz. In this case, we apply the polynomial ansatz, as $\mathfrak{g}$ is assumed to be reductive, and the commutant can be expressed as a polynomial in the dual space variables. This simplifies the analysis into solving sets of linear equations. Note that, in the case of non-semisimple Lie algebras, the solution may be expressed as rational or even transcendental functions.

\subsection{Construction of commutant and related polynomial algebra: The polynomial ansatz}
%Although the ordering does not directly influence this, it will introduce correction terms in the Lie algebra context.
Let us consider polynomials within the framework of the Poisson bracket.  Recall that, by construction, the $\mathfrak{a}$-invariant homogeneous polynomial in $S^k(\mathfrak{g})^\mathfrak{a}$ takes the form of $\eqref{eq:ci}.$ Subsequently, a list of all polynomials is compiled for each degree, and we examine decomposability, that is, the degree up to which all polynomials can be expressed in terms of polynomials of lower degrees.  If indecomposability is achieved up to the degree $\zeta$, the set of polynomials that form the commutant is described by

\[ \textbf{q}_1 : = \left\{ p_u^{(1)},\quad u=1,...,l_1 \right\}; \]
\[ \textbf{q}_2 : = \left\{ p_u^{(2)},\quad u=1,...,l_2 \right\};\]
\[ \vdots  \]
\[  \textbf{q}_\zeta : = \left\{p_u^{(\zeta)},\quad u=1,...,l_\zeta\right\}. \]
%\[ p_j^{(1)},\quad i=1,...,l_j \]
%\[ ... \]
%\[ p_j^{(n)},\quad i=1,...,l_n \]
Here $p_u^{(k)} $ is an indecomposable $\mathfrak{a}$-invariant homogeneous polynomial of degree $k \in \{1,\ldots,\zeta\}$.

% Here  $\textbf{q}_k = \left\{p_i^{(k)}: \text{ } 1 \leq i \leq l_k\right\}.$
Let $\textbf{Q}_\zeta := \left\{p_u^{(1)},\ldots, p_u^{(\zeta)}  \right\} =   \textbf{q}_1 \sqcup \textbf{q}_2 \sqcup \ldots \sqcup \textbf{q}_\zeta $ be a finite set consisting of all indecomposable polynomials up to degree $\zeta$, and let $\textbf{Alg} \left\langle \textbf{Q}_\zeta \right\rangle $ denote the algebra generated by the set $\textbf{Q}_\zeta.$ It is clear that $\textbf{Alg}\langle\textbf{Q}_\zeta\rangle$ is infinite-dimensional as a vector space.  Notice that the elements in $\textbf{Q}_\zeta$ are not necessarily functionally (a.k.a. algebraically) independent. Hence, they are not freely generated, which means that there may exist non-trivial polynomial relations among these generators. Now, for any $p_u^{(l)} \in \textbf{q}_l$ and $p_v^{(\ell)} \in \textbf{q}_\ell$, there exist some coefficients $ \Gamma^{s_1,...,s_r}_{uv}  \in \mathbb{F}$ such that the Poisson-Lie bracket $\{\cdot,\cdot\}  :  \textbf{Alg} \langle  \textbf{Q}_\zeta \rangle \times \textbf{Alg} \langle  \textbf{Q}_\zeta \rangle\rightarrow \textbf{Alg} \langle  \textbf{Q}_\zeta \rangle$ is given by \begin{equation}
     \left\{p_u^{(l)} ,p_v^{(\ell)} \right\}= \sum_{k_1+\ldots + k_r =\ell+l-1} \Gamma^{s_1,...,s_r}_{uv} p_{s_1}^{(k_1)} \cdots p_{s_r}^{(k_r)}  . \label{eq:comm}
\end{equation} Here $k_1,\ldots,k_r  \leq \zeta.$  It is a straightforward calculation to confirm that the Leibniz rule holds true in this context. Subsequently, the algebraic structure denoted by $\textbf{Alg}\left\langle \textbf{Q}_\zeta\right\rangle$, when equipped with the Poisson-Lie bracket $\{\cdot,\cdot\}$, in conjunction with additional polynomial relations $ P(\textbf{q}_1,\ldots,\textbf{q}_\zeta) = 0$, constitutes a finitely-generated polynomial Poisson algebra. This algebraic framework thus upholds the properties characteristic of Poisson algebras, ensuring it is defined by a finite set of generators.   Let \[ d := \max_{1 \leq k_j \leq \zeta} \sum_{\begin{matrix}
    j \in I \\
    I = \left\{1,\ldots,r: p_{s_j}^{(k_j)} \notin \mathcal{Z}\right\}
\end{matrix}} k_j \] denote the degree of this polynomial Poisson algebra. Here $\mathcal{Z} := \{p \in \textbf{Alg}\left\langle \textbf{Q}_\zeta\right\rangle :  \{p,q\} = 0, \text{ } \forall q \in \textbf{Alg}\left\langle \textbf{Q}_\zeta\right\rangle \}$ is the center of $\textbf{Alg}\left\langle \textbf{Q}_\zeta\right\rangle.$   In the following, we denote $\mathcal{Q}_\mathfrak{g}(d)  := \left(\textbf{Alg} \left\langle \textbf{Q}_\zeta \right\rangle,\{\cdot,\cdot\} \right)$ for simplicity. By construction, we observe that \begin{align*}
    \mathcal{Q}_\mathfrak{g}(d) = \mathfrak{t} \oplus \bigoplus_{k \in \Omega} \mathcal{Q}_k
\end{align*}  is a graded polynomial algebra, where $\mathfrak{t} :={ \mathfrak{g}^*}^\mathfrak{a}$ is the centralizer of $\mathfrak{a}^*$ in $\mathfrak{g}^*$, $\Omega \subset \mathbb{N}$ is an ordered set and $\mathcal{Q}_k  $ is the vector space consisting of $\mathfrak{a}$-invariant polynomials of degree $k$. Note that $\mathfrak{t}$ is a subaglebra of $\mathcal{Q}_\mathfrak{g}(d)$. It is clear that $\mathfrak{t} = \mathfrak{a}^*$ if $\mathfrak{a}^*$ is Abelian. Thus, if $\mathfrak{a}^*$ is Abelian, it follows that $\mathfrak{t}$ is also Abelian; however, the reverse implication is generally invalid. Moreover, we observe the following fact regarding to the center element of $\mathcal{Q}_\mathfrak{g}(d)$.  
\begin{proposition}
    Let $\mathcal{Q}_\mathfrak{g}(d)$ be the same as defined above. Then for any $\mathfrak{t} \neq \{0\}$, $\mathfrak{t} \subset \mathcal{Z}$ if and only if $\mathfrak{t}$ is Abelian.
\end{proposition}
\begin{proof}
We first assume that $\mathfrak{t} \subset \mathcal{Z}$. Then by definition, $\{f, p\} = 0$ for all $f  \in \mathfrak{t}  $ and $p \in \mathcal{Q}_\mathfrak{g}(d)$. Without loss of generality, let $y := p \vert_\mathfrak{t}$. Since $y \in \mathfrak{t} \subset \mathcal{Q}_\mathfrak{g}(d)$, $\{f,y\} = 0$. Hence, $\mathfrak{t}$ is Abelian.

Conversely, assume that $\mathfrak{t}$ is Abelian. Hence, for any $p \in \mathcal{Q}_\mathfrak{g}(d)$ with $\deg p = 1,$ $\{\mathfrak{t},p\} = \{0\}$. For any $f \in \mathfrak{t}$ and $p \in \mathcal{Q}_\mathfrak{g}(d)$ with $\deg p \geq 2$, we aim to show that $ \{f,p\} = 0$. By definition, for any $\xi \in \mathfrak{g}^*$, there exists a vector field $\tilde{X}_f$ such that \begin{align}
    \tilde{X}_f(p)(\xi) = \{f,p\} (\xi) = \langle \xi, [df, dp]\rangle, \label{eq:def1}
\end{align} where $df, dp \in \mathfrak{g}$, and $\langle \cdot,\cdot \rangle$ is a dual pair between $\mathfrak{g}$ and $\mathfrak{g}^*$. Note that $f: \mathfrak{g}^* \rightarrow \mathbb{R} $ is a linear $\mathfrak{a}$-invariant polynomial. Since $\mathfrak{g}^*$ can be also viewed as a vector space, and its tangent space at any point is naturally identified with $\mathfrak{g}^*$ itself.  It follows, with abuse of notation, that $df(\xi) = f$ for all $\xi \in \mathfrak{g}^*$. Hence \eqref{eq:def1} becomes \begin{align}
      \tilde{X}_f(p)(\xi) = \{f,p\} (\xi) = \langle \xi, [f, dp]\rangle . \label{eq:def2}
\end{align}  Let $\phi_t(\xi) = \mathrm{Ad}^*(\exp(t f)) \xi $ be a flow defined on $\mathfrak{g}^*$, where $ t \in \mathbb{R}$, and $\mathrm{Ad}^*: A \times \mathfrak{g}^*\rightarrow \mathfrak{g}^*$ represents a coadjoint action. Here $A$ is a connected Lie group such that $\mathfrak{a} = Lie (A)$. Then by chain rule and definition of co-adjoint operator, \begin{align}
    \left.\dfrac{d}{dt} \right\vert_{t = 0} p\left(\phi_t(\xi)\right) = \left\langle d p(\phi_t(\xi)), \mathrm{ad}^*(f)(\phi_t(\xi))\right\rangle =  \left\langle \phi_t(\xi), [f,d p(\phi_t(\xi)) ]\right\rangle. \label{eq:def3}
\end{align} Here $\mathrm{ad}^*$ is the derivative of $\mathrm{Ad}^*.$ Back to \eqref{eq:def2}, we observe that \begin{align}
      \left.\dfrac{d}{dt} \right\vert_{t = 0} p\left(\phi_t(\xi)\right) =  \{f,p\} (\phi_t(\xi)) =  \tilde{X}_f(p)(\phi_t(\xi)). \label{eq:def4}
\end{align} Therefore, the directional derivative of $p$ along the flow of $f$ is exactly given by the Poisson bracket evaluated at $\phi_t$. Now, to show $f$ and $p$ are Poisson commutative, from \eqref{eq:def4}, it is equivalent to show that $ \left.\dfrac{d}{dt} \right\vert_{t = 0} p\left(\phi_t(\xi)\right) = 0$. Note that both $p$ and $f$ are $\mathfrak{a}$-invariant, which means that they are constant on each $A$-orbits. Now, if the flow $\phi_t$ of $\tilde{X}_f$ maps each $A$-orbit to itself, then the point $\xi$ along the trajectory of the flow stays within the same $A$-orbit. Hence, its value remains unchanged as time evolves. To conclude the statement, we only need to show that $ \phi_t$ is an orbit-preserving map. For any $X \in \mathfrak{a}$ and $s \in \mathbb{R}$, direct computation shows that \begin{align*}
    \phi_t \left( \mathrm{Ad}^*(\exp(sX)) \xi  \right) = & \, \mathrm{Ad}^*(\exp(tf)) \mathrm{Ad}^*(\exp(sX)) \xi  \\
    = & \,  \mathrm{Ad}^*(\exp(sX)) \mathrm{Ad}^*(\exp(tf))\xi \quad \text{ (since $ f$ is $\mathfrak{a}$-invariant)} \\
    = & \,\mathrm{Ad}^*(\exp(sX)) \phi_t(\xi) \in \mathcal{O}_{\phi_t(\xi)}  := \phi_t(\mathcal{O}_\xi).
\end{align*} Here $\mathcal{O}_\xi := \{\mathrm{Ad}^*(\exp(tf))\xi: \xi \in \mathfrak{g}^*\}$ is a $A$-orbit in $\mathfrak{g}^*$. Hence, $p$ is constant along the flow, and \eqref{eq:def4} vanishes. This concludes the proof.
\end{proof}

%\begin{corollary}
  %  Let $\mathcal{Q}_\mathfrak{g}(d)$ be the same as defined above. Then for any $\mathfrak{t} \neq \{0\}$, $\mathfrak{t} \subset \mathcal{Z}$.
%\end{corollary}
%To show \eqref{eq:def2} vanishes, it is equivalent to show that the time derivative of $p$ along the flow of $\tilde{X}$ is equal to zero.We will show the argument by a contradiction.  Without loss of generality, assume that there exists a $p \in \mathcal{Q}_k$ such that $ \{x,p\} \neq 0$. If $k = 1$, since $\mathfrak{t}$ is Abelian, this leads to a contradiction. We now consider the case where $ k \geq 2$. By the induction assumption, there exists a $n_0  \in \mathbb{N}$ such that $\{x,p\} = 0$ with $ \deg p \leq n_0$. Now, let $\deg p = n_0 + 1$. On the one hand, assume that $p$ is decomposable. Without loss of generality, let $p = p_a p_b$ with $\deg p_a + \deg p_b = n_0 + 1$. In another words, $\max \left\{\deg p_a, \deg p_b \right\} \leq n_0$.

%On the other hand, suppose that $p$ is indecomposable. Let $\textbf{Q}: = \{p,p_1,\ldots,p_r\} \subset \textbf{Q}_\zeta$ be an arbitrary subset of indecomposable generators $\mathcal{Q}_\mathfrak{g}(d)$.

%\begin{remark}
%  It is clear that if $\mathfrak{t}$ is not Abelian, then $\mathfrak{t}$ does not necessarily contain the central element of $\mathcal{Q}_\mathfrak{g}(d)$. For instance, see the degree one indecomposable generators in Subsection \ref{4.2}.
%\end{remark}

The polynomial algebra $\textbf{Alg} \langle\textbf{Q}_\zeta\rangle$ admits the following filtration\begin{equation*}
\mathcal{Q}_0 := \mathbb{F} \subset \textbf{Alg}\langle\textbf{Q}_1 \rangle:= \mathfrak{t} \subset \dots \subset \textbf{Alg}\langle\textbf{Q}_\zeta\rangle. 
\end{equation*} Also, we could set \begin{align*}
    \dim_{FL}\textbf{Alg}\langle\textbf{Q}_\zeta\rangle= l_1 + \cdots +  l_\zeta.
\end{align*} Here $\dim_{FL}$ denotes the number of indecomposable monomials that generate $\textbf{Alg}\langle\textbf{Q}_\zeta\rangle$.  Note that $ \dim_{FL}$ should be viewed as an upper bound for the rank of a finitely generated algebra since the generators of $\textbf{Alg}\langle\textbf{Q}_\zeta\rangle$ are not guaranteed to be algebraically independent. That is, $\dim_{FL} \mathcal{Q}_\mathfrak{g}(d) \geq \dim_{KL} \mathcal{Q}_\mathfrak{g}(d),$ where $\dim_{KL}$ is the Krull  dimension of the algebra (see e.g. \cite{MR1322960}).

 %  Note that the polynomials of the commutant can then be categorized as either part of an Abelian subalgebra or not. If $\textbf{Alg} \left\langle \textbf{Q}_\zeta \right\rangle$ is Cartan, the functionally independent polynomials from $\textbf{Q}_\zeta$ will give us an integrable system

Determining the explicit form of the polynomial relations in $\eqref{eq:comm}$ can be computationally challenging, even in the dual-space context, and additional constraints may be applied to limit the form of monomials that constitute the polynomials. The decomposition of polynomials to a certain degree implies that the Poisson-Lie bracket $\{\cdot,\cdot\}$ in $\eqref{eq:comm}$ produces higher-degree polynomials, eventually leading to a polynomial that decomposes into algebraically independent polynomials. Now, we introduce the terminology proposed in \cite{MR4660510} to describe the expansion of the Poisson brackets, encompassing all possible combinations of polynomials to various degrees.  For a degree $\zeta$ polynomial, the bilinear operation $\{\cdot,\cdot\}$ of these \textit{compact forms} up to degree $\zeta$ are denoted as follows: 
\begin{align}
\nonumber
    \{\textbf{q}_1,\textbf{q}_2\} \sim & \text{ } \textbf{q}_2 + \textbf{q}_1^2  \\
    \nonumber
    \{\textbf{q}_2,\textbf{q}_2\} \sim & \text{ } \textbf{q}_3 + \textbf{q}_1 \{\textbf{q}_1 ,\textbf{q}_2\} \\
      \{\textbf{q}_2,\textbf{q}_3\} \sim & \text{ } \textbf{q}_4 +  \textbf{q}_2^2 + \textbf{q}_1 \{\textbf{q}_2,\textbf{q}_2\} \label{eq:compact} \\
      \nonumber
       \{\textbf{q}_2,\textbf{q}_4\} \sim & \text{ } \textbf{q}_5 +  \textbf{q}_2 \textbf{q}_3 + \textbf{q}_1 \{\textbf{q}_2,\textbf{q}_3\}\\
       \nonumber
       \{\textbf{q}_2,\textbf{q}_5\} \sim & \text{ } \textbf{q}_6 + \textbf{q}_2 \textbf{q}_4 +  \textbf{q}_2^3 +  \textbf{q}_3^2 + \textbf{q}_1 \{\textbf{q}_2,\textbf{q}_4\} \\
       & \vdots  \nonumber
\end{align} For any $\textbf{q}_k,\textbf{q}_l$ with $1 \leq k,l \leq \zeta$ and $k+l \geq 2$,  the Poisson bracket in terms of the compact term is given as follows: \begin{align}
    \{\textbf{q}_k,\textbf{q}_l\} \sim \{\textbf{q}_g,\textbf{q}_h\}\sim &\, \textbf{q}_{k+l-1} + \textbf{q}_2 \textbf{q}_{k+l-3} + \textbf{q}_3 \textbf{q}_{k+l-4} + \ldots + \prod_{\ 
         j_1  + \ldots + j_\zeta = k+l-1
   } \textbf{q}_{j_1} \cdots \textbf{q}_{j_\zeta} + \ldots + \textbf{q}_1 \{\textbf{q}_k,\textbf{q}_{l-1}\} \nonumber\\
    \sim &\, \sum_{l_1 a_1 + l_2 a_2 + \ldots + l_\zeta a_\zeta = k + l } \textbf{q}_1^{a_1} \textbf{q}_2^{a_2} \cdots \textbf{q}_\zeta^{a_\zeta}. \label{eq:cp}
\end{align} Here $\textbf{q}_k^{a_k} = \left(p_1^{(k)}\right)^{w_1} \cdots \left(p_{l_k}^{(k)}\right)^{w_{l_k}} $ with $w_1 + \ldots + w_{l_k} = a_k $ and $k + l = g + h \leq \zeta$. Notably, in the scope of Equation \eqref{eq:cp}, Poisson brackets of identical degree are indistinguishable in terms of their compact forms. To illustrate, consider the scenario where $\deg \{\textbf{q}_2, \textbf{q}_2\}= \deg \{\textbf{q}_1, \textbf{q}_3\}$. Under such circumstances, we can ascertain that the expansion in $\{\textbf{q}_2, \textbf{q}_2\}$ is equivalent to that of $\{\textbf{q}_1, \textbf{q}_3\}$, which further translates to an expansion resembling $\textbf{q}_3 + \textbf{q}_1 \{\textbf{q}_1, \textbf{q}_2\}$. Note that $\textbf{q}_t = \{0\}$ for any $t \leq 0,$ and the number of the allowed polynomials in $\textbf{q}_k^{a_k}$ is $\binom{m_k + a_k - 1 }{a_k}$ for all $1 \leq k  \leq \zeta$. Clearly, depending on the embedding chain of Lie algebras, some of the coefficients appearing above will be zero. However, without further insight into the structure of the polynomial generators, there is a large number of monomials in the polynomials $p_k^{(l)}$ at a given degree. This means that finding the coefficient $\Gamma^{s_1,...,s_r}_{uv}$ and determining the allowed monomials in $\textbf{q}_1^{a_1} \textbf{q}_2^{a_2} \cdots \textbf{q}_\zeta^{a_\zeta}$ are difficult tasks. In the following section, we will introduce a method which allows to simplify the number of terms appearing in the Poisson brackets, i.e., in polynomial expansions of a given degree. %This permits us to construct the corresponding polynomial algebra.

\section{The grading of polynomials in $\mathcal{Q}_\mathfrak{g}(d)$}

\label{3}

In Section $\ref{2}$, we presented a comprehensive framework for the development of polynomial algebras through a centralizer subalgebra. In this section, we shall elaborate on the specific terminology and concepts that enable us to anticipate and subsequently reduce the permissible set of polynomials appearing in non-trivial Poisson brackets. This involves a meticulous analysis of the structural properties inherent in these algebraic constructions. We also emphasize that the indices presented here differ from those in Section \ref{2}, and it is crucial to avoid confusing the two sets.  In the following,  assume that $\mathfrak{g} = \bigoplus_{r \in J} \mathfrak{g}_r $ with $[\mathfrak{g}_r,\mathfrak{g}_s] \subset \mathfrak{g}_t$,  ($r,s,t \in J$), where $J \subset \mathbb{N}$ is a finite index set, and $\mathfrak{g}_1, \mathfrak{g}_2$ are subalgebras of $\mathfrak{g}$. Then $S(\mathfrak{g}) \cong \bigotimes_{r \in J} S(\mathfrak{g}_r)$. Recall that    $\mathcal{Q}_\mathfrak{g}(d) = \mathfrak{t}  \oplus \bigoplus_{k \in \Omega}  \mathcal{Q}_k$ is a graded polynomial algebra with a Poisson bracket $\{\cdot,\cdot\}$.  For any homogeneous polynomial $p^{(k)} \in \mathcal{Q}_k$, the construction involves using the degree for grading and identifying monomials of the form \begin{align*}
 p^{(i_1 + \ldots+i_m)} =  x_1^{a_1}\cdots x_{s_1}^{a_{s_1}}\cdots  \underbrace{x_{{s_{w-1}}+1}^{a_{s_{w-1}+1}}\cdots x_{s_w}^{a_{s_w}}}_{\text{elements in $\mathfrak{g}_w$}}  \cdots  x_{s_{m-1}+1}^{a_{s_{m-1}+1}} \cdots x_{s_m}^{a_{s_m}} \in \mathcal{Q}_{i_1 + \ldots + i_m}.
\end{align*} Here $i_1 + \ldots + i_m = k $, $x_{{s_{w-1}}+1}^{a_{s_{w-1}}+1}\cdots x_{s_w}^{a_{s_w}} $ with $a_{s_{w-1}+1} + \ldots + a_{s_w} = i_w$ is a monomial in $S^{i_w}(\mathfrak{g}) $ for all $ 1\leq w \leq m$. This induces the following definition.

\begin{definition}
\label{grading}
  A $\textit{grading of a monomial}$ $\mathcal{G}$ assigns to each monomial $p \in \mathcal{Q}_{i_1+\ldots+i_m}/\{0\}$ an ordered tuple $(i_1, \ldots, i_m) \in \mathbb{N}_0^m$ where each $i_j$ is an element of $\mathbb{N}_0$, symbolizing the quantity of times that the generator $\mathfrak{g}_j$ appears in the monomial.  
\end{definition}
\begin{remark}
\label{re2.4}
(i)  It is clear that the grading $\mathcal{G}$ defined above is not a function, as the grading of a monomial may correspond to many distinct polynomials in $\mathcal{Q}_{i_1+\ldots+i_m}/\{0\}$. We will see lots of examples in Section \ref{4} to illustrate this point.

 (ii)   Observe that the grading of a monomial can be established in any polynomial algebra. For example, let $\mathcal{P} \subset S(\mathfrak{g}) $ be a graded non-Abelian polynomial algebra that admits a PBW basis. That is, $\mathcal{P} = \bigoplus_{k \in I} \mathcal{P}_k$ with an ordered basis generated by the power of the generators of $\mathfrak{g}$, where $\mathcal{P}_k$ consists of all degree $k$ polynomials. Here $I \subset \mathbb{N}$ is a non-empty ordered set of indices. The grading of a monomial is described by $\mathcal{G}: \mathcal{P}_k/\{0\} \rightarrow \mathbb{N}_0 \times \cdots \times \mathbb{N}_0$.
 
  (iii) Observe that if $p$ remains constant, we define $\mathcal{G}(p)$ as $(0, \ldots, 0)$. However, based on our construction of polynomial algebras in Section $\ref{2}$, we exclude $\mathcal{Q}_0 = \mathbb{F}$. Consequently, we will not consider the scenario where $\mathcal{G}(p) = (0, \ldots, 0)$.
\end{remark}

%{\color{blue} All the polynomial with the same grading will form a set!! not a space!}

%Notice that the definition above also holds under the quantization.

We then have a look at some properties of $\mathcal{G}$. Starting with the grading of two monomials.

\begin{lemma}
\label{mul}
  For any non-zero monomials $p ,q \in \mathcal{Q}_k,$ then $\mathcal{G}(pq) = \mathcal{G}(p) + \mathcal{G}(q)$.
\end{lemma}

\begin{proof}
Without loss of generality, assume that $\mathcal{G}(p) = (i_1,\ldots,i_m)$ and $\mathcal{G}(q) = (j_1,\ldots,j_m)$. Here $i_1 + \ldots + i_m = j_1 + \ldots + j_m = k$. A direct computation shows that $$\mathcal{G}(pq) = (i_1 +j_1,\ldots,i_m + j_m) = (i_1,\ldots,i_m) + (j_1,\ldots,j_m) = \mathcal{G}(p) + \mathcal{G}(q).$$
\end{proof}

\begin{remark}
\label{3.4}
(i) Note that if a monomial $p \in \mathcal{Q}_\mathfrak{g}(d)$ is decomposable, then, by definition, $p = \prod_{l=1}^k p_l$ where each $p_l$ is an indecomposable monomial in $\mathcal{Q}_\mathfrak{g}(d)$. It is possible that there exists an indecomposable monomial $p' \in \mathcal{Q}_\mathfrak{g}(d)$ such that $\mathcal{G}(p) = \sum_{l=1}^k \mathcal{G}(p_l) = \mathcal{G}(p')$. Therefore, to avoid confusion, we will represent the grading of all decomposable polynomials as a sum. That is, $\mathcal{G}(p) = \mathcal{G}(p_1) + \ldots + \mathcal{G}(p_k)$.

 (ii) Consider an indecomposable polynomial $p$ expressed as $p =   c_1 p_1 + \ldots + c_k p_k$, where $c_l \in \mathbb{F}$ for each $l \in \{1,\ldots,k\}$. We define the operation $\tilde{+}$ such that $\mathcal{G}(p) = \mathcal{G}(p_1) \tilde{+} \ldots \tilde{+} \mathcal{G}(p_k)$. Obviously, the operation $\tilde{+}$ is both associative and commutative. We shall refer to such $\mathcal{G}(p)$ as an in-homogeneous grading, whereas the grading within each $p_l$ is homogeneous. Notably, if there are indices $s \neq l \in \{1,\ldots,k\}$ such that $\mathcal{G}(p_s) = \mathcal{G}(p_l)$, then $$\mathcal{G}(p) = \underbrace{\mathcal{G}(p_1) \tilde{+} \ldots  \tilde{+} \mathcal{G}(p_s) \tilde{+} \ldots  \tilde{+} \mathcal{G}(p_{l-1}) \tilde{+} \mathcal{G}(p_{l+1}) \tilde{+} \ldots \tilde{+} \mathcal{G}(p_k)}_{k-1\text{-terms}}.$$ This imposes limitations on the permitted generators in the Poisson bracket.% to ensure closure.
\end{remark}

%the construction involves using the degree for grading and identifying monomials of the form,
%\begin{equation}
%p_{r_1,...,r_n}^{h}= x_1^{r_1} \cdots x_n^{r_n}, \quad \text{ } r_1+\ldots + r_n = h \in \mathbb{N}_0
%\end{equation}Then, polynomials can be grouped according to their total degree $h $.

We now provide two examples demonstrating the process of computing the grading of polynomial algebra generators. These examples illustrate how the decomposition of a Lie algebra influences the grading of the generators. % Building on these examples, Section $\ref{4}$ will explore more detailed cases.

\subsection{Lie algebra with decomposition into two subalgebras}

\label{3.1}

Now, we will present some calculations on the grading of a monomial in the symmetric algebra $S(\mathfrak{g})$ induced from a Lie algebra that decomposes into two subalgebras. Let $\mathfrak{g}_1$ and $\mathfrak{g}_2$ be subalgebras in $\mathfrak{g}$ such that $\mathfrak{g}_1 \cap \mathfrak{g}_2 = \{0\}$. In this way, $\mathfrak{g}$ has the vector space decomposition given by $\mathfrak{g} = \mathfrak{g}_1 \oplus \mathfrak{g}_2,$ which satisfies the following commutator relations:
\begin{equation}
[\mathfrak{g}_1,\mathfrak{g}_1] \subset \mathfrak{g}_1 , \quad[ \mathfrak{g}_2, \mathfrak{g}_2 ] \subset \mathfrak{g}_1,  \quad [\mathfrak{g}_1, \mathfrak{g}_2 ] \subset \mathfrak{g}_2.  \label{eq:Z2commu}
\end{equation}  In the following, $\oplus$ denotes a vector space direct sum.  Assume that $\mathfrak{g}_1 = \mathrm{span} \{x_1,\ldots,x_u\}$ and $\mathfrak{g}_2 = \mathrm{span} \{x_{u+1},\ldots,x_{u+v}\}$ such that $\mathfrak{g} = \mathfrak{g}_1 \oplus\mathfrak{g}_2$. Then $S(\mathfrak{g}) \cong S(\mathfrak{g}_1) \otimes S(\mathfrak{g}_2) $ is a vector space isomorphism. We further assume that a polynomial with degree $i_1+i_2$ in $\mathcal{Q}_\mathfrak{g}(d)$ has the form of \begin{align*}
    p^{(i_1+i_2)}=   \underbrace{x_1^{a_1} \cdots x_u^{a_u}}_{\text{elements in $\mathfrak{g}_1$}} \underbrace{x_{u+1}^{a_{u+1}} \cdots x_{u+v}^{a_{u+v}}}_{\text{elements in $\mathfrak{g}_2$}} , \quad \begin{matrix}
           a_1 + \ldots + a_u = i_1, \\
           a_{u+1} + \ldots + a_{u+v} = i_2
      \end{matrix}
  \end{align*} with  $  a_1,\ldots,a_{u+v} \in \mathbb{N}_0 $. From Definition $\ref{grading}$, for any non-zero $p \in \mathcal{Q}_{i_1+i_2}$, $\mathcal{G}(p) = (i_1,i_2)$ where $i_1$ and $i_2$ are the numbers of elements that belong to $\mathfrak{g}_1$ and $\mathfrak{g}_2$, respectively.   %It is important to note that the symbols $i$ and $j$ will be assigned distinct interpretations when discussed in Subsection $\ref{4.3}$ and in Appendix \ref{app} compared to their usage in Section $\ref{5}$. Although these variables are represented by the same letters, their contexts and roles will differ significantly between these two sections.

  \begin{lemma}
  \label{3.6}
      Let $\mathfrak{g}$ be a Lie algebra with an ordered basis $\beta_{\mathfrak{g}}$ satisfying the commutator relations in $\eqref{eq:Z2commu}$, and let $\mathcal{Q}_{\mathfrak{g}}(d) \subset S(\mathfrak{g})^{\mathfrak{g}_1}$. For any non-zero indecomposable monomial $p,q \in \mathcal{Q}_\mathfrak{g}(d)$, suppose that $\mathcal{G}(p) = (i_1,i_2)$  and $\mathcal{G}(q) = (i_1',i_2')$.   Then a (non-vanishing) Poisson bracket will have the effect of
\begin{align}
     \mathcal{G}\left(\{p,q\}\right)= (i_1+i_1'-1,i_2+i_2' ) \tilde{+} (i_1+i_1'+1,i_2+i_2'-2) .
\end{align}
  \end{lemma}

  \begin{proof}
  Without loss of generality, assume that $p = A_1A_2, q = B_1B_2 $, where\begin{align*}
     &  A_1 = x_1^{a_1} \cdots x_u^{a_u},\text{ }    \text{ } B_1 = x_1^{a_1'} \cdots x_u^{a_u'} ,  \\
      & A_2 = x_{u+1}^{a_{u+1}} \cdots x_{u+v}^{a_{u+v}}, \text{ } B_2 = x_{u+1}^{a_{u+1}'} \cdots x_{u+v}^{a_{u+v}'}.
  \end{align*} Then using the Leibniz rule
  \begin{align*}
       \left\{A_1A_2,B_1B_2  \right\} =  \{A_1,B_1B_2\}A_2 + A_1\{A_2,B_1B_2\}.
     % =\, &  \left\{x_1^{a_1} \cdots x_u^{a_u} , x_{u+1}^{a_{u+1}'} \cdots x_{u+v}^{a_{u+v}'}  \right\}x_1^{a_1'} \cdots x_u^{a_u'}x_{u+1}^{a_{u+1}} \cdots x_{u+v}^{a_{u+v}} + \left\{x_1^{a_1} \cdots x_u^{a_u},x_1^{a_1'} \cdots x_u^{a_u'}   \right\}  x_{u+1}^{a_{u+1}} \cdots x_{u+v}^{a_{u+v}}x_{u+1}^{a_{u+1}'} \cdots x_{u+v}^{a_{u+v}'} \\
    %  + & \left\{ x_{u+1}^{a_{u+1}} \cdots x_{u+v}^{a_{u+v}},x_1^{a_1'} \cdots x_u^{a_u'}   \right\}x_1^{a_1} \cdots x_u^{a_u}x_{u+1}^{a_{u+1}'} \cdots x_{u+v}^{a_{u+v}'} + \left\{x_{u+1}^{a_{u+1}} \cdots x_{u+v}^{a_{u+v}}, x_{u+1}^{a_{u+1}'} \cdots x_{u+v}^{a_{u+v}'}  \right\}x_1^{a_1} \cdots x_u^{a_u}x_1^{a_1'} \cdots x_u^{a_u'} .
  \end{align*}
  Observe that, by repeatedly applying the Leibniz rule and the definition of the centralizer subalgebra,
  \begin{align*}
      \{A_1,B_1B_2\} = \sum_{m=1}^u x_1^{a_1} \cdots x_{m-1}^{a_{m-1}} \{x_m,B_1B_2\} x_{m+1}^{a_{m+1}} \cdots x_u^{a_u} =0.
  \end{align*} Hence $\left\{A_1A_2,B_1B_2  \right\} =   A_1\{A_2,B_1B_2\}$. On the other hand, the second term has the following expansion $ \left\{A_2,B_1B_2  \right\} =   \{A_2,B_1\} B_2 + B_1 \{A_2,B_2\} $. In particular,   \begin{align*}
      \{A_2,B_1\} =    \sum_m \sum_k x_1^{a_1'}\cdots x_{k-1}^{a_{k-1}'} \left\{x_m^{a_m},x_k^{a_k'}\right\} x_{k+1}^{a_{k+1}'} \cdots x_u^{a_u'} x_{u+1}^{a_{u+1}} \cdots x_{m-1}^{a_{m-1}'}x_{m+1}^{a_{m+1}} \cdots x_{u+v}^{a_{u+v}},
  \end{align*}
  where
  \begin{align*}
      \left\{x_m^{a_m},x_k^{a_k'}\right\} = \sum_{r=1}^{a_k'-1} \sum_{s=1}^{a_m-1} x_k^r x_m^s \{x_m,x_k\}x_k^{a_k'-1-r} x_m^{a_m-1-s}.
  \end{align*}

 We now compute the grading for $\{A_1A_2,B_1B_2\}$. The gradings for these terms are
 \begin{align}
      \mathcal{G}\left(A_1\right) = (i_1,0), \quad    \mathcal{G}\left(A_2\right) = (0,i_2),\quad \mathcal{G}\left(B_1\right) = (i_1',0),\quad \mathcal{G}\left(B_2\right) = (0,i_2').
  \end{align}
  Hence,  using the commutator relation $\eqref{eq:Z2commu}$, we deduce that
  \begin{align*}
      \mathcal{G}\left( \left\{A_2,B_1  \right\}\right) = \, & \left(a_1' + \ldots + a_{k-1}' + (a_k'-1) + a_{k+1}' + \ldots + a_u'\right)  \\
      & + \left(a_{u+1} + \ldots + a_{m+1} + (a_m -1) + \ldots + a_{u+v} \right) +1\\
      = \,& \left(i_1-1 , (i_2'-1) +1\right) =\left( i_1  -1, i_2' \right) .
  \end{align*}
   It follows that
  \begin{align*}
      \mathcal{G} \left(A_1\{A_2,B_1 \}B_2\right) = (0,i_2') +  \mathcal{G} \left(\left\{A_2,B_1  \right\}\right) = (i_1,i_2') + (i_1'-1, i_2) = (i_1+i_1'-1,i_2+i_2').
  \end{align*}
 Similarly to the previous cases, the grading of the rest of the cases is given by \begin{align*}
%      \mathcal{G} \left([A_1,B_1] B_2A_2\right) = & \, (0,i_2+i_2') + \mathcal{G}\left([A_1,B_1]\right) = (i_1+i_1'-1,i_2+i_2')\\
  %  \mathcal{G} \left(A_1[A_2,B_1] B_2 \right) =& \, (i,i_2') + \mathcal{G}\left([A_2,B_1]\right) = (i+ i_1'-1,i_2+i_2' )\\
   \mathcal{G} \left(A_1B_1\{A_2,B_2\} \right) = & \,  (i_1+i_1',0)+ \mathcal{G}\left(\{A_2,B_2\}\right) = (i_1+i_1'+1,i_2+i_2'-2) .
  \end{align*}  In conclusion, using Remark $\ref{3.4}$ (ii), we deduce that $\mathcal{G} \left(\{A_1A_2,B_1B_2\}\right) = (i_1 +i_1'-1,i_2+i_2')    \tilde{+} (i_1+i_1'+1,i_2+i_2'-2) $ as required.
  \end{proof}

 % \begin{remark}
%   We can further show that for any non-zero monomials $p_1,p_2$ in $S(\mathfrak{g}),$ the following relation still holds  \begin{align*}
 %     \mathcal{G}\left( \left\{p_1,p_2  \right\}\right) =      (i_1+i_1'-1,i_2+i_2')\tilde{+} (i_1+i_1'+1,i_2+i_2'-2).
%  \end{align*}

 % (i) In the commute setting, for any indecomposable monomials $P_1,P_2 \in \tilde{\mathcal{Q}}_\mathfrak{g}(d),$ we have

%\[ \mathcal{G}  \left([ P_1,P_2]\right)= (i_1+i_1'-1,i_2+i_2' ) \tilde{+} (i_1+i_1'+1,i_2+i_2'-2). \]

%  (ii)
%\end{remark}
\begin{corollary}
\label{3.7}
    Let $\mathcal{Q}_\mathfrak{g}(d) \subset S(\mathfrak{g})^{\mathfrak{g}_1}$. For any non-zero indecomposable monomials $p,q \in \mathcal{Q}_\mathfrak{g}(d)$, assume that $\mathcal{G}(p) = (i_1,0)$ and $\mathcal{G}(q) = (i_1',i_2')$ with $i_1',i_2',i_1  \neq 0 $ or $\mathcal{G}(p) = (i_1,i_2)$ and $\mathcal{G}(q) =(i_1',0)$. Then $\left\{p,q\right\} =0 $.
\end{corollary}

In line with Lemma $\ref{3.6}$, once the grading of the Poisson brackets has been set, the particular structures of the monomials on the right-hand side of the non-trivial brackets must encompass all related monomials with matching grading. This approach will facilitate the determination of the total count of all possible terms. For instance, replace $p$ by $B_{i_1i_2} $ and $q$ by $B_{i_1'i_2'}$. Based on the grading of $\{p,q\}$, using Lemma $\ref{3.6},$ the allowed terms that match the grading of $\{p,q\}$ are given by

\[  \{B_{i_1i_2},B_{i_1'i_2'}\}=  a_1 B_{i_1+i_1'-1,i_2+i_2'} + b_1 B_{i_1+i_1'+1,i_2+i_2'-2} \]
\[ +  \sum_{k+m=i_1+i_1'-1,l+n=i_2+i_2'} a_{klmn}  B_{kl} B_{mn}  +  \sum_{k+m=i_1+i_1'+1,l+n=i_2+i_2'-2} b_{klmn}  B_{kl} B_{mn}  \]
\[ +   \sum_{s+k+m=i_1+i_1'-1,l+n+t=i_2+i_2'} a_{klmnst}  B_{kl} B_{mn} B_{st}  +  \sum_{s+k+m=i_1+i_1'+1,l+n+t=i_2+i_2'-2} b_{klmnst}  B_{kl} B_{mn} B_{st} + .... \]
Here $a_1,b_1,\ldots,b_{klmnst}$ are arbitrary coefficients that can be determined by explicit Poisson relations of $\mathfrak{g}^*$.

%and similarly in commutator setting
%\[  [B_{i_1i_2},B_{i_1'i_2'}] =  a_1 B_{i_1+i_1'-1,i_2+i_2'+0} + b_1 B_{i_1+i_1'+1,i_2+i_2'-2} \]
%\[ +  \sum_{k+m=i_1+i_1'-1,l+n=i_2+i_2'+0} a_{klmn}  B_{kl} B_{mn}  +  \sum_{k+m=i_1+i_1'+1,l+n=i_2+i_2'-2} b_{klmn}  B_{kl} B_{mn}  \]
%\[ +   \sum_{j+k+m=i_1+i_1'-1,l+n+o=i_2+i_2'+0} a_{klmn}  B_{kl} B_{mn} B_{jo}  +  \sum_{j+k+m=i_1+i_1'+1,l+n+o=i_2+i_2'-2} b_{klmn}  B_{kl} B_{mn} B_{jo} + .... \]

%As all polynomials have the same grading in the sense of calculating with the commutators, i.e. monomials and polynomials are formed by a monomial with total eigenvalue 0 relative to $l_0$, this is not useful to close the Poisson bracket. However, polynomials are also homogeneous and we will look at the case where they are formed by a monomial with the same structure in terms of element of $\mathfrak{g}_1$ and $\mathfrak{g}_2$.

\subsection{Lie algebras with a decomposition into three subalgebras}

Consider now that $\mathfrak{g}$ has more than two layers. Let $\mathfrak{g}_1,\mathfrak{g}_2,\mathfrak{g}_3$ be the subalgebras of $\mathfrak{g}$ such that $\mathfrak{g} = \mathfrak{g}_1 \oplus \mathfrak{g}_2 \oplus \mathfrak{g}_3 $ with the following commutator relations  \begin{align}
    [\mathfrak{g}_1,\mathfrak{g}_1] = & \,\{0\}, \text{ }
    [\mathfrak{g}_1,\mathfrak{g}_2] \subset \mathfrak{g}_2, \text{ }
    [\mathfrak{g}_1,\mathfrak{g}_3]  \subset  \mathfrak{g}_3, \label{eq:commutator3}\\
    \nonumber
     [\mathfrak{g}_2,\mathfrak{g}_2]  \subset  & \, \mathfrak{g}_2, \text{ }
    [\mathfrak{g}_2,\mathfrak{g}_3]  \subset  \mathfrak{g}_1, \text{ }
    [\mathfrak{g}_3,\mathfrak{g}_3]  \subset  \mathfrak{g}_3.
\end{align}
Then $S(\mathfrak{g}) \cong S(\mathfrak{g}_1) \otimes S(\mathfrak{g}_2) \otimes S(\mathfrak{g}_3)$.  As an example, consider that $\mathfrak{g}$ is a complex semisimple Lie algebra with a triangular decomposition. It is clear that $\mathfrak{g}$ admits the commutator relations in $\eqref{eq:commutator3}$.  Without loss of generality, assume that $\mathfrak{g}_1 = \mathrm{span} \{x_1,\ldots,x_u\}$, $\mathfrak{g}_2 = \mathrm{span} \{x_{u+1},\ldots,x_{u+v}\}$, and $ \mathfrak{g}_3  = \mathrm{span} \{x_{u+v+1},\ldots,x_{u+v+w}\} $. Let $\mathfrak{a}$ be a subalgebra of $\mathfrak{g}$ such that $\mathcal{Q}_\mathfrak{g}(d)$ is a polynomial algebra with respect to the subalgebra $\mathfrak{a}$. We further assume that a monomial with degree $i_1+i_2 +i_3$ in $\mathcal{Q}_{i_1+i_2 +i_3} \subset \mathcal{Q}_\mathfrak{g}(d)$ has the form of \begin{align*}
    p^{(i_1+i_2 +i_3)}=   \underbrace{x_1^{a_1} \cdots x_u^{a_u}}_{\text{elements in $\mathfrak{g}_1$}} \underbrace{x_{u+1}^{a_{u+1}} \cdots x_{u+v}^{a_{u+v}}}_{\text{elements in $\mathfrak{g}_2$}}\underbrace{x_{u+v+1}^{a_{u+v+1}} \cdots x_{u+v+w}^{a_{u+v+w}}}_{  \text{elements in $\mathfrak{g}_3$}} , \quad \begin{matrix}
           a_1 + \ldots + a_u = i_1 \\
           a_{u+1} + \ldots + a_{u+v} = i_2, \\
           a_{u+v+1} + \ldots + a_{u+v+w} = i_3
      \end{matrix}
  \end{align*} $ \text{ with } a_1,\ldots,a_{u+v+w} \in \mathbb{N}_0 $. By definition, for any non-zero $p \in \mathcal{Q}_{i_1+i_2 +i_3}$, we may write that $\mathcal{G}(p) = (i_1,i_2 ,i_3)$.

\begin{lemma}
\label{3.9}
  Let $\mathfrak{g}^*$ be its dual admitting the same relations as in $\eqref{eq:commutator3}$ in a Poisson-Lie bracket $\{\cdot,\cdot\}$. For any non-zero indecomposable monomials $p \in \mathcal{Q}_{i_1+i_2+i_3}$ and $q \in \mathcal{Q}_{i_1'+i_2'+i_3'}$, the following holds:

(i) Let $\mathcal{Q}_\mathfrak{g}(d_1) \subset S(\mathfrak{g})^{\mathfrak{g}_2}$. Then  \begin{align*}
    \mathcal{G} \left(\{p,q\}\right) = (i_1+i_1'-1,i_2+i_2',i_3+i_3') \tilde{+}(i_1+i_1'+1,i_2+i_2'-1,i_3+i_3'-1) \tilde{+}(i_1+i_1',i_2+i_2',i_3+i_3'-1)
\end{align*}

(ii) Let $\mathcal{Q}_\mathfrak{g}(d_2) \subset S(\mathfrak{g})^{\mathfrak{g}_1}$. Then \begin{align*}
    \mathcal{G} \left(\{p,q\}\right) = & \, (i_1+i_1'-1,i_2+i_2',i_3+i_3') \tilde{+}(i_1+i_1'+1,i_2+i_2'-1,i_3+i_3'-1) \tilde{+}(i_1+i_1',i_2+i_2',i_3+i_3'-1) \\
   & \, \tilde{+}(i_1+i_1',i_2+i_2'-1,i_3+i_3') .
\end{align*}
\end{lemma}
\begin{proof}
   For any non-zero monomials $p  \in \mathcal{Q}_{i_1+i_2+i_3}$ and $q \in \mathcal{Q}_{i_1'+i_2'+i_3'},$ without loss of generality, we may write that $p = A_1A_2A_3$ and $q = B_1 B_2 B_3,$ where \begin{align*}
         &  A_1 = x_1^{a_1} \cdots x_u^{a_u},    \text{ } A_2 = x_{u+1}^{a_{u+1}} \cdots x_{u+v}^{a_{u+v}}, \text{ }  A_3 =  x_{u+v+1}^{a_{u+v+1}} \cdots x_{u+v+w}^{a_{u+v+w}} ; \\
      & B_1 = x_1^{a_1'} \cdots x_u^{a_u'} ,\text{ } B_2 = x_{u+1}^{a_{u+1}'} \cdots x_{u+v}^{a_{u+v}'} , \text{ }  B_3 =  x_{u+v+1}^{a_{u+v+1}'} \cdots x_{u+v+w}^{a_{u+v+w}'}.
    \end{align*} Here \begin{align*}
       & a_1 + \ldots + a_u = i_1, \quad  a_{u+1} + \ldots + a_{u+v} = i_2 , \quad  a_{u+v+1} + \ldots + a_{u+v+w} = i_3, \\
       & a_1' + \ldots + a_u' = i_1', \quad  a_{u+1}' + \ldots + a_{u+v}' = i_2' ,\quad a_{u+v+1}' + \ldots + a_{u+v+w}' = i_3'.
    \end{align*}  A direct computation shows that \begin{align}
     \{p,q\} = \{A_1,q\}A_2A_3 + A_1\{A_2,q\}A_3 + A_1 A_2 \{A_3,q\}. \label{eq:Leibniz}
 \end{align}

We will only show part (i) as the similar argument holds for the second part. By Leibniz's rule, the equation $\eqref{eq:Leibniz}$ becomes \begin{align}
 \nonumber
     \{p,q\}  % & \, \{A_1,q\}A_2A_3 +   A_1 A_2 \{A_3,q\} \\
      = & \, \{A_1,B_1\} B_2B_3 A_2 A_3 + \{A_1,B_2\} B_1B_3A_2 A_3 + \{A_1,B_3\} B_1B_2 A_2 A_3   \\
      & \, +  A_1 A_2 \{A_3,B_1\}B_2 B_3 +  A_1 A_2 \{A_3,B_2\}B_1 B_3 +  A_1 A_2 \{A_3,B_3\}B_2 B_1 .  \label{eq:grading3}
 \end{align}   By definition, the grading of $\{p,q\}$ in $\eqref{eq:grading3}$ is equal to the grading of each of the components. From the commutator relations $\eqref{eq:commutator3}$, $\{A_1,B_1\} = 0$. We can then discard this term in $\eqref{eq:grading3}$. For the rest of the components in $\eqref{eq:grading3}$, we will compute them case by case. Starting from the term $ \{A_1,B_2\} B_1B_3A_2 A_3$, a direct computation shows that \begin{align*}
     \mathcal{G} \left(\{A_1,B_2\} B_1B_3 A_2 A_3\right) = & \, \mathcal{G} \left(\{A_1,B_2\}\right) + \mathcal{G} (B_1) + \mathcal{G} (B_3) + \mathcal{G} (A_2) +\mathcal{G} (A_3)  \\
     = & \, (i_1-1,i_2',0) + (i_1', 0,0 ) + (0,0,i_3')+ (0,i_2,0)+ (0,0,i_3 ) \\
     = & \, (i_1 +i_1'-1,i_2' +i_2,i_3'+i_3) .
 \end{align*} Similarly, we have \begin{align*}
     \mathcal{G} \left(\{A_1,B_3\} B_1B_2 A_2 A_3 \right) = & \,   (i_1 +i_1'-1,i_2' +i_2,i_3'+i_3)  =   \mathcal{G} \left(A_1 A_2 \{A_3,B_1\}B_2 B_3\right) ; \\
      \mathcal{G} \left(A_1 A_2 \{A_3,B_2\}B_1 B_3 \right) = & \, (1,i_2'-1,i_3'-1) + (i_1'+i_1', 0,0 ) + (0,i_2 ,i_3')   = (i_1+i_1'+1,i_2+i_2'-1,i_3+i_3'-1);\\
       \mathcal{G} \left(A_1 A_2 \{A_3,B_3\}B_2 B_1\right) = & \, (0,0,i_3'+i_3-1) + (i_1', i_2' ,0 )   + ( i_1  ,i_2,0 ) =  (i_1+i_1' ,i_2+i_2' ,i_3+i_3'-1).
 \end{align*} Summing all the terms together, we deduce that $\mathcal{G}\left(\{p,q\}\right)=  (i_1+i_1'-1,i_2+i_2',i_3+i_3') \tilde{+}(i_1+i_1'+1,i_2+i_2'-1,i_3+i_3'-1) \tilde{+}(i_1+i_1',i_2+i_2',i_3+i_3'-1)$,  as required.
\end{proof}

From the construction above, we can see that the grading of the polynomials in the Poisson bracket is heavily based on the commutator relations and the decomposition of Lie algebras. We now propose a generic way on finding the grading in the non-trivial bracket relations of a polynomial algebra. Assume that $\mathfrak{g} = \mathfrak{g}_1 \oplus \ldots \oplus \mathfrak{g}_m$ with $[\mathfrak{g}_s,\mathfrak{g}_r] \subset \mathfrak{g}_t$, where $1 \leq s,r,t\leq m $ with $m >3$. Let $p,q \in \mathcal{Q}_\mathfrak{g}(d)$ be the generators with $\mathcal{G}(p) = (i_1,\ldots,i_m)$ and $\mathcal{G}(q) = (i_1',\ldots,i_m') $. In this context, $i_w$ represents the count of generators of $\mathfrak{g}_w$ within $p,$ and $i_w'$ follows the same definition. By direct calculation, it turns out that there exists a sequence $(a_{1k},\ldots,a_{mk})$ with $1\leq k \leq \xi$ and $a_{1k},\ldots,a_{mk} \in \{-2,-1,0,1,2\}$ such that \begin{align}
\mathcal{G} \left(\{p,q\}\right) = \underbrace{\left(i_1 + i_1' + a_{11}, \ldots,i_m + i_m' + a_{mk}\right) \tilde{+} \ldots \tilde{+}\left(i_1 + i_1' + a_{1 \xi}, \ldots,i_m + i_m' + a_{m \xi}\right)}_{\text{contains $k$-terms of grading }}. \label{eq:gradinglong}
\end{align} Here $\xi < m$ is a finite integer. To this extent, using the grading in \eqref{eq:gradinglong}, we are able to reduce some terms in the compact forms given in \eqref{eq:compact}. Given the varying commutator relations from different Lie algebras, the specific grading will also vary depending on the particular algebra. In the following section, we show how the grading \eqref{eq:gradinglong} simplifies the components into a reduced compact form in different examples.

%By construction, the Poisson bracket $\{\textbf{q}_i,\textbf{q}_j\} \sim \sum_{a_1 + 2 a_2 + \ldots + n a_n = i+j} \textbf{q}_1^{a_1} \cdots \textbf{Q}_{n+1}^{a_n}$ will provide us with the summation of the compact forms. Here \begin{align*}
%    \textbf{q}_1^{a_1} \cdots \textbf{Q}_{n+1}^{a_n} = \sum
%\end{align*} From \cite{campoamor2023algebraic}, we know that $|\textbf{q}_j| = \binom{n}{j}(n-j)!$

%Recall that the polynomial algebra $\mathcal{Q}_\mathfrak{g}(d)$ is generated by a finite set $\textbf{Q}_\zeta$ (defined in Section $\ref{2}.$) Once we determined the grading of each

%Notice that, by Corollary $\ref{3.7}$, $\{(i,0),(i_1',i_2')\} = (0,0).$ Moreover, if $0 \neq p = \gamma_1 p_1 + \ldots + \gamma_\ell p_\ell \in S(\mathfrak{g})^{\mathfrak{g}_i}$ is an indecomposable polynomial with $\mathcal{G}(p) = (i_1,j_1) \tilde{+} (i_2,j_2) \tilde{+} \ldots \tilde{+} (i_\ell,j_\ell).$ Then $\{(i_1,0),\mathcal{G}(p)\} = 0.$ Here $\gamma_j$ is constant for all $j.$
\section{Constructing polynomial Poisson algebras from subalgebras of $\mathfrak{sl}(3,\mathbb{C})$ via the grading method}

\label{4}

In this section, we will provide specific examples to illustrate the application of the grading method described in Section $\ref{3}$, aimed at identifying the potentially permissible monomials in the Poisson bracket relations. Specifically, we will examine the following reduction chains within the complex semisimple Lie algebra $\mathfrak{sl}(3,\mathbb{C})$ and its compact real form $\mathfrak{su}(3)$: $\mathfrak{so}(3) \subset \mathfrak{su}(3 )$, $\mathfrak{o}(3) \subset \mathfrak{sl}(3,\mathbb{C}) $, and $\mathfrak{h} \subset \mathfrak{sl}(3,\mathbb{C})$. Here, $\mathfrak{h}$ represents the Cartan subalgebra of $\mathfrak{sl}(3,\mathbb{C})$.  It can be demonstrated that identifying the grading for each generator of polynomial algebras simplifies the construction of the compact form. %{At the computation level, this will help us reduce the polynomial time in computing the explicit generators. For more details, see \cite{} {\rr [NOT specified reference]}\color{blue} Here I would like to cite our $\mathfrak{su}(4)$, but not sure if it is a good ideal or not}
%\subsection{Finitely generated algebras from reduction chains of $\mathbb{Z}_2$-grading Lie algebras}

\subsection{The reduction chain $\mathfrak{so}(3) \subset \mathfrak{su}(3)$}
\label{4.1}
The first example that we consider is related to one of the best studied missing label problems, so-called Elliott chain $\mathfrak{so}(3) \subset \mathfrak{su}(3)$ relevant to the study of the Elliot model in Nuclear Physics \cite{MR4411095,MR4660510,MR92620,MR94178}. The generators of the Lie algebra $\mathfrak{su}(3)$ admit the Gell-Mann basis \cite{MR4411095}, and the subalgebra $\mathfrak{so}(3)$ is spanned by three orbital angular momentum operators $\textbf{L} = (L_{-1},L_0,L_{+1})$, where  \[ L_0=E_{11}-E_{22} \] \[ L_{+1}= -E_{13}-E_{32} \] \[ L_{-1}= E_{31} + E_{23}. \]    Here $E_{ij}$ is the $3\times 3$ elementary matrix with entries $(E_{ij})_{kl}=\delta_{ik}\delta_{jl}$, where we introduced the Kronecker delta.  In terms of these matrices, the generators are explicitly given by
\[ J_0 = \sqrt{\frac{3}{2}}  Y, \text{ with } Y= \frac{1}{2} (E_{11}+E_{22}-2 E_{33}) \]
\[ J_{+1}= \sqrt{2} (E_{32} -E_{13} ), \quad J_{-1}= \sqrt{2}  (E_{31}-E_{23} ) \] \[ J_{+2} = E_{12} , \quad J_{-2} = E_{21}. \]

In this way, we can find a linear basis of $\mathfrak{su}(3)$ given by $\{L_0,L_{\pm 1}, J_0,J_{\pm 1},J_{\pm 2}\}$. It admits the linear decomposition $\mathfrak{su}(3) = \mathfrak{g}_1 \oplus \mathfrak{g}_2 ,$ where \begin{align}
    \mathfrak{g}_1 = \mathrm{span} \{L_0,L_{+1},L_{-1}\} \text{ and } \mathfrak{g}_2 = \mathrm{span} \{J_0,J_{+1},J_{-1},J_{+2},J_{-2}\} . \label{eq:basu3}
\end{align} Here $L_0$ plays the role of a generator of the Cartan subalgebra. The following commutator relations are satisfied: 

\begin{table}[h]
\centering
    \begin{tabular}{|c|c|c|c|c|c|c|c|c|}
        \hline
        $[\cdot,\cdot]$ & $L_0$ & $L_{+1}$ & $L_{-1}$  & $J_0$ & $J_{+1}$ & $J_{-1}$ & $J_{+2}$ & $J_{-2}$ \\  \hline
        $L_0$ & $0$ & $L_{+1}$ & $-L_{-1}$ &  $0$ & $J_{+1}$ & $-J_{-1}$ & $2J_{+2}$ & $-2J_{-2}$ \\ \hline
        $L_{+1}$ & $-L_{+1}$ & $0$  & $-L_0$   & $-\frac{3\sqrt{3}}{4} J_{+1}$ & $-2\sqrt{2} J_{+2}$ & $-\frac{4}{\sqrt{3}} J_0$ & $0$ &$-\frac{1}{\sqrt{2}} J_{-1}$ \\ \hline
        $L_{-1}$  & $L_{-1}$ & $L_0$ & $0$    & $\frac{3\sqrt{3}}{4} J_{-1}$ & $\frac{4}{\sqrt{3}}J_0$ & $2\sqrt{2}J_{-2}$ &$\frac{1}{\sqrt{2}} J_{+1}$ & $0$ \\ \hline
        $J_0$ & $0$ & $\frac{3\sqrt{3}}{4} J_{+1}$ & $-\frac{3\sqrt{3}}{4} J_{-1}$ & $0$  & $\frac{3\sqrt{3}}{2} L_{+1}$ & $-\frac{3\sqrt{3}}{2}L_{-1}$ & $0$ & $0$ \\ \hline
        $J_{+1}$ & $-J_{+1}$ & $2\sqrt{2} J_{+2}$ & $-\frac{4}{\sqrt{3}}J_0$   & $-\frac{3\sqrt{3}}{2} L_{+1}$ & $0$  & $-2L_0$ & $0$ & $\sqrt{2}L_{-1}$ \\ \hline
       $J_{-1}$ & $J_{-1}$& $\frac{4}{\sqrt{3}} J_0$ & $-2\sqrt{2} J_{-2}$   & $\frac{3\sqrt{3}}{2} L_{-1}$ &$2L_0$ &$0$  & $- \sqrt{2} L_{+1}$ & $0$ \\ \hline
        $J_{+2}$ & $-2J_{+2}$ & $0$ &$ -\frac{1}{\sqrt{2}}J_{+1}$  &$0$ & $0$ & $ \sqrt{2}L_{+1}$  & $0$  & $L_0$ \\ \hline
       $J_{-2}$ & $2J_{-2}$ & $\frac{1}{\sqrt{2}}J_{-1}$ & $0$   & $0$ & $-\sqrt{2} L_{-1}$ &$0$ & $-L_0$ & $0$  \\ \hline
    \end{tabular}

    \quad

    \caption{Commutator relations of $\mathfrak{su}(3)$}
    \label{tab:placeholder_label}
\end{table}

Analogously, in the dual space $\mathfrak{su}^*(3),$ we consider the basis $\beta_{\mathfrak{su}^*(3)} = \{l_0, l_{\pm 1}, j_0, j_{\pm 1}, j_{\pm 2} \}$ with the similar non-trivial Poisson bracket defined in Table \ref{tab:placeholder_label}. Later, we will determine the finitely-generated polynomial algebra in the centralizer $S(\mathfrak{su}(3))^{\mathfrak{so}(3)}$.

\subsubsection{Elliott chain and construction of the generators via a weight zero type criteria}
\label{4.1.1}

To determine elements in the commutant, we use the Cartan generator (which we denote $L_0$ for elements in the Lie algebra and $l_0$ for the dual space) of $\mathfrak{g}_2$ and the related eigenvalue for the variables of the dual space of $ \mathfrak{g}_2$, but also for the variable of the dual space of $ \mathfrak{g}_1$. If the construction of the Poisson centralizer is not based on basis-dependent results, for example, the construction of a Casimir invariant, certain properties can be highlighted on a certain basis. We propose labeling the generators of $\mathfrak{g}_2$ as $L_t$ ($l_t$ for the dual space) where $t$ are the numbers of the eigenvalue relative to $L_0$ and denoting the element of $\mathfrak{g}_2$ as $J_t$ ($j_t$ for the dual space). Namely, we have 
 \begin{equation}
\{l_0,l_t\}= \lambda_t l_t,  \quad \quad \{l_0,j_t\}= \chi_t j_t.
\end{equation}
We now look for monomials such that their eigenvalue relative to $l_0$ is 0. This reduces the dimensionality of the problem and facilitates the construction of the commutant. In the basis from Table $\ref{tab:placeholder_label}$ (also valid in the Poisson setting), all terms have weight zero with respect to taking the sum of the index in each monomial of the different polynomials. From the first column of Table $\ref{tab:placeholder_label}$, the eigenvalues of all the generators are 
\begin{align}
    |L_0|=& \, 0,\quad |L_{+1}|=1,\quad |L_{-1}|=-1; \\
     |J_0|=& \, 0,\quad |J_{\pm1}|=\pm 1,\quad |J_{\pm 2}|=\pm 2 \, .
\end{align} This points out how an appropriate basis can be used not only to constrain the number of monomials by the degree in the PBW basis of the Lie algebras, but also to use criteria such as the eigenvalue relative to certain Cartan generators or certain gradings. Using symbolic computing packages, we are able to find that the generators (linearly independent and indecomposable polynomials) of the Poisson centralizer
$S(\mathfrak{su}(3))^{\mathfrak{so}(3)}$
 are given by
\begin{align*}
 M_1 &= l_0^2 - 2 l_{-1} l_{+1} ; \\
 M_2 &= j_0^2 - \frac{9}{8} j_{-1} j_{+1} + \frac{9}{2} j_{-2} j_{+2}; \\
 M_3 &= l_0^2 j_0 + l_{-1} l_{+1} j_0 - \frac{3}{4} \sqrt{3} l_0 l_{+1} j_{-1} +  \frac{3}{2} \sqrt{\frac{3}{2}} l_{+1}^2 j_{-2} -   \frac{3}{4} \sqrt{3} l_0 l_{-1} j_{+1} +  \frac{3}{2} \sqrt{\frac{3}{2}} l_{-1}^2 j_{+2}; \\
 M_4 &= -  \frac{16}{27} j_0^3 + j_0 j_{-1} j_{+1} - \frac{3}{2}  \sqrt{\frac{3}{2}} j_{-2} j_{+1}^2 - \frac{3}{2} \sqrt{\frac{3}{2}} j_{-1}^2 j_{+2} + 8 j_0 j_{-2} j_{+2}; \\
 M_5 &= \frac{16}{3}  l_{-1} l_{+1} j_0^2 - \frac{4}{\sqrt{3}} l_0 l_{+1} j_0 j_{-1}  + \frac{3}{2} l_{+1}^2 j_{-1}^2   -  8  \sqrt{\frac{2}{3}} l_{+1}^2 j_0 j_{-2} - \frac{4}{\sqrt{3}}  l_0 l_{-1} j_0 j_{+1}  +  l_0^2 j_{-1} j_{+1} \\
  &- 5 l_{-1} l_{+1} j_{-1} j_{+1} + 6 \sqrt{2} l_0 l_{+1} j_{-2} j_{+1} +  \frac{3}{2} l_{-1}^2 j_{+1}^2 - 8 \sqrt{\frac{2}{3}} l_{-1}^2 j_0 j_{+2}  + 6 \sqrt{2} l_0 l_{-1} j_{-1} j_{+2} \\
  & - 16 l_0^2 j_{-2} j_{+2} + 8 l_{-1} l_{+1} j_{-2} j_{+2}; \\
 M_6 &=  \frac{1}{4 \sqrt{2}}  l_0 l_{+1}^2 j_0 j_{-1}^2 - \frac{1}{16} \sqrt{\frac{3}{2}} l_{+1}^3 j_{-1}^3 - \frac{2}{\sqrt{3}}  l_0 l_{+1}^2 j_0^2 j_{-2}   + \frac{1}{2} l_{+1}^3 j_0 j_{-1} j_{-2} -  \frac{1}{8} \sqrt{\frac{3}{2}} l_0^2 l_{+1} j_{-1}^2 j_{+1}\\
 &- \frac{1}{16} \sqrt{\frac{3}{2}} l_{-1} l_{+1}^2 j_{-1}^2 j_{+1} + l_0^2 l_{+1} j_0 j_{-2} j_{+1} +  \frac{1}{2} l_{-1} l_{+1}^2 j_0 j_{-2} j_{+1} + \frac{1}{8} \sqrt{3} l_0 l_{+1}^2 j_{-1} j_{-2} j_{+1} - \frac{1}{2} \sqrt{\frac{3}{2}} l_{+1}^3 j_{-2}^2 j_{+1} \\
&-
\frac{1}{4 \sqrt{2}} l_0 l_{-1}^2 j_0 j_{+1}^2+ \frac{1}{8} \sqrt{\frac{3}{2}} l_0^2 l_{-1} j_{-1} j_{+1}^2  + \frac{1}{16} \sqrt{\frac{3}{2}} l_{-1}^2 l_{+1} j_{-1} j_{+1}^2 - \frac{1}{8}  \sqrt{3} l_0^3 j_{-2} j_{+1}^2 \\
&- \frac{3}{8} \sqrt{3} l_0 l_{-1} l_{+1} j_{-2} j_{+1}^2  + \frac{1}{16} \sqrt{\frac{3}{2}} l_{-1}^3 j_{+1}^3 + \frac{2}{\sqrt{3}}  l_0 l_{-1}^2 j_0^2 j_{+2} - l_0^2 l_{-1} j_0 j_{-1} j_{+2} \\
& - \frac{1}{2} l_{-1}^2 l_{+1} j_0 j_{-1} j_{+2} +   \frac{1}{8} \sqrt{3} l_0^3 j_{-1}^2 j_{+2}  +  \frac{3}{8} \sqrt{3} l_0 l_{-1} l_{+1} j_{-1}^2 j_{+2}  - \sqrt{\frac{3}{2}} l_0^2 l_{+1} j_{-1} j_{-2} j_{+2} \\
& - \frac{1}{2} \sqrt{\frac{3}{2}} l_{-1} l_{+1}^2 j_{-1} j_{-2} j_{+2} + \sqrt{3} l_0 l_{+1}^2 j_{-2}^2 j_{+2} - \frac{1}{2} l_{-1}^3 j_0 j_{+1} j_{+2} \\
&- \frac{1}{8} \sqrt{3} l_0 l_{-1}^2 j_{-1} j_{+1} j_{+2} + \sqrt{\frac{3}{2}} l_0^2 l_{-1} j_{-2} j_{+1} j_{+2} + \frac{1}{2} \sqrt{\frac{3}{2}} l_{-1}^2 l_{+1} j_{-2} j_{+1} j_{+2} \\
& + \frac{1}{2} \sqrt{\frac{3}{2}} l_{-1}^3 j_{-1} j_{+2}^2 - \sqrt{3} l_0 l_{-1}^2 j_{-2} j_{+2}^2.
\end{align*}

  Using the notation of Section $\ref{2},$ we deduce that $\textbf{Q}_6 = \textbf{q}_2 \sqcup \textbf{q}_3\sqcup \textbf{q}_4\sqcup \textbf{q}_6,$ where $\textbf{q}_2 = \{M_1,M_2\}$, $\textbf{q}_3 = \{M_3,M_4\}$, $\textbf{q}_4 = \{M_5\}$ and $\textbf{q}_6 = \{M_6\}$.  It is important to note that, in the following sections, the symbol $\textbf{q}_j$ will be used consistently in all illustrative examples. However, it should be observed that they may vary in cardinality from one context to another, and thus do not maintain a uniform length throughout. In the compact form, the relations can be reformulated as follows
  \begin{align}
  \nonumber
      \{\textbf{q}_2,\textbf{q}_2\} \sim & \, \textbf{q}_3   \\
      \nonumber
      \{\textbf{q}_2,\textbf{q}_3\} \sim &  \,  \textbf{q}_2^2 + \textbf{q}_4  \\
      \nonumber
       \{\textbf{q}_2,\textbf{q}_4\} \sim & \,  \textbf{q}_2\textbf{q}_3    \\
       \nonumber
       \{\textbf{q}_2,\textbf{q}_6\} \sim & \,  \textbf{q}_3\textbf{q}_4 + \textbf{q}_2^2\textbf{q}_3  \\
    \{\textbf{q}_3,\textbf{q}_4\} \sim & \,  \textbf{q}_2^3 + \textbf{q}_3^2 + \textbf{q}_6 \label{eq:compact1}  \\
    \nonumber
     \{\textbf{q}_3,\textbf{q}_6\} \sim & \,  \textbf{q}_2^4 + \textbf{q}_4^2 + \textbf{q}_2 \textbf{q}_6 + \textbf{q}_2\textbf{q}_3^2  + \textbf{q}_2^2\textbf{q}_4 \\
     \nonumber
        \{\textbf{q}_4,\textbf{q}_6\} \sim & \,  \textbf{q}_3^3 +    \textbf{q}_2^3\textbf{q}_3 +  \textbf{q}_2\textbf{q}_3 \textbf{q}_4 + \textbf{q}_3 \textbf{q}_6.
  \end{align}

We then limit the permissible monomials from $\eqref{eq:compact1}$ by employing the grading method.  Following the definition, we see that the generators derived from the Elliott chain admit the following homogeneous grading   \begin{align}
    \begin{matrix}
        \mathcal{G}\left(M_1\right)= (2,0),\quad \mathcal{G}\left(M_2\right)=(0,2) ,\quad \mathcal{G}\left(M_3\right)= (2,1)  \\
  \mathcal{G}\left(M_4\right)=(0,3),\quad \mathcal{G}\left(M_5\right)=(2,2), \quad \mathcal{G}\left(M_6\right)= (3,3).
    \end{matrix} \label{eq:grad1}
\end{align}

Using the original notations, we aim to close the bracket relations in the following forms \[ \{M_k,M_l\}= \sum_r \Gamma^r_{kl} M_r + \sum_{s,t} \Gamma^{st}_{kl} M_s M_t +  \sum_{u,v,w}  \Gamma^{uvw}_{kl} M_u M_v M_w \] such that \begin{align}
         \mathcal{G}\left(\{M_k,M_l\}\right)= \mathcal{G} \left(M_r\right) \tilde{+} \mathcal{G}\left(M_s M_t\right)  \tilde{+}\mathcal{G}\left(M_u M_v M_w\right).
\end{align} Here $\Gamma^r_{kl}, \Gamma^{st}_{kl}$ and $\Gamma^{uvw}_{kl}$ are arbitrary constants with indices running from $1$ up to $6$.

From Corollary $\ref{3.7}$, it is sufficient to conclude that $\{M_1,M_l\} = 0$ for all $1 \leq l \leq 6$, therefore $M_1$ is the central element of the polynomial algebra.  We now start with $M_2$, and calculate the grading of the term $\{M_2,M_l\}$ for all $3 \leq l \leq 6$. Using Lemma $\ref{3.6}$, we deduce
\begin{align}
    \mathcal{G} (\{M_2,M_3\}) = & \, (1,3) \tilde{+} (3,1), \text{ }\mathcal{G} (\{M_2,M_4\}) = (0,4) \tilde{+} (1,3),  \nonumber \\
    \mathcal{G} (\{M_2,M_5\}) = & \, (2,3) \tilde{+} (3,2), \text{ }\mathcal{G} (\{M_2,M_6\}) = (3,4) \tilde{+} (4,3).\label{eq:grad2}
\end{align}
By Lemma $\ref{mul},$ $(i_1,i_2)+(i_1',i_2')=(i_1+i_1',i_2+i_2')$. Together with $\eqref{eq:grad1},$ the permissible polynomials from the homogeneous gradings in \eqref{eq:grad2} are 
\begin{align}
    \{M_2,M_3\} = & \, 0 \nonumber\\
    \{M_2,M_4\} = & \, \Gamma_{24}^{22} M_2^2 \nonumber \\
    \{M_2,M_5\} = & \,\Gamma_{25}^{14} M_1M_4 + \Gamma_{25}^{23} M_2M_3\label{eq:grad3} \\
    \{M_2,M_6\} =& \, \Gamma_{26}^{35} M_3M_5 + \Gamma_{26}^{114} M_1^2 M_4 + \Gamma_{26}^{123} M_1 M_2 M_3. \nonumber
\end{align}
Here $\Gamma_{24}^{22}, \Gamma_{25}^{14},\Gamma_{25}^{23}, \Gamma_{26}^{35} , \Gamma_{26}^{114} ,\Gamma_{26}^{123} \in \mathbb{R}$  are arbitrary constants that could be zero. Note that $\{M_2,M_3\} = 0$ is deduced from the absence of any polynomials in $\mathcal{Q}_{\mathfrak{su}(3)}(d)$ with the grading $(1,3)$ or $(3,1)$. We subsequently proceed to examine the Poisson brackets $\{M_3,M_4\}, \{M_3,M_5\}$ and $\{M_3,M_6\}$. A direct calculation shows that the grading of these terms are
\begin{align*}
    \mathcal{G} (\{M_3,M_4\}) = & \, (1,4) \tilde{+} (3,2), \\
    \mathcal{G} (\{M_3,M_5\}) = & \,(3,3) \tilde{+} (5,1), \\
    \mathcal{G} (\{M_3,M_6\}) = & \, (4,4) \tilde{+} (6,2)   .
\end{align*}
Each of the Poisson brackets above is spanned as follows
\begin{align}
    \{M_3,M_4\} = & \, 0,\text{ } \quad \{M_3,M_5\} = \Gamma_{35}^6 M_6  \nonumber \\
    \{M_3,M_6\} = & \, \Gamma_{36}^{115} M_1^2 M_5 + \Gamma_{36}^{125} M_1M_2M_5 + \Gamma_{36}^{1122} M_1^2 M_2^2 + \Gamma_{36}^{1112} M_1^3 M_2 + \Gamma_{36}^{234} M_2 M_3M_4 \nonumber  \\
    & + \Gamma_{36}^{134} M_1M_3M_4 + \Gamma_{36}^{55} M_5^2 +\Gamma_{36}^{233} M_2 M_3^2+\Gamma_{36}^{1122} M_1^2 M_2^2.
\end{align}
Here $\Gamma_{35}^6, \Gamma_{36}^{115},\ldots,\Gamma_{36}^{1122} \in \mathbb{R}$ are constants. Next, we present the grading of the rest of the commutator relations. A direct calculation gives us that
\begin{align}
    \mathcal{G} (\{M_4,M_5\}) = & \, (1,5) \tilde{+} (3,3), \nonumber \\
    \mathcal{G} (\{M_4,M_6\}) = & \,(2 ,6) \tilde{+} (4,4),\label{eq:gardin3} \\
    \mathcal{G} (\{M_5,M_6\}) = & \, (4,5) \tilde{+} (6,3)  \nonumber .
\end{align}
From the grading in $\eqref{eq:gardin3}$, the permissible polynomials in each Poisson bracket are given by
\begin{align*}
    \{M_4,M_5\} = & \, \Gamma_{45}^6 M_6 \\
    \{M_4,M_6\} = & \, \Gamma_{46}^{55} M_5^2+ \Gamma_{46}^{125} M_1 M_2 M_5+ \Gamma_{46}^{134} M_1 M_3 M_4+ \Gamma_{46}^{233} M_2 M_3^2 + \Gamma_{46}^{1122} M_1^2 M_2^2+ \Gamma_{46}^{144} M_1 M_4^2 \\
    & + \Gamma_{46}^{225}  M_2^2 M_5+ \Gamma_{46}^{234} M_2  M_3 M_4\\
     \{M_5,M_6\} = & \,  \Gamma_{56}^{135} M_1 M_3 M_5    + \Gamma_{56}^{145}  M_1 M_4 M_5+ \Gamma_{56}^{235}  M_2 M_3 M_5+   \Gamma_{56}^{333} M_3^3+ \Gamma_{56}^{334} M_3^2 M_4 \\
     & + \Gamma_{56}^{1114} M_1^3 M_4+ \Gamma_{56}^{1123} M_1^2 M_2 M_3+ \Gamma_{56}^{1124} M_1^2 M_2 M_4+ \Gamma_{56}^{1223} M_1 M_2^2 M_3,
\end{align*}
where $\Gamma_{45}^6, \Gamma_{46}^{55},\ldots,\Gamma_{46}^{234}, \Gamma_{56}^{135}\ldots,\Gamma_{56}^{1223} \in \mathbb{R}$ are arbitrary constants.

With this information, the compact reformulation adopts the form 
\begin{align}
  \nonumber
      \{\textbf{q}_2,\textbf{q}_2\} \sim & \, 0   \\
      \nonumber
      \{\textbf{q}_2,\textbf{q}_3\} \sim  & \,  \textbf{q}_2^2   \\
      \nonumber
       \{\textbf{q}_2,\textbf{q}_4\} \sim & \,  \textbf{q}_2\textbf{q}_3    \\
       \{\textbf{q}_2,\textbf{q}_6\} \sim  & \,     \textbf{q}_3\textbf{q}_4 + \textbf{q}_2^2\textbf{q}_3 \label{eq:compact11} \\
       \nonumber
    \{\textbf{q}_3,\textbf{q}_4\} \sim & \,   \textbf{q}_6\\
    \nonumber
    \{\textbf{q}_3,\textbf{q}_6\} \sim   & \, \textbf{q}_2^4   +  \textbf{q}_2\textbf{q}_6 + \textbf{q}_2\textbf{q}_3^2  + \textbf{q}_2^2\textbf{q}_4  \\
     \nonumber
        \{\textbf{q}_4,\textbf{q}_6\} \sim & \,  \textbf{q}_3^3 +  \textbf{q}_2 \textbf{q}_3\textbf{q}_4  + \textbf{q}_3\textbf{q}_6 .
  \end{align} In order to visually demonstrate the efficacy of the grading method, we present an extended table that facilitates a detailed comparison of the number of allowed polynomials within non-trivial Poisson brackets. Specifically, in the following Table \ref{tab:my_label0}, the second column reveals the number of expected polynomials from the compact forms, whereas the first one indicates the number of permissible terms obtained via the grading method. For instances, from \ref{eq:compact1}, the polynomials allowed in the Poisson bracket $\{\textbf{q}_2,\textbf{q}_2\} $ are a linear combination of the elements in $\textbf{q}_3$. Hence, without using grading method, the number of allowed polynomials from the Poisson bracket $\{\textbf{q}_2,\textbf{q}_2\}$ is $2$. On the other hand, the last column illustrates the maximum number of polynomials allowed in each of the Poisson brackets in the compact form obtained through the grading method. We observe that in the $\{\textbf{q}_2,\textbf{q}_2\}$ case, we obtain zero terms, which means that everything commutes. This is an expected result, as $M_1$ is the central element. To provide another explicit example, let us consider the compact form $\{\textbf{q}_2,\textbf{q}_3\}$, which contains $\{M_1,M_3\}, \{M_1,M_4\}$, $\{M_2,M_3\}$, and $\{M_2,M_4\}.$ Using \eqref{eq:grad3}, we see that the number of polynomials allowed in these non-trivial Poisson brackets is either $0$ or $1.$ Hence, as shown in Table \ref{tab:my_label0} below, the maximum number of permissible polynomials appearing in the expansion of $\{\textbf{q}_2,\textbf{q}_3\}$ by the grading method is $1$. It is important to note that the interpretation of all the comparison tables presented below can be approached in a consistent manner. 
    \begin{table}[h]
     \centering
     \begin{tabular}{|c|c|c|c|}
     \hline
     \textsf{ Poisson brackets}    & $\begin{matrix}
         \textsf{   No. of polynomials without} \\
         \textsf{  using the grading method }
      \end{matrix}$  & $\begin{matrix}
          \textsf{Maximum No. of polynomials after applying} \\
          \textsf{the grading method} 
      \end{matrix}$    \\
          \hline 
  $\{\textbf{q}_2,\textbf{q}_2\}$       & $ 2 $ & $0 $ \\
   \hline 
  $\{\textbf{q}_2,\textbf{q}_3\}$       & $4 $ & $1 $  \\
          \hline 
 $\{\textbf{q}_2,\textbf{q}_4\}$       & $4 $ & $2 $  \\
          \hline 
    $\{\textbf{q}_2,\textbf{q}_6\}$       & $8 $ & $7 $   \\
           \hline 
  $\{\textbf{q}_3,\textbf{q}_4\}$       & $ 8 $ & $1 $ \\
   \hline 
  $\{\textbf{q}_3,\textbf{q}_6\}$       & $17 $ & $9 $ \\
          \hline 
 $\{\textbf{q}_4,\textbf{q}_6\}$       & $18 $ & $9 $  \\
          \hline  
     \end{tabular}
     
     \quad
     
     \caption{Comparison of the number of polynomials }
     \label{tab:my_label0}
 \end{table}

Now, using the commutator relations in Table $\ref{tab:placeholder_label}$, together with $\eqref{eq:compact11}$, we can further determine the unknown coefficients given in the Poisson brackets above. This enables a change of notation by applying their respective gradations. We then have
\[ \{B_{20},B_{i_1i_2}\}=  \, 0 ,\quad i_1i_2 \in \{02,21,03,22,33\} ;   \quad \text{ }  \{B_{02},B_{i_1i_2}\}=   \, 0,\quad  i_1i_2 \in \{21,03,22,33\}; \quad \{B_{21},B_{03}\}=  \,0;   \] and the non-trivial brackets are
\begin{align*}
 \{B_{21},B_{22}\}=  & \,-36 \sqrt{2} B_{33} \\
 \{B_{21},B_{33}\}= & \, - \frac{16}{27} \sqrt{2} B_{02}^2 B_{02}^2 + \frac{8}{27} \sqrt{2} B_{02} B_{21}^2 - \frac{1}{\sqrt{2}} B_{20} B_{21} B_{03}  - \frac{1}{\sqrt{2}} B_{20} B_{02} B_{22} - \frac{3}{16 \sqrt{2}} B_{22}^2\\
 \{B_{03},B_{22}\}= & \, 72 \sqrt{2} B_{33} \\
 \{B_{03},B_{33}\}= & \,  \frac{32}{27} \sqrt{2} B_{20}^2 B_{02}^2 - \frac{16}{27} \sqrt{2} B_{02} B_{21}^2 + \sqrt{2} B_{20} B_{21} B_{03} +\sqrt{2} B_{20} B_{02} B_{22} + \frac{3}{8 \sqrt{2}} B_{22}^2 \\
 \{B_{22},B_{33}\}=  & \, - \frac{16}{9} \sqrt{2} B_{20}^2 B_{02} B_{21} - \frac{128}{243} B_{20} B_{02}^2 B_{21} + \frac{32}{27} \sqrt{2} B_{21}^3  - \sqrt{2} B_{20}^3 B_{03} + \frac{8}{9} \sqrt{2} B_{20}^2 B_{02} B_{03}\\
 & \, - \frac{16}{9} \sqrt{2} B_{21}^2 B_{03} - \sqrt{2} B_{20} B_{21} B_{22} - \frac{16}{27} \sqrt{2} B_{02} B_{21} B_{22} + \frac{1}{\sqrt{2}} B_{20} B_{03} B_{22}.
\end{align*}  Note that $B_{20}$ lies in the center of this polynomial algebra.  Hence, we claim that $\textbf{Alg} \left\langle \textbf{Q}_6 \right\rangle$ with a Poisson bracket $\{\cdot,\cdot\}$ forms a cubic Poisson algebra $\mathcal{Q}_{\mathfrak{su}(3)}(3)$.

\subsection{The reduction chain $\mathfrak{o}(3) \subset \mathfrak{sl}(3,\mathbb{C})$}
\label{4.2}
In this Subsection \ref{4.2}, we analyze the reduction chain $\mathfrak{o}(3) \subset \mathfrak{sl}(3,\mathbb{C})$.  Instead of using a two-step decomposition, we consider a triangular decomposition such that \[\mathfrak{sl}(3,\mathbb{C}) =    \mathrm{span}_\mathbb{C} \{E_{11} -E_{22}, E_{22} - E_{33}\} \oplus \bigoplus_{j\neq k=1}^3 \mathrm{span}_\mathbb{C} \{E_{jk}\},\] where $\mathrm{span}_\mathbb{C} \{E_{11} -E_{22}, E_{22} - E_{33}\} $ is the Cartan subalgebra. We observe that, with this choice, $E_{12}$ and $E_{23}$ are the generators corresponding to the simple roots $\alpha_1,\alpha_2,$ and $E_{13}$ is associated with the sum $\alpha_1 + \alpha_2$.  The commutator relations are given in the following table:

\begin{table}[ht]
\centering
\begin{tabular}{|l|c|c|c|c|c|c|c|c|}
\hline
\([\cdot, \cdot]\) & \( H_1 \) & \( H_2 \) & \( E_{12} \) & \( E_{21} \) & \( E_{23} \) & \( E_{32} \) & \( E_{13} \) & \( E_{31} \) \\
\hline
\( H_1 \) & \( 0 \) & \( 0 \) & \( 2E_{12} \) & \( -2E_{21} \) & \( -E_{23} \) & \( E_{32} \) & \( E_{13} \) & \( -E_{31} \) \\
\hline
\( H_2 \) & \( 0 \) & \( 0 \) & \( E_{12} \) & \( -E_{21} \) & \( 2E_{23} \) & \( -2E_{32} \) & \( E_{13} \) & \( -E_{31} \) \\
\hline
\( E_{12} \) & \( -2E_{12} \) & \( -E_{12} \) & \( 0 \) & \( H_1 \) & \( E_{13} \) & \( 0 \) & \( 0 \) & \( 0 \) \\
\hline
\( E_{21} \) & \( 2E_{21} \) & \( E_{21} \) & \( -H_1 \) & \( 0 \) & \( 0 \) & \( 0 \) & \( 0 \) & \( 0 \) \\
\hline
\( E_{23} \) & \( E_{23} \) & \( -2E_{23} \) & \( -E_{13} \) & \( 0 \) & \( 0 \) & \( H_2 \) & \( 0 \) & \( 0 \) \\
\hline
\( E_{32} \) & \( -E_{32} \) & \( 2E_{32} \) & \( 0 \) & \( 0 \) & \( -H_2 \) & \( 0 \) & \( 0 \) & \( 0 \) \\
\hline
\( E_{13} \) & \( -E_{13} \) & \( -E_{13} \) & \( 0 \) & \( 0 \) & \( 0 \) & \( 0 \) & \( 0 \) & \( H_1 + H_2 \) \\
\hline
\( E_{31} \) & \( E_{31} \) & \( E_{31} \) & \( 0 \) & \( 0 \) & \( 0 \) & \( 0 \) & \( - (H_1 + H_2) \) & \( 0 \) \\
\hline
\end{tabular}

\quad

\caption{Commutator Table of \( \mathfrak{sl}(3, \mathbb{C}) \)}
\label{tab:sl3_commutator}
\end{table}

% \cite{MR1025215}

In what follows, we take $\mathfrak{g}_1 = \mathfrak{h}$, $\mathfrak{g}_2 = \mathfrak{g}^+$, and $\mathfrak{g}_3 = \mathfrak{g}^-.$ According to the commutator relations in  Table $\ref{tab:sl3_commutator}$, the nilpotent subalgebra $\mathfrak{g}^+ \cong \mathfrak{o}(3)$ of $\mathfrak{sl}(3,\mathbb{C})$ is spanned by the basis elements $\{E_{12},E_{23},E_{13}\}$. Consider the embedding chain $\mathfrak{o}(3) \subset \mathfrak{sl}(3,\mathbb{C})$. This reduction typically appears in the problem of decomposition of the enveloping algebra of a semi-simple Lie algebra, which has already been considered in \cite{MR4710584}. In this Subsection, we will construct the polynomial Poisson algebra in $S(\mathfrak{sl}(3,\mathbb{C}))^{\mathfrak{o}(3)}$, looking for the algebraic structure of the commutant via the grading method.

Recall that the coordinate in $\mathfrak{sl}^*(3,\mathbb{C})$ is given by $\textbf{x} = (h_1,h_2,e_{12},e_{13},e_{23},e_{21},e_{31},e_{32}).$ By construction, there are $6$ indecomposable (polynomial solutions) to the system of PDEs \begin{align*}
    \{e_{12},S(\mathfrak{g})\} = \{e_{23},S(\mathfrak{g})\} = \{e_{13},S(\mathfrak{g})\} = 0,
\end{align*} that are given by \begin{align}
\nonumber
    A_1 =& \, e_{13}, \quad A_2 = 3 e_{12}e_{23} + (h_1 -h_2) e_{13}, \\
    \nonumber
    A_3 =& \, h_1^2 + h_2^2 + h_1h_2 + 3(e_{12} e_{21} + e_{23}e_{32} + e_{13}e_{31} ), \label{eq:comm2}\\
    A_4 = & \, e_{12}e_{23}^2 + e_{13} \left(h_1 e_{23} - e_{13} e_{21}\right) , \\
    \nonumber
     A_5 = & \, e_{13} \left(e_{13}e_{32} + h_2 e_{12}\right) - e_{12}^2 e_{23}, \\
     \nonumber
     A_6 = & \, \frac{1}{2} \left(2 h_1^3+3 h_2 h_1^2-3 \left(h_2^2-3 e_{12} e_{21}+6 e_{23} e_{32}-3 e_{13} e_{31}\right) h_1 \right. \\
     \nonumber
     & \left.-2 h_2^3+27 (e_{13} e_{21} e_{32}+e_{12} e_{23} e_{31})+9 h_2 (2 e_{12} e_{21}-e_{23} e_{32}-e_{13} e_{31})\right).
\end{align} From the construction in Section $\ref{2},$ we deduce that $\textbf{Q}_3 = \textbf{q}_1 \sqcup \textbf{q}_2 \sqcup \textbf{q}_3$ with $\textbf{q}_1 = \{A_1\}$, $ \textbf{q}_2 = \{A_2,A_3\}$ and $ \textbf{q}_3 = \{A_4,A_5,A_6\}.$ The commutator relations in terms of the compact forms are \begin{align}
\nonumber
    \{\textbf{q}_1, \textbf{q}_2\} \sim  & \,\textbf{q}_2 + \textbf{q}_1^2 \\
    \nonumber
     \{\textbf{q}_1, \textbf{q}_3\} \sim  & \,  \{\textbf{q}_2,\textbf{q}_2\}  \sim  \textbf{q}_3 + \textbf{q}_1\textbf{q}_2 + \textbf{q}_1^3 \\
       \{\textbf{q}_2, \textbf{q}_3\} \sim  & \,\textbf{q}_2^2  + \textbf{q}_1^2 \textbf{q}_2  + \textbf{q}_1\textbf{q}_3 + \textbf{q}_1^4  \label{eq:compact2}\\
       \nonumber
    \{\textbf{q}_3, \textbf{q}_3\} \sim  & \,\textbf{q}_1\textbf{q}_2^2   + \textbf{q}_1^2\textbf{q}_3 + \textbf{q}_2\textbf{q}_3 + \textbf{q}_1^5 .
\end{align}

We now apply the grading method to reduce the possible terms from the compact form $\eqref{eq:compact2}.$ By definition, each term admits a non-homogeneous gradings listed as follows: \begin{align}
\nonumber
& \mathcal{G} (A_1 ) = (0,1,0) , \\
\nonumber
& \mathcal{G}(A_2) = (1,1,0) \tilde{+} (0,2,0) , \\
& \mathcal{G}(A_3) = (2,0,0) \tilde{+} (0,1,1), \label{eq:regr} \\
\nonumber
&  \mathcal{G} (A_4 ) = \mathcal{G}(A_5) = (1,2,0) \tilde{+} (0,3,0) \tilde{+} (0,2,1), \\
\nonumber
&  \mathcal{G} (A_6 ) = (3,0,0) \tilde{+} (1,1,1) \tilde{+} (0,1,2) \tilde{+} (0,2,1).
  % & \mathcal{G}(A_1) = (1,0), \text{ } \mathcal{G}(A_2) = (2,0) \tilde{+} (1,1), \text{ } \mathcal{G}(A_3) = (0,2) \tilde{+} (1,1) \\
 %   & \mathcal{G} (A_4) = (3,0) \tilde{+} (2,1)  =\mathcal{G}(A_5) = (3,0) \tilde{+} (2,1), \text{ } \mathcal{G}(A_6) =  (0,3) \tilde{+} (2,1) \tilde{+}(1,2).
 \end{align}%Notice that $ \mathcal{G}(A_2A_3) = (2,2) \tilde{+} (1,3) \tilde{+}(3,1) \tilde{+} (2,2) = (2,2) \tilde{+} (1,3) \tilde{+}(3,1) = \mathcal{G}(A_1A_6).$

The grading in $\eqref{eq:regr}$ provides us with an example in which two distinct elements may have the same grading. Using a new notation, we may write  \begin{align*}
    A_1 = &\, B_{010}, \text{ }  A_2 =B_{110} + B_{020}, \text{ }A_3 = B_{200} + B_{011}, \\
A_4 = & \,B_{030} + B_{120} + B_{021} ,\text{ } A_5 = B_{030}' + B_{120}' + B_{021}',\\
A_6 = & \, B_{300} + B_{111}+ B_{021}'' + B_{012} .
\end{align*} Here $B_{abc}$, $B_{abc}'$ and $ B_{abc}''$ are the generators corresponding to distinct homogeneous gradings for any $0 \leq a,b,c \leq 3$. We now calculate the grading of each non-trivial bracket. By definition, $\{B_{010},A_u\} = 0$ for all $1 \leq u \leq 6.$ Hence, $A_1$ is a central element, and we omit the calculation on these gradings. Moreover, from this fact, we can further conclude that $\{B_{0b0},A_u\} = 0$ for all $b.$ Starting from $A_2.$ Note that $ \{A_2,A_3\} = \{B_{110} + B_{020},   B_{200} + B_{011}\} =\{B_{110} ,   B_{200} \}+ \{B_{110} ,B_{011}\} .$ Then \begin{align*}
    \mathcal{G} \left(\{A_2,A_3\}\right) = & \, \mathcal{G} \left(\{B_{110} ,   B_{200} \}\right) \tilde{+} \mathcal{G} \left(\{B_{110} ,B_{011}\}\right)  =  (2,1,0) \tilde{+} (1,2,0) \tilde{+} (0,2,1).
\end{align*} This implies that $\{A_2,A_3\} = \Gamma_{23}^{13} A_1A_3$ for any $\Gamma_{23}^{13} \in \mathbb{R}.$ The determination of the coefficients depends on the structure constants of the Lie algebra. Under the commutator relations in Table $\ref{tab:sl3_commutator}$ and the generators in $\eqref{eq:comm2}$, we observe that, indeed, $\Gamma_{23}^{13} =0$ so that these two elements commute. We proceed to calculating the grading of Poisson brackets $\{A_2,A_4\}, \{A_2,A_5\}$ and $\{A_2,A_6\}$:
\begin{align*}
    \mathcal{G}\left(\{A_2,A_4\}\right) = & \,\mathcal{G} \left(\{B_{110},B_{030}\}\right) \tilde{+} \mathcal{G} \left(\{B_{110},B_{120}\}\right) \tilde{+} \mathcal{G} \left(\{B_{110},B_{021}\}\right)  \\
    = & \, (0,4,0) \tilde{+} (1,3,0) \tilde{+} (0,3,1) \tilde{+} (2,2,0) \\
    =& \,  \mathcal{G}\left(\{A_2,A_5\}\right) ,\\
       \mathcal{G}\left(\{A_2,A_6\}\right) = & \,\mathcal{G} \left(\{B_{110},B_{300}\}\right)\tilde{+} \mathcal{G} \left(\{B_{110},B_{111}\}\right) \tilde{+} \mathcal{G} \left(\{B_{110},B_{012}\}\right) \tilde{+} \mathcal{G} \left(\{B_{110},B_{021}''\}\right)  \\
       = & \, (3,1,0) \tilde{+} (1,2,1)  \tilde{+} (2,2,0) \tilde{+} (0,3,1) \tilde{+} (1,3,0).
\end{align*}
Using the grading on each generator defined in $\eqref{eq:regr}$, we conclude that the permissible polynomials in each bracket are as follows: \begin{align}
\nonumber
    \{A_2,A_4\} = & \,\Gamma_{24}^{1111} A_1^4 + \Gamma_{24}^{112} A_1^2A_2 +   \Gamma_{24}^{113} A_1^2A_3+    \Gamma_{24}^{14} A_1 A_4 +    \Gamma_{24}^{15} A_1 A_5  \\
   \{A_2,A_5\} = & \,\Gamma_{25}^{1111} A_1^4 + \Gamma_{25}^{112} A_1^2A_2 +   \Gamma_{25}^{113} A_1^2A_3+    \Gamma_{25}^{14} A_1 A_4 +    \Gamma_{25}^{15} A_1 A_5  \label{eq:comma2} \\
   \nonumber
       \{A_2,A_6\} = & \, \Gamma_{26}^{23}  A_2A_3 + \Gamma_{26}^{113} A_1^2 A_3.
 \end{align}
 Here, $\Gamma_{24}^{1111},\ldots,\Gamma_{24}^{15},\Gamma_{25}^{1111},\ldots,\Gamma_{25}^{15}, \Gamma_{26}^{23},\Gamma_{26}^{113} \in \mathbb{R}$ are arbitrary constants. Analyzing Table $\ref{tab:sl3_commutator}$ in conjunction with the expressions for the generators provided in $\eqref{eq:comm2}$, it becomes evident that the term $e_{13}^n$, valid for any integer $n$ where $n \geq 2$, along with the terms $(h_1-h_2)e_{13}$ and $h_1h_2 e_{13}^2$, do not belong to the brackets $\{A_2,A_4\}$ or $\{A_2,A_5\}$. Consequently, the expression for terms of the form $A_1^n$ as well as those of the form $A_1^2 A_2$ and $A_1^2 A_3$ must necessarily vanish, leading to the conclusion that $$\Gamma_{24}^{1111} = \Gamma_{24}^{112} = \Gamma_{24}^{113}  = \Gamma_{25}^{1111}  = \Gamma_{25}^{112} = \Gamma_{25}^{113} = 0.$$ Furthermore, it is important to note that within the grading $\mathcal{G}\left(\{A_2,A_6\}\right)$, the sole term comprising the structure $(0,3,1)$ derives from $\mathcal{G} \left(\{B_{110},B_{012}\}\right)$. When considering the generators detailed in $\eqref{eq:comm2}$, it is evident that the expression $3(e_{12}e_{21} + e_{23}e_{32} + e_{13}e_{31})e_{13}^2 \notin \{(h_1-h_2)e_{13},  e_{12}e_{23}e_{31}\}$. This exclusion decisively results in $\Gamma_{26}^{113} = 0$. Similarly, it is essential to consider the fact that the element $(0,3,1)$ included in $\mathcal{G}(A_2A_3)$ takes the form of $e_{12}e_{23}(e_{12}e_{21} + e_{23}e_{32} + e_{13}e_{31})$, and this particular form does not reside within the aforementioned bracket $\{(h_1-h_2)e_{13},  e_{12}e_{23}e_{31}\}$, which confirms that $\Gamma_{26}^{23} = 0$ and, consequently, $\{A_2,A_6\} = 0$.

Now, consider the Poisson brackets $\{A_3,A_4\},\{A_3,A_5\}$ and $\{A_3,A_6\}$. A direct computation shows that \begin{align*}
    \mathcal{G} \left(\{A_3,A_4\}\right) = & \, \mathcal{G} \left(\{B_{200},B_{120}\}\right) \tilde{+} \mathcal{G} \left(\{B_{011},B_{120}\}\right) \tilde{+} \mathcal{G} \left(\{B_{200},B_{021}\}\right) \tilde{+} \mathcal{G} \left(\{B_{011},B_{021}\}\right)  \\
     = & \, (1,3,0) \tilde{+} (2,2,0)\tilde{+} (3,1,0) \tilde{+} ( 1,2,1) \tilde{+} (0,3,1) \\
     = & \, \mathcal{G} \left(\{A_3,A_5\}\right) ,  \\
   \mathcal{G} \left(\{A_3,A_6\}\right) = & \, \mathcal{G} \left(\{B_{200},B_{111}\}\right) \tilde{+} \mathcal{G} \left(\{B_{200},B_{012}\}\right) \tilde{+} \mathcal{G} \left(\{B_{200},B_{021}\}\right) \tilde{+} \mathcal{G} \left(\{B_{011},B_{300}\}\right) \\
   & \,\tilde{+} \mathcal{G} \left(\{B_{011},B_{012}\}\right) \tilde{+}\mathcal{G} \left(\{B_{011},B_{111}\}\right) \tilde{+}\mathcal{G} \left(\{B_{011},B_{021}\}\right) \\
   = & \, (2,1,1) \tilde{+} (3,1,0) \tilde{+} (1,1,2) \tilde{+} (1,2,1) \tilde{+} (2,2,0) \tilde{+} (0,3,1)\tilde{+} (0,2,2) .% (3,0,1) \tilde{+}(4,0,0) \tilde{+}
\end{align*}
 The grading implies the following relations
\begin{align}
\nonumber
       \{A_3,A_4\} = & \,  \Gamma_{34}^{23} A_2 A_3 +  \Gamma_{34}^{113} A_1^2A_3   \\
   \{A_3,A_5\} = & \,  \Gamma_{35}^{23} A_2 A_3 +  \Gamma_{35}^{113} A_1^2A_3 \label{eq:a3aj}\\
   \nonumber
       \{A_3,A_6\} = & \, \Gamma_{36}^{113}  A_1^2A_3+ \Gamma_{36}^{16} A_1  A_6 +  \Gamma_{36}^{23}  A_2A_3
\end{align}
with $\Gamma_{34}^{23},\Gamma_{34}^{113},\Gamma_{35}^{23},\Gamma_{35}^{113}, \Gamma_{36}^{113},\Gamma_{36}^{16},\Gamma_{36}^{23} \in \mathbb{R}$ are arbitrary constants. Notice that the term $A_1^2 A_3$ is deduced from the grading of $\{B_{011},B_{120}\}.$ However, from the commutator relations in Table $\ref{tab:sl3_commutator}$ and the explicit form of the generators in $\eqref{eq:comm2}$, we observe that the term $(h_1^2 + h_2^2 + h_1h_2) e_{13}^2 \notin \{A_3,A_4\}.$ Hence $\Gamma_{34}^{113} = 0.$ For the same reason, we may conclude that $\Gamma_{35}^{113} =0.$ By examining the rest of the term in $\eqref{eq:a3aj},$ we deduce that $\{A_3,A_4\} = \{A_3,A_5\} = \{A_3,A_6\} = 0$.

For the last Poisson brackets, we have \begin{align*}
    \mathcal{G} \left(\{A_4,A_5\}\right) = & \,    \mathcal{G} \left(\{B_{120},B_{021}'\}\right) \tilde{+}\mathcal{G} \left(\{B_{120},B_{120}'\}\right)  \tilde{+} \mathcal{G} \left(\{B_{021},B_{021}'\}\right)\tilde{+} \mathcal{G} \left(\{B_{021},B_{120}'\}\right)  \\
    = & \, (1,4,0) \tilde{+} (0,4,1) \tilde{+} (1,3,1) \tilde{+} (2,3,0) \tilde{+} (0,5,0),\\
       \mathcal{G} \left(\{A_4,A_6\}\right)= & \,   \mathcal{G} \left(\{B_{120},B_{300}\}\right) \tilde{+} \mathcal{G} \left(\{B_{120},B_{111}\}\right) \tilde{+} \mathcal{G} \left(\{B_{120},B_{012}\}\right) \tilde{+}\mathcal{G} \left(\{B_{120},B_{021}''\}\right)\tilde{+} \mathcal{G} \left(\{B_{021},B_{111}\}\right) \\
    & \,\tilde{+} \mathcal{G} \left(\{B_{021},B_{300}\}\right) \tilde{+} \mathcal{G} \left(\{B_{021},B_{012}\}\right) \tilde{+} \mathcal{G} \left(\{B_{021},B_{021}''\}\right)  \\
    = & \,  (1,3,1) \tilde{+} (3,2,0) \tilde{+} (2,3,0) \tilde{+}(2,2,1) \tilde{+} (0,4,1) \tilde{+}(1,4,0) \tilde{+}(4,1,0)\tilde{+} (1,2,2) \tilde{+} (0,3,2) \\
    = & \, \mathcal{G}\left(\{A_5,A_6\}\right).
\end{align*}
The gradings on the generators indicate that  
\begin{align}
    \{A_4,A_5\} = & \, \Gamma_{45}^{11111} A_1^5 + \Gamma_{45}^{1112} A_1^3A_2   + \Gamma_{45}^{1113} A_1^3A_3 + \Gamma_{45}^{114} A_1^2A_4 + \Gamma_{45}^{115} A_1^2A_5 + \Gamma_{45}^{122}A_1A_2^2 \nonumber \\
    \{A_4,A_6\} = & \,\Gamma_{46}^{34} A_3A_4 + \Gamma_{46}^{35} A_3A_5   + \Gamma_{46}^{1113} A_1^3A_3 + \Gamma_{46}^{116} A_1^2A_6 + \Gamma_{46}^{221} A_2^2A_1 + \Gamma_{46}^{331} A_3^2A_1  + \Gamma_{46}^{123} A_1A_2A_3 \label{eq:4.16} \\
    \{A_5,A_6\} = & \,  \Gamma_{56}^{34} A_3A_4 + \Gamma_{56}^{35} A_3A_5   + \Gamma_{56}^{1113} A_1^3A_3 + \Gamma_{56}^{116} A_1^2A_6 + \Gamma_{56}^{221} A_2^2A_1 + \Gamma_{56}^{331} A_3^2A_1  + \Gamma_{56}^{123} A_1A_2A_3. \nonumber
\end{align}
with $ \Gamma_{45}^{11111},\ldots,\Gamma_{45}^{122},\Gamma_{46}^{34},\ldots,\Gamma_{46}^{123}, \Gamma_{56}^{34},\ldots,\Gamma_{56}^{123} \in \mathbb{R}$ arbitrary constants.  Using the preceding argument, we immediately conclude that $\Gamma_{45}^{11111}= \Gamma_{46}^{1113} = \Gamma_{56}^{1113} = 0.$ Again, the rest of the coefficients can also be obtained by analyzing the commutator relations in Table $\ref{tab:sl3_commutator}$. 

After applying the grading method, the terms in non-trivial Poisson brackets \eqref{eq:comma2}, \eqref{eq:a3aj} and \eqref{eq:4.16} can be written in compact forms as 
\begin{align}
    \{\textbf{q}_1,\textbf{q}_j\} \sim & \, 0, \text{ for all } 1 \leq j \leq 6, \nonumber \\
    \{\textbf{q}_2,\textbf{q}_2 \} \sim & \, \textbf{q}_1\textbf{q}_2 \label{eq:4.17} \\
    \{\textbf{q}_2,\textbf{q}_3 \} \sim & \, \textbf{q}_1^2\textbf{q}_2 + \textbf{q}_2^2 + \textbf{q}_1 \textbf{q}_3 \nonumber  \\
       \{\textbf{q}_3,\textbf{q}_3 \} \sim & \, \textbf{q}_1^3\textbf{q}_2 + \textbf{q}_1\textbf{q}_2^2 +  \textbf{q}_2\textbf{q}_3. \nonumber
 \end{align} Similar to Table \ref{tab:my_label0} in Subsection \ref{4.1.1}, to illustrate the effectiveness of the grading method, we now provide the comparison of the number of components in non-trivial brackets in the following Table \ref{tab:my_label1}: 

\qquad
 
  \begin{table}[h]
     \centering
     \begin{tabular}{|c|c|c|c|}
     \hline
        \textsf{ Poisson brackets}    & $\begin{matrix}
         \textsf{   No. of polynomials without} \\
         \textsf{  using the grading method }
      \end{matrix}$ & $\begin{matrix}
          \textsf{Maximum No. of polynomials after applying} \\
          \textsf{ the grading method} 
      \end{matrix}$     \\
          \hline 
  $\{\textbf{q}_1,\textbf{q}_2\}$       & $ 3 $ & $0 $    \\
   \hline 
  $\{\textbf{q}_2,\textbf{q}_2\}$       & $6 $ & $2 $   \\
          \hline 
 $\{\textbf{q}_2,\textbf{q}_3\}$       & $9 $ & $5 $   \\
          \hline 
    $\{\textbf{q}_3,\textbf{q}_3\}$       & $13 $ & $7 $   \\
          \hline       
     \end{tabular}
     
     \quad
     
     \caption{Comparison of the number of polynomials }

     \label{tab:my_label1}
 \end{table}
 
 We infer that the polynomial algebra derived from the Poisson centralizer $S(\mathfrak{sl}(3))^{\mathfrak{o}(3)}$ remains closed, as expected. Moreover, taking into account the Poisson relations of $\mathfrak{sl}^*(3,\mathbb{C})$ mentioned earlier, we conclude that the generators are closed under the non-trivial Poisson bracket $\{\cdot,\cdot\}$ in the subsequent form
 \begin{align}
    \{A_2,A_4\} = & \, -3 A_1A_4 , \text{ } \{A_2,A_5\} =  3 A_1A_5   \nonumber \\
    \{A_4,A_5\} = & \,-\frac{1}{3} \left( A_1^3A_3 - A_1A_2^2\right).
\end{align}
As $A_1,\ldots,A_6$ are functionally dependent, together with one algebraic relation, we determine that the algebra $\textbf{Alg} \left\langle\textbf{Q}_3 \right\rangle$ is finitely generated, which defines a polynomial Poisson algebra $\mathcal{Q}_{\mathfrak{sl}(3,\mathbb{C})}(2)$ endowed with the Poisson-Lie bracket $\{\cdot,\cdot\}$. It is worth noting that $A_1$, $A_3$ and $A_6$ form the center of this algebra. In other terms, $\mathcal{Q}_{\mathfrak{sl}(3,\mathbb{C})}(2)$ takes the form of a finitely generated quadratic Poisson algebra over $\mathbb{C}[A_1,A_3,A_6]$.

\subsection{Reduction chain $\mathfrak{h} \subset \mathfrak{sl}(3,\mathbb{C})$}

\label{4.3}

 In \cite{campoamor2023algebraic}, the commutant associated with the Cartan subalgebra of semi-simple Lie algebras of type $A_n$ was analyzed in detail. In this section, we apply the grading method to reconstruct the polynomial algebra in $S(\mathfrak{sl}(3,\mathbb{C}))^\mathfrak{h}.$ The indecomposable polynomial solutions of $\{\mathfrak{h}^*,S(\mathfrak{sl}(3,\mathbb{C}))\} = 0$ form the finite set of polynomials as follows:
 $$  \textbf{Q}_3 = \{h_1, h_2, p_{1,2},p_{1,3}, p_{2,3}, p_{1,2,3 },p_{1,3,2}\} = \textbf{q}_1 \sqcup \textbf{q}_2 \sqcup \textbf{q}_3. $$ Here $ \textbf{q}_1 = \{h_1,h_2\}$, $ \textbf{q}_2 = \{p_{1,2},p_{1,3}, p_{2,3}\}$ and $\textbf{q}_3 = \{p_{1,2,3 },p_{1,3,2}\},$ where $p_{i,j} = e_{ij}e_{ji}$ and $p_{i,j,k}=e_{ij} e_{jk}e_{ki}$ for any $ 1 \leq i \neq j \neq k \leq 3.$  Note that $p_{i,j,k} = p_{j,k,i} = p_{k,i,j}.$ Let us remark that, in this Subsection \ref{4.3}, since each generator is a homogeneous polynomial, we will not relabel it by another letter such as $A_1 = h_1$, $A_2= h_2$, etc. Thus, the expression of the coefficients will be different from that in \eqref{eq:comm}.  Without involving the grading, the Poisson brackets in the compact form are closed as follows
\begin{align}
     \begin{matrix}
           \{\textbf{q}_2,\textbf{q}_2\} \sim \textbf{q}_3 + \textbf{q}_1 \textbf{q}_2 + \textbf{q}_1^3 \\
        \{\textbf{q}_2,\textbf{q}_3\} \sim    \textbf{q}_2^2 + \textbf{q}_1\textbf{q}_3 + \textbf{q}_1^2\textbf{q}_2 + \textbf{q}_1^4 \\
           \{\textbf{q}_3,\textbf{q}_3\} \sim \textbf{q}_2\textbf{q}_3  + \textbf{q}_1  \textbf{q}_2^2 + \textbf{q}_1^2\textbf{q}_3 + \textbf{q}_1^3\textbf{q}_2 + \textbf{q}_1^5.
     \end{matrix} \label{eq:ma67}
 \end{align} %Note that $h_1,h_2$ are central element so that we do not consider the compact form $\{\textbf{q}_1,\textbf{q}_j\}$ for all $j =1,2,3$.

Recall that the $3$-graded Poisson algebra $\mathcal{Q}_{\mathfrak{sl}(3,\mathbb{C})}(2)$ in $S(\mathfrak{sl}(3,\mathbb{C}))^\mathfrak{h}$ is given in \cite{campoamor2023algebraic}. Instead of using direct computation, we apply the grading method to restrict the number of monomials allowed in each commutator relation from $\eqref{eq:ma67}$. Starting with the grading of each term, we have \begin{align}
\begin{matrix}
     \mathcal{G} (h_1) =   \,\mathcal{G}(h_2) = (1,0,0),  \\
    \mathcal{G} (p_{1,2}) =  \,\mathcal{G}(p_{1,3}) = \mathcal{G}(p_{2,3}) = (0,1,1), \\
    \mathcal{G} (p_{1,2,3 }) =  \, (0,2,1), \text{ } \quad \mathcal{G} (p_{1,3,2}) = (0,1,2).
\end{matrix} \label{eq:grad3}
\end{align}

By definition, we omit the grading of $ \{\textbf{q}_1,\textbf{q}_u\}  $ as $\{\textbf{q}_1,\textbf{q}_u\} = 0$ for all $1 \leq u \leq 3.$ Starting with the generators in $\textbf{q}_2$, using Lemma $\ref{3.9}$ (ii), we deduce that \begin{align}
    \mathcal{G} \left(\{p_{i,j},p_{j,k}\}\right) = & \, (0,1,2) \tilde{+} (0,2,1) , \nonumber   \\
     \mathcal{G} \left( \{p_{1,2},p_{1,2,3 }\}\right) =  & \,    (1,2,1) \tilde{+} (0,3,1) \tilde{+}(0,2,2), \label{eq:4.19}\\
      \mathcal{G} \left( \{p_{1,2},p_{1,3,2}\}\right)=& \,      (0,2,2) \tilde{+}(1,1,2) \tilde{+} (0,1,3). \nonumber
\end{align} Here $ 1 \leq i \neq j \neq k \leq 3.$ We remark that  \begin{align}
     \mathcal{G} \left( \{p_{1,2},p_{1,2,3 }\}\right) & =   \mathcal{G} \left( \{p_{1,3},p_{1,2,3 }\}\right) =    \mathcal{G} \left( \{p_{2,3},p_{1,2,3 }\}\right) , \label{eq:4.21}\\
     \mathcal{G} \left( \{p_{1,2},p_{1,3,2}\}\right) & =    \mathcal{G} \left( \{p_{1,3},p_{1,3,2}\}\right) =  \mathcal{G} \left( \{p_{2,3},p_{1,3,2}\}\right).
\end{align}  To help the identification of all the components in each of the Poisson brackets, we will now list all the permissible polynomials from each grading in $\eqref{eq:4.19}$ and $\eqref{eq:4.21}$ case by case as follows: \begin{align}
    (0,1,2) = & \, \{p_{1,3,2}\} , \text{ } (0,2,1) = \{p_{1,2,3}\}; \nonumber \\
    (1,2,1) = & \, \{h_1 p_{1,2,3}, h_2 p_{1,2,3}\}; \text{ } \text{ } (0,3,1) = \emptyset; \label{eq:4.24}\\
    (1,1,2) = & \, \{h_1 p_{1,3,2}, h_2 p_{1,3,2}\} ; \text{ } \text{ } (0,1,3) = \emptyset; \nonumber \\
     (0,2,2) = & \,  \{p_{1,2}p_{1,3}, p_{1,2}p_{2,3}, p_{1,3}p_{2,3}\} . \nonumber
\end{align} Using $\eqref{eq:4.24}$, we deduce that the allowed components in each non-trivial Poisson bracket are given as follows: \begin{align*}
\{p_{1,2},p_{2,3}\} = & \,  a_1 p_{1,2,3} + a_2 p_{1,3,2}  \\
\{p_{1,2},p_{1,3}\} =  & \, a_3 p_{1,2,3} + a_4 p_{1,3,2}   \\
\{p_{1,3},p_{2,3}\} =  & \, a_5 p_{1,2,3} + a_6 p_{1,3,2}   \\
\{p_{1,2},p_{1,2,3}\} =  & \, \left( b_1 h_1  + b_2 h_2 \right) p_{1,2,3} + c_1 p_{1,2} p_{1,3}  + c_2 p_{1,2} p_{2,3} + c_3 p_{1,3} p_{2,3}  \\
\{p_{1,3},p_{1,2,3}\} =  & \, \left( b_3 h_1  + b_4 h_2 \right) p_{1,2,3} + c_4 p_{1,2} p_{1,3}  + c_5 p_{1,2} p_{2,3} + c_6 p_{1,3} p_{2,3}  \\
\{p_{1,2},p_{1,2,3}\} =  & \, \left( b_5 h_1  + b_6 h_2 \right) p_{1,2,3} + c_7 p_{2,3} p_{1,3}  + c_8 p_{1,2} p_{2,3} + c_9 p_{1,3} p_{2,3}  \\
\{p_{1,2},p_{1,3,2}\} =  & \, \left( b_7 h_1  + b_8 h_2 \right) p_{1,2,3} + c_{10} p_{1,2} p_{1,3}  + c_{11} p_{1,2} p_{2,3} + c_{12} p_{1,3} p_{2,3}  \\
\{p_{1,2},p_{1,3,2}\} =  & \, \left( b_9 h_1  + b_{10} h_2 \right) p_{1,2,3} + c_{13} p_{1,2} p_{1,3}  + c_{14} p_{1,2} p_{2,3} + c_{15} p_{1,3} p_{2,3}  \\
\{p_{1,2},p_{1,3,2}\} =  & \, \left( b_{11} h_1  + b_{12} h_2 \right) p_{1,2,3} + c_{16} p_{1,2} p_{1,3}  + c_{17} p_{1,2} p_{2,3} + c_{18} p_{1,3} p_{2,3}.
\end{align*}  In this context, $a_1,\ldots,a_6,b_1,\ldots,b_{12},c_1,\ldots,c_{18}$ represent arbitrary coefficients.  Taking into account the Poisson relation provided in Table $\ref{tab:sl3_commutator}$, we are able to present the explicit expansion in each Poisson bracket. For example, the Poisson bracket $\left\{p_{1,2}, p_{1,2,3}\right\}$, developed via the Poisson relations $\{e_{12},e_{23}\}$ and $\{e_{12},e_{31}\}$, includes only the terms $p_{1,2} p_{1,3}$ and $p_{1,2} p_{2,3}$. Finally, for the term in $\textbf{q}_3,$ we have  $\mathcal{G} \left( \{p_{1,3,2},p_{1,2,3 }\}\right) =  (1,2,2)  \tilde{+} (0,3,2) \tilde{+} (0,2,3).$  Similar to the analysis above, the allowed polynomials in each homogeneous gradings in $\mathcal{G} \left( \{p_{1,3,2},p_{1,2,3 }\}\right) $ are given by    \begin{align}
\nonumber %\hskip 0.5cm
   (1,2,2)  = & \, \{h_1p_{1,2}p_{1,3}, h_1 p_{1,2}p_{2,3}, h_1p_{1,3}p_{2,3}, h_2p_{1,2}p_{1,3}, h_2p_{1,2}p_{2,3}, h_2p_{1,3}p_{2,3}\} ; \\
    (0,3,2) = & \, \{ p_{1,2,3} p_{1,2}, p_{1,2,3} p_{1,3} , p_{1,2,3} p_{2,3} \}; \label{eq:4.25} \\
    (0,2,3) =& \, \{ p_{1,3,2} p_{1,2}, p_{1,3,2} p_{1,3} , p_{1,3,2} p_{2,3} \}. \nonumber
\end{align} Then \begin{align*}
     \{p_{1,3,2},p_{1,2,3 } \}= & \, \left(d_1  p_{1,2}p_{1,3}+ d_2  p_{1,2}p_{2,3}+ d_3  p_{1,3}p_{2,3}\right)h_1 + \left(d_4  p_{1,2}p_{1,3}+ d_5  p_{1,2}p_{2,3}+ d_6 p_{1,3}p_{2,3} \right)h_2\\
     & \, + \left(  e_1 p_{1,2} + e_2  p_{1,3} +e_3  p_{2,3} \right) p_{1,2,3}  + \left( f_1   p_{1,2} + f_2   p_{1,3}+ f_3  p_{2,3}\right)p_{1,3,2} ,
\end{align*}
where $d_1,\ldots, d_6, e_1,e_2,e_3,f_1,f_2,f_3$ are constants. In this discussion, we will refrain from explicitly stating the intricate forms of these non-trivial brackets, as our intention is to employ the root system of $\mathfrak{sl}(n+1,\mathbb{C})$ in Section \ref{5}. This approach will enable us to systematically simplify and further reduce the complexity of the components involved. 

Summing up, we observe that the closed form in $\eqref{eq:ma67}$ becomes
\begin{align*}
     \begin{matrix}
         \{\textbf{q}_2,\textbf{q}_2\} \sim \textbf{q}_3 \\
            \{\textbf{q}_2,\textbf{q}_3\} \sim    \textbf{q}_2^2 +\textbf{q}_1\textbf{q}_3 \\
           \{\textbf{q}_3,\textbf{q}_3\} \sim \textbf{q}_2\textbf{q}_3 + \textbf{q}_1  \textbf{q}_2^2 .
     \end{matrix}
 \end{align*}   We provide the comparison on how the grading reduces the number of components with the one in the compact form \eqref{eq:ma67}: 
 
 \begin{table}[h]
     \centering
     \begin{tabular}{|c|c|c|c|}
     \hline
      \textsf{ Poisson brackets}   & $\begin{matrix}
         \textsf{   No. of polynomials without} \\
         \textsf{  using the grading method }
      \end{matrix}$  & $\begin{matrix}
          \textsf{Maximum No. of polynomials after applying} \\
          \textsf{ the grading method} 
      \end{matrix}$    \\
          \hline 
  $\{\textbf{q}_2,\textbf{q}_2\}$       & $ 12 $ & $2 $  \\
          \hline 
 $\{\textbf{q}_2,\textbf{q}_3\}$       & $24 $ & $6 $ \\
          \hline 
    $\{\textbf{q}_3,\textbf{q}_3\}$       & $42 $ & $12 $ \\
          \hline       
     \end{tabular}

     \quad
     
     \caption{Comparison of the number of polynomials }
     \label{tab:my_label}
 \end{table} 
 
 We observe that in Table \ref{tab:my_label}, even in the worst case, a large number of polynomials is eliminated by the grading method. Compared with the illustrative examples provided in Subsections \ref{4.1} and \ref{4.2}, the effectiveness of the grading method is demonstrated.

Building upon the work in \cite{campoamor2023algebraic}, there exists a linear isomorphism characterized by the following basis transformation: \begin{align*}
      \quad\text{ }  \begin{matrix}
        c_1 = \frac{1}{3}(2h_1+h_2),\;  c_2 = \frac{1}{3}(h_2 -h_1),\;  c_3 = - \frac{1}{3}(h_1 + 2 h_2),\; c_{ij} = p_{i,j}, \\
        f_{123} = \frac{1}{2}(p_{1,3,2}-p_{1,2,3}),\quad g_{123} = \frac{1}{2}(p_{1,3,2} + p_{1,2,3}).
    \end{matrix}
\end{align*} %where $R(3)$ is identified as a Racah quadratic algebra. For an in-depth definition of a Racah algebra, we refer to \cite{MR4256391}.
In this basis, the polynomials display different symmetry/antisymmetry properties. Our goal is to ascertain the grading associated with brackets in this Racah-type algebra. Here, $\tilde{\textbf{Q}}_3  := \tilde{\textbf{q}}_1 \sqcup \tilde{\textbf{q}}_2 \sqcup \tilde{\textbf{q}}_3 $ with $\tilde{\textbf{q}}_1 = \{c_1,c_2\}$, $\tilde{\textbf{q}}_2 = \{c_{ij}: 1 \leq i < j \leq 3 \}$, and $\tilde{\textbf{q}}_3 = \{f_{123},g_{123}\}$. The grading for monomials in both $\tilde{\textbf{q}}_1$ and $\tilde{\textbf{q}}_2$ is identical. In contrast, for $\tilde{\textbf{q}}_3$, we determine that $\mathcal{G}(f_{123}) = \mathcal{G}(g_{123}) = (0,1,2) \tilde{+} (0,2,1).$ Consequently, the grading of the Poisson brackets are \begin{align}
\nonumber
    \mathcal{G}\left(\{c_{ij},c_{jk}\}\right)=  &\,\mathcal{G}(f_{ijk}) = (0,1,2) \tilde{+} (0,2,1), \\
    \mathcal{G} \left(\{c_{jk},f_{ijk}\}\right) =& \, \mathcal{G} \left(\{c_{jk},g_{ijk}\}\right) = (0,2,2)   \tilde{+}   (1,2,1)   \tilde{+}   (1,1,2)  , \label{eq:grad33}   \\
    \nonumber
    \mathcal{G} \left(\{f_{ijk},g_{ijk}\}\right)  =  & \, (1,2,2)  \tilde{+} (0,2,3) \tilde{+}  (0,3,2)   
\end{align} with $ 1 \leq i \neq j \neq k \leq 3 .$    The allowed polynomials from each homogeneous gradings defined in \eqref{eq:grad33} are similar to what we had in \ref{eq:4.24} and \eqref{eq:4.25}, therefore we will omit the permissible polynomials in each of the non-trivial Poisson brackets.   %Then the allowed polynomials from $\eqref{eq:grad33}$ are \begin{align}
%\nonumber
 %   \{c_{jk},f_{ijk}\} = & \,  \left(\Xi_{jk,ik}^1  c_{ik} +\Xi_{jk,ij}^1  c_{ij}   \right) c_{jk} + \Xi_{ki,ij}^1 c_{ki}c_{ij} +   \left(\sum_{\ell=1}^2 \Xi_1^\ell h_\ell\right) f_{ijk} +\left( \sum_{\ell=1}^2 \Xi_2^\ell h_\ell \right) g_{ijk};  \\
 %   \{c_{jk},g_{ijk}\} = & \, \left(\Xi_{jk,ik}^2  c_{ik} +\Xi_{jk,ij}^2  c_{ij}   \right) c_{jk} + \Xi_{ki,ij}^2 c_{ki}c_{ij} + \left(\sum_{\ell=1}^2 \Xi_3^\ell h_\ell \right) f_{ijk} + \left(\sum_{\ell=1}^2 \Xi_4^\ell h_\ell \right) g_{ijk} ; \label{eq:comm3}\\
%    \nonumber
 %    \{ f_{ijk} ,g_{ijk}\}=  & \,\sum_{\ell=1}^2\left(\left(\Xi_{jk,ik}^{3,\ell}  c_{ik} +\Xi_{jk,ij}^{3,\ell}  c_{ij}   \right) h_\ell c_{jk} + \Xi_{ki,ij}^{3,\ell} h_\ell c_{ki}c_{ij} \right)    + \sum_{1 \leq s < t \leq 3} \Xi_{st}^1     c_{st} f_{ijk}+ \sum_{1 \leq s < t \leq 3} \Xi_{st}^2 c_{st} g_{ijk} .
%\end{align} Here $\Xi_{jk,st}^1, \ldots, \Xi_{ki,ij}^{3,\ell},\Xi_u,\Xi_{st}^1$ and $\Xi_{st}^2$ are arbitrary constants that can be determined by the commutator relations and solving systems of linear equations. % {\color{blue}The ordered monomials follows from the order in the expected grading.}

\section{Grading of monomials in $\mathcal{Q}_{A_n}(n)$}

\label{5}

 Within Subsection \ref{4.3} we have elucidated, with an illustrative example, the effectiveness of the grading technique when applied to the non-trivial brackets of the Cartan invariant generators within the symmetric algebra $S(\mathfrak{sl}(3,\mathbb{C}))$. In accordance with our exploration, this section will delve into the employment of root systems as a supplementary analytical tool to these generators, facilitating further reductions of the polynomial components within the non-trivial Poisson brackets. Consider the special linear algebra $\mathfrak{sl}(n+1,\mathbb{C}),$ which is a Lie algebra consisting of $(n+1) \times (n+1)$ matrices with trace zero, in its defining representation. In what follows, we denote $\mathfrak{sl}(n+1,\mathbb{C})$ by $A_n$. Let $E_{ij}$ with $1 \leq i,j \leq n+1$ be the generators of $A_n$ subjected to the constraint $\sum_{i=1}^{n+1} E_{ii}= 0$. Note that $  A_n $ admits a triangular decomposition $\mathfrak{h} \oplus\mathfrak{g}^+ \oplus \mathfrak{g}^- $ that also satisfies the commutator relations $\eqref{eq:commutator3},$ where $\mathfrak{h}$ is the Cartan subalgebra, $\mathfrak{g}^+$ consists of all positive root vectors, and $\mathfrak{g}^-$ contains all negative root vectors. In detail, the commutation relations are given by \begin{align}
    [E_{ij},E_{kl}] = \delta_{jk} E_{il} - \delta_{il} E_{kj}  \quad 1 \leq i,j,k,l \leq n +1. \label{eq:relat}
\end{align} In particular, the Cartan subalgebra is determined by \begin{align*}
    [E_{i,i+1} ,E_{i+1,i}] = E_{i,i} - E_{i+1, i+1} = H_i \quad 1 \leq i \leq n .
\end{align*} Let $\mathfrak{sl}^*(n+1,\mathbb{C})$ be the dual space of $A_n$ with the following lexicographically ordered coordinates \begin{align*}
  &  h_i, \hskip 0.86cm 1 \leq i \leq n; \\
   &   e_{i,i+a}, \quad 1 \leq i \leq n , \text{ } 1 \leq a \leq n+1 -i ; \\
& e_{i+a,i}, \quad 1 \leq i \leq n, \text{ } 1 \leq a \leq n+1 -i   . 
\end{align*} 

%Recall that from \cite{MR1920389}, the roots are of the form $e_{ij} = \epsilon_i - \epsilon_j$, where $\epsilon_i$ and $\epsilon_j$ are standard basis vectors in $\mathbb{R}^{n+1}$ and $1 \leq i \neq j \leq n+1,$

In this Section \ref{5}, we focus on the centralizer in $S(A_n)$ with respect to the Cartan subalgebra $\mathfrak{h}.$  It is established that Casimir invariants and commutants relative to Cartan belong to the weight zero space.  From \cite{MR1920389}, the generators of the Cartan commutant can be identified with a $k$-cycle in the symmetric group $S_{n+1}.$ Recall that the Cartan centralizer of $A_n$ is generated by \begin{align}
    \textbf{Q}_{n+1}= \left\{h_1,\ldots,h_n, p_{i_1,i_2},p_{i_1 ,i_2,i_3},\ldots,p_{i_1,\ldots, i_{n+1}}\right\}   \label{eq:finiteindic}
\end{align} with $1 \leq i_1 , i_2 , \ldots , i_{n+1} \leq n  +1.$ Here $ \left|\textbf{Q}_{n+1}\right| =  \sum_{r=1}^{n+1} \frac{(n+1)!}{(n+1-r)!r} -1$ and $p_{i_1,\ldots, i_{n+1} } = e_{i_1i_2} e_{i_2i_3} \cdots e_{i_{n}i_{n+1}} e_{i_{n+1}i_1}$. It had been shown in \cite{campoamor2023algebraic} that, for $n \geq 2$, $\textbf{Alg} \left\langle \textbf{Q}_{n+1}\right\rangle$ is a degree $n $ polynomial algebra, and is closed in the Poisson-Lie bracket $\{\cdot,\cdot\}$ with extra polynomial relations.  In the following, we will denote it as $\mathcal{Q}_{A_n}(n)$. For the sake of completeness, we mention that when $n = 1$ the polynomial algebra is Abelian. In the subsequent sections, we will perform a detailed calculation of the admissible monomials for each potential grading of the non-trivial brackets.

%Suppose that $p_{[ij],ji} \in \mathcal{Q}_{A_n}(n)$. Then we can conclude the number of choices of the term in each equivalent classes $[ij].$  We now try to find the grading of all the polynomial algebras in $U(\mathfrak{sl}(n))^\mathfrak{h}$. Starting with an elementary one.

\subsection{Grading and root systems in $\mathcal{Q}_{A_n}(n)$}

%To count the elements in Poisson brackets of $\mathcal{Q}(d)$, we define $\textit{grading}$ of a polynomial in $\mathcal{S}(\mathfrak{g})^{\mathfrak{h}}$ as $\mathcal{G}:  \mathcal{Q} (d) \rightarrow \mathbb{N} \times\mathbb{N} \times\mathbb{N}  $ given by $ \prod_i^n e_{\alpha_i} \prod_j^r e_{-\alpha_i} \prod_k^m h_k \rightarrow (m,n,r)$, where $m,n,r$ represents the number of elements for $\mathfrak{h}, \mathfrak{g}^+$ and $\mathfrak{g}^-$, respectively, in $p_{i_1 \ldots i_s} $.

We first provide the grading of the generators in $\textbf{Q}_{n+1}$, and then introduce some basic terminology that allows us to obtain the explicit expression of the components in each non-trivial bracket. Using Lemma $\ref{3.9}$ (ii), the following properties are deduced: 

\begin{proposition}
\label{5.1}
     Let $\mathcal{G}$ be a grading of a monomial, and let $p_{i_1, i_2, \dots, i_r} $ be a generator in $\mathcal{Q}_{A_n}(n)$. Then $\mathcal{G} (p_{i_1, i_2  \dots, i_r }) = \left\{(0,n_+,n_-) : n_+ + n_- = r, \text{ } n_\pm \neq r\right\} $ for any $2 \leq r \leq n+1.$ In particular, suppose that $p \in \textbf{q}_k$ with $\mathcal{G}(p) = \left(0,n_+^{(k)},n_-^{(k)}\right)$ and $q \in \textbf{q}_l$ with $\mathcal{G}(q) = \left(0,n_+^{(l)},n_-^{(l)}\right)$. Then \begin{align}
        \mathcal{G}\left(\{\textbf{q}_k,\textbf{q}_l\}\right) = & \, \left(1,n_+^{(k)} + n_+^{(l)} -1,n_-^{(k)} + n_-^{(l)}-1\right) \tilde{+}  \left(0,n_+^{(k)} + n_+^{(l)}  ,n_-^{(k)}   + n_-^{(l)}-1 \right)   \tilde{+} \left(0,n_+^{(k)} + n_+^{(l)} -1,n_-^{(k)} + n_-^{(l)}  \right). \label{eq:comm33}
    \end{align}
\end{proposition}
\begin{remark}
    The notation $n_+^{(k)}$ denotes the number of positive root vectors in a degree $k$ generator $p$ in the set $\textbf{q}_k,$ and $n_-^{(k)}$ represents the number of negative root vectors in $p.$ Similar explanations hold for the notations $n_+^{(l)}$ and $n_-^{(l)}$.
\end{remark}

\begin{proof}
% Define the partition of indices by
% \begin{align}
 %   \mathcal{S}^+ = \left\{(i_k,j_k): i_k < j_k, \text{ }  1 \leq i_k,j_k \leq n  +1\right\}  \text{ and }  \mathcal{S}^- = \left\{(i_k,j_k): i_k > j_k, \text{ }  1 \leq  i_k,j_k \leq n +1 \right\}. \label{eq:indices}
%\end{align}
%The ordered index pairs in $\eqref{eq:indices}$ distinguish the positive or negative roots that appear in $\Phi$.  
Starting with $\deg p = 2$, by definition, we have $\mathcal{G}\left(p_{i_1 ,i_2}\right) = (0,1,1) $. We now assume that $\deg p \geq 3$. We then have \begin{align*}
    \mathcal{G}\left(p_{i_1 ,i_2,  i_3}\right) =& \,\{(0,1,2) , (0,2,1)\}, \\
    \mathcal{G}\left(p_{i_1, i_2,i_3 ,i_4  }\right) = & \, \{(0,1,3),(0,2,2),(0,3,1)\} , \\
    & \vdots \\
    \mathcal{G}\left(p_{i_1,\ldots,i_r}\right) =& \, \left\{(0,n_+,n_-) : n_+ + n_- = r , \text{ } n_\pm \neq r\right\}
\end{align*}
with all $1 \leq i_1 ,\ldots,i_r \leq n +1 ,$ where $n_+  $ is the number of positive roots  and $n_-  $ is the number of negative roots in $p_{i_1,i_2,\ldots,i_r}$ respectively. For the second part, the grading of $\{\textbf{q}_k,\textbf{q}_l\}$ follows directly by applying Lemma $\ref{3.9}$ part (ii). Hence, we omit the details here.
\end{proof}

\begin{corollary}
    For any $k,l \in \{1,\ldots,n+1\}$, suppose that $\mathcal{G} \left(\{\textbf{q}_k,\textbf{q}_l\}\right) = (t, n_+,n_-)$ such that $ t + n_+ + n_- = \deg\{\textbf{q}_k,\textbf{q}_l\}. $ Then $ t = 0,1.$
\end{corollary}

It is important to note, as inferred from the outcome presented in Subsection $\ref{4.3},$ that a vast array of allowable generators exist for each homogeneous degree in equation $\eqref{eq:comm33}$. Our goal is to further decrease the number of generators within each degree of the Poisson brackets. In the forthcoming analysis, for any $p = \epsilon_{\beta_1} \cdots \epsilon_{\beta_r} 
 \in S(\mathfrak{g}^+ \oplus \mathfrak{g}^-)$, we denote the roots corresponding to each root vector contained in $p$ by $R(\epsilon_{\beta_j}) = \beta_j.$ Then $R(p) = \beta_1 + \ldots + \beta_r$. This notation allows us to systematically address the relations between the generators and the associated roots in their respective gradings. Moreover, from the terminologies in \cite{campoamor2024superintegrable}, we observe that $p  \in S(\mathfrak{g})^\mathfrak{h} $ is linearly independent and indecomposable if and only if there exists a root $\beta_k$ such that $ \mathrm{lgh}(\beta_k) = \max_{1 \leq j \leq r} \left\{\mathrm{lgh}(\beta_j)\right\}$ and $R(p) =\sum_{j=1}^r \beta_j = 0$, where $\mathrm{lgh}:\Phi \rightarrow \mathbb{N}_0$ is the length of a root in $\Phi$. For the rest of this section, let $\Phi_{A_n} : = \Phi$ be the root system of type $A_n.$

In this Subsection \ref{5.1}, with the help of the root system of $A_n,$ we will provide a classification of the allowed components appearing in the expansions of the Poisson brackets of a certain degree. For any non-zero generators $p, q \in \mathcal{Q}_{A_n}(n)$ consider, without loss of generality, that $p = \prod_{i \geq 1}^t e_{\beta_i}$ and $q = \prod_{j \geq 1}^r e_{\gamma_j}$ with $ 1\leq t,r \leq n+1$.  Here $  \beta_1 + \ldots +  \beta_t = 0$ and $ \gamma_1 + \ldots  +\gamma_r = 0 $. Using Leibniz's rule, we have

\begin{align}
\left\{\prod_{i \geq 1}^t e_{\beta_i}, \prod_{j \geq 1}^r e_{\gamma_j}\right\} & = \sum_{\begin{small}
\begin{matrix}
1 \leq i_b \leq t \\
1 \leq j_c \leq r
\end{matrix}
\end{small}} \left\{e_{\beta_{i_b}}, e_{\gamma_{j_c}}\right\} \prod_{i \ne i_b} e_{\beta_{i_b}} \prod_{j \ne j_c} e_{\gamma_{j_c}}. \label{eq:mulro}
\end{align}
Taking into account the expression on the right-hand side of $\eqref{eq:mulro}$ and the relations among the roots in the root decomposition, our focus narrows to the case where $\beta_{i_b} + \gamma_{j_c} \in \Phi$ for every pair $i_b , j_c$, as $\left\{e_{\beta_{i_b}}, e_{\gamma_{j_c}}\right\} = 0$ whenever $\beta_{i_b} + \gamma_{j_c}$ is not in $\Phi$. This fact motivates the following definition.

\begin{definition}
    For any $\alpha, \beta \in \Phi,$ we say that they are $\textit{connected}$ if $\alpha+ \beta \in \Phi$. We denote the connected roots by $(\alpha,\beta).$
\end{definition}

%{\rr [Remarque: "non-zero" root is inappropriate, as by definition $0$ never belongs to a root system. In the following, I erase "non-zero".]}

We now present several key observations that assist in identifying the components within $\{p,q\}$.

\begin{proposition}
 \label{root}
  Let $\Phi$ be the root system of $A_n$.  The following properties hold:
  
(i) For any roots $\alpha_1,\alpha_2 \in \Phi$, assume that $(\alpha_1,\alpha_2)$ is connected. We further assume that there exists a root $\beta \in \Phi$ such that both $(\alpha_1 ,\beta )$ and $(\alpha_2 ,\beta )$ are connected. Then $\alpha_1 + \alpha_2 + \beta = 0. $

(ii) Let $\alpha,\beta \in \Phi$ with $\alpha + \beta = 0.$ Then there does not exist a $\beta' \neq \beta$ in $\Phi$ such that $(\alpha,\beta')$ is connected.
\end{proposition}
\begin{proof}
Given that roots $\alpha_1,\alpha_2,\beta \in \Phi$ and $(\alpha_1,\alpha_2)$ are connected, without loss of generality, assume that $\alpha_1 = \epsilon_{j_1} - \epsilon_{j_2},$ $\alpha_2 = \epsilon_{j_2} - \epsilon_{j_3}$ and $\beta = \epsilon_{i_k} - \epsilon_{i_{k+1}}$ for some $k$. Since both $(\alpha_1,\beta)$ and $(\alpha_2,\beta)$ are connected, by definition,
\begin{align*}
    \alpha_1 + \beta  = & \, \epsilon_{j_1} -\epsilon_{j_2} + \epsilon_{i_k} - \epsilon_{i_{k+1}} = \left\{\begin{matrix}
        \epsilon_{j_1}  - \epsilon_{i_{k+1}} & \text{ if } \epsilon_{j_2} =  \epsilon_{i_k} \\
         \epsilon_{i_k}-\epsilon_{j_2}  & \hskip 0.3cm \text{ if } \epsilon_{j_1}  = \epsilon_{i_{k+1}}
    \end{matrix}\right. \\
   \alpha_2 + \beta = & \, \epsilon_{j_2} -\epsilon_{j_3} + \epsilon_{i_k} - \epsilon_{i_{k+1}} = \left\{\begin{matrix}
        \epsilon_{j_2}  - \epsilon_{i_{k+1}} & \hskip -0.1cm \text{ if } \epsilon_{j_3} =  \epsilon_{i_k} \\
         \epsilon_{i_k}-\epsilon_{j_3}  & \hskip 0.3cm \text{ if } \epsilon_{j_2}  = \epsilon_{i_{k+1}} .
    \end{matrix}\right.
\end{align*} This implies that
\begin{center}
    (a) $\epsilon_{j_1} = \epsilon_{i_{k+1}}$ and $\epsilon_{j_3} = \epsilon_{i_k};$ \quad  (b) $\epsilon_{j_1} = \epsilon_{i_{k+1}}$ and $\epsilon_{j_2} = \epsilon_{i_{k+1}};$ \quad(c) $\epsilon_{j_2} = \epsilon_{i_k}$ and $\epsilon_{j_3} = \epsilon_{i_k}.$
\end{center}
It is clear that cases (b) and (c) need to be discarded, as there is no zero root. Hence $\alpha_1 + \alpha_2 + \beta = 0.$

Now we process part (ii). Without loss of generality, let $\beta = \epsilon_{i_k} - \epsilon_{i_{k+1}}$ and $\beta' = \epsilon_{i_j} - \epsilon_{i_{j+1}}$ with $j \neq k.$  Suppose, by contradiction, that $(\alpha,\beta')$ is connected. Since $\alpha+ \beta = 0,$ $\alpha = -\beta.$ Then \begin{align*}
    \alpha  + \beta ' = \epsilon_{i_j} - \epsilon_{i_{j+1}}  +   \epsilon_{i_{k+1}} -\epsilon_{i_k}= \left\{\begin{matrix}
        \epsilon_{i_j}  - \epsilon_{i_k} & \hskip 0.55cm \text{ if } \epsilon_{i_{j+1}}  = \epsilon_{i_{k+1}} \\
         \epsilon_{i_{k+1}}-\epsilon_{i_{j+1}}  & \text{ if } \epsilon_{i_j} =  \epsilon_{i_k}.
    \end{matrix}\right.
\end{align*} This implies that $j = k,$ which is a contradiction.
\end{proof}

\begin{proposition}
\label{rootc}
      For any monomial $p \in \textbf{Q}_{n+1}$ defined in $\eqref{eq:finiteindic}$, let $J_r := \{\beta_1,\ldots,\beta_r: \beta_1 + \ldots + \beta_r = 0\}  $ be the set consisting of roots in $p. $  Then, for any $\alpha \notin \Phi/J_r$, there are at most two distinct roots $\beta_\ell \neq \beta_l$ in $J_r$ with $1 \leq l \neq \ell \leq r$ such that $(\alpha,\beta_\ell)$ and $(\alpha,\beta_l)$ are connected. Here $\Phi/J_r$ means that we exclude the set $J_r$ from $\Phi.$
\end{proposition}

\begin{proof}
 Suppose that the statement does not hold. That is, for any $\alpha \in \Phi/J_r$ we can assume that $(\alpha,\beta_\ell),(\alpha,\beta_l) $ and $ (\alpha,\beta_s)$ are connected with fixed $\ell \neq l \neq s   \in \{1,\ldots,r\}$. Without loss of generality, assume that $\alpha = \epsilon_i - \epsilon_j, \beta_\ell = \epsilon_{i_\ell} - \epsilon_{i_{\ell+1}}, \beta_l = \epsilon_{i_l} - \epsilon_{i_{l+1}}$ and $\beta_s = \epsilon_{i_s} - \epsilon_{i_{s+1}}$. By definition,
 \begin{align}
\nonumber
    \alpha+ \beta_\ell = & \, \epsilon_i - \epsilon_j + \epsilon_{i_\ell} - \epsilon_{i_{\ell+1}} = \left\{\begin{matrix}
         \epsilon_{i_\ell}  - \epsilon_j  & \text{ if } \epsilon_i = \epsilon_{i_{\ell+1}} \\
      \epsilon_i    - \epsilon_{i_{\ell+1}} & \hskip -0.2cm \text{ if } \epsilon_j = \epsilon_{i_\ell}
    \end{matrix}\right.  \\
    \alpha+ \beta_l = & \, \epsilon_i - \epsilon_j + \epsilon_{i_l} - \epsilon_{i_{l+1}} = \left\{\begin{matrix}
         \epsilon_{i_l}  - \epsilon_j  & \hskip  0.1cm \text{ if } \epsilon_i = \epsilon_{i_{l+1}} \\
      \epsilon_i    - \epsilon_{i_{l+1}} & \hskip  -0.1cm \text{ if } \epsilon_j = \epsilon_{i_l}
    \end{matrix}\right. \label{eq:classification0} \\
    \nonumber
    \alpha + \beta_s =  & \epsilon_i - \epsilon_j + \epsilon_{i_s} - \epsilon_{i_{s+1}} = \left\{\begin{matrix}
         \epsilon_{i_s}  - \epsilon_j  & \hskip  0.1cm \text{ if } \epsilon_i = \epsilon_{i_{s+1}} \\
      \epsilon_i    - \epsilon_{i_{s+1}} & \hskip  -0.10cm\text{ if } \epsilon_j = \epsilon_{i_s}.
    \end{matrix}\right. %\\
  %  \nonumber
  %  \alpha + \beta_t =  & \epsilon_i - \epsilon_j + \epsilon_{i_t} - \epsilon_{i_{t+1}} = \left\{\begin{matrix}
   %      \epsilon_{i_t}  - \epsilon_j  & \text{ if } \epsilon_i %= \epsilon_{i_{t+1}} \\
  %    \epsilon_i    - \epsilon_{i_{t+1}} & \text{ if } \epsilon_j = \epsilon_{i_t}
%    \end{matrix}\right.
\end{align}
We note that determining the specific value of $\alpha$ merely requires two constraints from $\eqref{eq:classification0}$.  For example, given that $(\alpha,\beta_\ell) $ and $(\alpha,\beta_l) $ are connected, we infer that $\alpha = \left\{\begin{matrix}
    \epsilon_{i_{\ell+1}} - \epsilon_{i_l}\\
    \epsilon_{i_{l+1}} - \epsilon_{i_\ell}
\end{matrix}\right.. $  Then \begin{align}
    \alpha + \beta_s=\left\{\begin{matrix}
     \epsilon_{i_{\ell+1}} - \epsilon_{i_l}  + \epsilon_{i_s} - \epsilon_{i_{s+1}}\\
    \epsilon_{i_{l+1}} - \epsilon_{i_\ell} + \epsilon_{i_s} - \epsilon_{i_{s+1}}.
\end{matrix}\right. \label{eq:rootsum}
\end{align}
Since $\alpha + \beta_s \in \Phi,$ from $\eqref{eq:rootsum}$ we deduce that either $\ell = s  $ or $l = s,$ which contradicts our assumption. A similar argument holds if either $(\alpha ,\beta_\ell)$ and $(\alpha,\beta_s)$ are connected or $(\alpha ,\beta_l)$ and $(\alpha,\beta_s)$ are connected.
\end{proof}

\begin{example}
    Suppose that $r = 3$. We provide an example of the maximal number of connected pairs.  As $r = 3,$ then $(\alpha,\beta_j)$ are connected for all $1 \leq j \leq 3.$ From $\eqref{eq:classification0}$, once $\alpha$ is determined, we induce the following cases:

(a) The pairs $(\alpha,\beta_1)$ and $(\alpha,\beta_2)$ are connected, deducing that $\alpha = \epsilon_{i_3} - \epsilon_{i_1} = \beta_3.$ Then $\alpha+ \beta_3 = 2 \beta_3 \notin \Phi  ;$

(b) The pairs $(\alpha,\beta_1)$ and $(\alpha,\beta_3)$ are connected, deducing that $\alpha = \epsilon_{i_2} - \epsilon_{i_3} = \beta_2.$ Then $\alpha+ \beta_2 = 2 \beta_2 \notin \Phi  ;$

(c) The pairs $(\alpha,\beta_2)$ and $(\alpha,\beta_3)$ are connected, deducing that $\alpha = \epsilon_{i_1} - \epsilon_{i_2} = \beta_1.$ Then $\alpha+ \beta_1 = 2 \beta_1 \notin \Phi.$

Hence, not all pairs are connected.
\end{example}

  We now consider an interesting observation.  Suppose that $p,q \in \textbf{q}_t$ are such that $\mathcal{G}(p) = (0,n_+,n_-)$ and $\mathcal{G}(q) = (0,n_-,n_+)$ with $n_- + n_+ = t$.  Using Proposition \ref{5.1}, we deduce that $$\mathcal{G} (\{p,q\}) = \left(1,t -1,t-1\right) \tilde{+}  \left(0,t  ,t-1       \right)   \tilde{+} \left(0,t -1,t   \right). $$ We then have the following proposition.

  \begin{proposition}
  \label{5.8}
      Let $p = e_{\beta_1}\cdots e_{\beta_{t-1}} e_{\beta_t},q = e_{-\beta_1}  \cdots e_{-\beta_{t-1}} e_{-\beta_t} \in \textbf{q}_t$. Then \begin{align*}
          \{p,q\} = \sum_{l=1}^t N_{\beta_l,-\beta_l} \prod_{k \neq l}^t e_{\beta_k} e_{-\beta_k},
      \end{align*} where $N_{\beta_l,-\beta_l}   = \sum_{i=1}^n C_i^l h_i  $ and $h_i  $ is a coordinate function of Cartan generators in $\mathfrak{h}^* $ for each $i.$ Here $C_i^l$ are arbitrary coefficients.
  \end{proposition}

  \begin{proof}
  The proof follows from using \eqref{eq:mulro} by a direct computation.
  %   $   p = e_{i_1i_2}\cdots e_{i_{t-1}i_t} e_{i_t i_1},q = e_{i_1i_t} e_{i_ti_{t-1}} \cdots e_{i_2i_1} \in \textbf{q}_t$
  \end{proof}

\subsection{Explicit polynomials in $\{\textbf{q}_2,\textbf{q}_r\}$}
\label{5.2}

 We refer to Section $ \ref{5.1}$, where it was pointed out that the exact number of monomials derived from the grading of the non-trivial Poisson brackets depends on the Poisson brackets between the individual roots. In this Subsection \ref{5.2}, for any non-zero indecomposable monomial $p  \in \textbf{q}_2, q \in \textbf{q}_r $, we assume that $p = e_{\alpha}e_{-\alpha}$ and $q =    e_{\beta_1} \cdots e_{\beta_r}$ with $R(p) = \alpha + (-\alpha) = 0$ and $R(q) =   \beta_1 + \ldots + \beta_r = 0.$ Here $2 \leq r \leq n+1$ and $\alpha,\beta_1,\ldots,\beta_r \in \Phi $, unless stated otherwise. Recall that $J_r= \{\beta_1,\ldots,\beta_r: \beta_1 + \ldots + \beta_r = 0\}$ is the set consisting of all the roots from $R(q).$ Using Proposition $\ref{5.1},$ we deduce that \begin{align}
\mathcal{G} \left( \{p, q\}\right) = (1,n_+,n_-) \tilde{+} (0,n_++1,n_-) \tilde{+} (0,n_+,n_-+1) \text{ with } n_+ + n_- = r. \label{eq:commt}
\end{align} We aim to determine all the allowed polynomials from the grading in $\eqref{eq:commt}.$  From $\eqref{eq:commt},$ we can construct the components in the Poisson brackets by considering two blocks of the gradings: $  (1,n_+,n_-)$ and $(0,n_++1,n_-) \tilde{+} (0,n_+,n_-+1)$ separately. We first consider the permissible polynomials from the homogeneous grading $(1,n_+,n_-)$. By definition, they are decomposable into a Cartan generator and a monomial in $\textbf{q}_{r}$. On the other hand, a direct computation shows that
\begin{align}
    \left\{p,q\right\}  = \sum_{j =1}^r N_{\alpha,\beta_j} e_{-\alpha} \prod_{k\neq j} e_{\beta_k} + \sum_{j =1}^r N_{-\alpha,\beta_j} e_\alpha \prod_{k\neq j} e_{\beta_k}, \label{eq:commu2r}
\end{align}
where $N_{\alpha,\beta_j} = C_{\alpha,\beta_j} e_{\alpha+ \beta_j}$. Here, $C_{\alpha,\beta}$ are the structure constants in the commutator relations of $A_n$. In particular, the grading of \eqref{eq:commu2r} is as follows
\begin{align*}
    \mathcal{G} \left(\left\{e_{\alpha} e_{-\alpha}, e_{\beta_1} \cdots e_{\beta_r}\right\} \right) = &\underbrace{\sum_{\alpha + \beta_j \in \Phi} \left( \mathcal{G}\left(\{e_{\alpha},e_{\beta_j}\}e_{-\alpha} \prod_{k \neq j}e_{\beta_k}  \right) +  \mathcal{G}\left(\{e_{-\alpha},e_{\beta_j}\}e_{\alpha} \prod_{k \neq j}e_{\beta_k}  \right)\right)}_{\text{$\in (0,n_++1,n_-) \tilde{+} (0,n_+,n_-+1)$}}  \\
    &+ \underbrace{\sum_{\alpha + \beta_j =0 }   \mathcal{G}\left(\{e_{\alpha},e_{\beta_j}\}e_{-\alpha} \prod_{k \neq j}e_{\beta_k}  \right)   }_{\text{$ \in (1,n_+,n_-)$}}.
\end{align*}
As presented in the argument of Proposition \ref{rootc}, in the rest of this section, we assume that $\alpha = \epsilon_i - \epsilon_j,$ $\beta_j = \epsilon_{i_j} - \epsilon_{i_{j+1}}$ and $\beta_r = \epsilon_{i_r} - \epsilon_{i_1},$ where $\epsilon_i,\epsilon_j, \epsilon_{i_j}   \in \mathfrak{h}^* $ and $i,j,i_1,\ldots,i_r \in \{1,\ldots,n\}$.    Therefore, for a more comprehensive analysis of the allowed monomials in $\eqref{eq:commt}$, we should concentrate on categorizing the connectivity properties between the roots $\alpha$ and $\beta_j$.

%From $\eqref{eq:commt},$ we see that the grading of the brackets is divided into two parts. One part contains Cartan elements, and the other part is Cartan-free. Initially, we demonstrate that when $(1,n_+,n_-) \in \mathcal{G} \left( \{p, q\}\right) ,$ each term can be decomposed according to $\eqref{eq:comp}$.

\begin{proposition}
\label{5.5}
  Let $p = e_\alpha e_{-\alpha}\in \textbf{q}_2$ and $q = e_{\beta_1} \cdots e_{\beta_r} \in \textbf{q}_r$ be generators of $\mathcal{Q}_{A_n}(n)$, and let $J_r = \{\beta_1,\ldots,\beta_r: \beta_1 + \ldots + \beta_r = 0\}$ be the set consisting of all the roots in $q.$ Suppose that $(1,n_+,n_-) $ is contained in $\mathcal{G} (\{p,q\})$. Then for a fixed $\beta_\ell  = \epsilon_{i_\ell} - \epsilon_{i_{\ell+1}} \in J_r $ such that $\beta_\ell = -\alpha,$ we have
  \begin{align}
    \{p,q\} = %\left\{\begin{matrix}
     %\Xi_{r,2}  h_\ell  \left( e_{\beta_\ell} \prod_{k\neq \ell}^r e_{\beta_k} \right)  \\
       N_{-\beta_\ell,\beta_\ell}  \left( e_{\beta_\ell} \prod_{k\neq \ell}^r  e_{\beta_k} \right)  +  \left(C_{\beta_{\ell},\beta_{\ell+1}}   \left( e_{\beta_\ell+\beta_{\ell+1}}\prod_{k\neq \ell+1,\ell}^r e_{\beta_k}\right) + C_{\beta_{\ell-1},\beta_{\ell}}  \left(e_{\beta_\ell+\beta_{\ell-1}} \prod_{k\neq \ell-1}^r e_{\beta_k}\right) \right) p_{\ell,-\ell}. \label{eq:comp}%\\
  %  \end{matrix}\right.
\end{align}
Here $C_{\beta_{\ell},\beta_{\ell+1}} $ and $C_{\beta_{\ell-1},\beta_{\ell}} $ are structure constants, $p_{\ell,-\ell} = e_{\beta_{\ell}} e_{-\beta_{\ell}}$ is a degree $2$ generator of $\mathcal{Q}_{A_n}(n)$ and $  N_{-\beta_\ell,\beta_\ell} = \sum_{i=1}^n C_i^\ell h_i $, where $h_i$ are Cartan elements and $C_i^\ell$ are constants.
\end{proposition}
\begin{proof}
  By Proposition \ref{root}, for each $\alpha,$ we can have only one non-adjacent term in $J_r$ such that $\alpha = - \beta_\ell.$   Now, given $\beta_\ell \in J_r$ with $\alpha = -\beta_\ell$, we deduce that \begin{align*}
    N_{\alpha,\beta_\ell} = & \,  \sum_{i=1}^nC_i^\ell h_i \in \mathfrak{h}^*, \text{ } \quad  N_{\alpha,\beta_l} = 0\text{ with $l \neq \ell$} , \\
    N_{-\alpha,\beta_\ell} = & \, 0, \text{ } \quad  N_{-\alpha,\beta_l}= \left\{ \begin{matrix}
        0  &  \text{ with $ |l-\ell | \geq 2$   } \\
         C_{\alpha, \beta_l} \mathfrak{g}_{\alpha+\beta_l}  & \, \,  \text{ with $|l-\ell| \leq 1.$ }
    \end{matrix}\right.
\end{align*}  Back to $\eqref{eq:commu2r},$  we find that $  \sum_{j=1}^r N_{\alpha,\beta_j} e_{-\alpha} \prod_{k\neq j} e_{\beta_k} =     N_{-\beta_\ell,\beta_\ell}  \left( e_{\beta_\ell} \prod_{k\neq \ell}^r e_{\beta_k} \right) ,$ and if $|l - \ell| \leq 1,$ it is clear that $(\beta_\ell,\beta_{\ell+1})$ and $(\beta_{\ell-1},\beta_\ell)$ are connected. Then \begin{align*}
     \sum_{j =1}^r N_{-\alpha,\beta_j} e_\alpha \prod_{k\neq j} e_{\beta_k} =   \,C_{\beta_{\ell},\beta_{\ell+1}}  \left(e_{-\beta_\ell}e_{ \beta_\ell}\right) \left( e_{\beta_\ell+\beta_{\ell+1}}\prod_{k\neq \ell+1,\ell} e_{\beta_k}\right)  + C_{\beta_{\ell-1},\beta_{\ell}} \left( e_{-\beta_\ell} e_{\beta_\ell}\right)\left(e_{\beta_\ell+\beta_{\ell-1}} \prod_{k\neq \ell-1} e_{\beta_k}\right).
\end{align*} 
Summing up all the components, the expansion \eqref{eq:comp} is as required.
\end{proof}

We now look at the Cartan-free part. That is, assume that $(1,n_+,n_-) \notin \mathcal{G}(\{p,q\}) $. Starting with $r = 2$ and $\mathcal{G}(q) = (0,1,1)$, from $\eqref{eq:commt},$ we deduce that $\mathcal{G} \left( \{p, q\}\right) =  (1,1,1) \tilde{+} (0,2,1) \tilde{+} (0,1,2) .$ It is clear that the allowed monomials in the grading $(1,1,1)$ must be decomposable with respect to the generators in $\textbf{Q}_{n+1}.$ On the other hand, a direct computation shows that
\begin{align*}
     \{e_\alpha e_{-\alpha},e_{\beta} e_{-\beta}\} = N_{\alpha,\beta} e_{-\alpha} e_{-\beta} +  N_{-\alpha,\beta} e_{ \alpha} e_{-\beta} +  N_{\alpha, -\beta} e_{-\alpha} e_\beta +  N_{-\alpha,-\beta} e_{\alpha} e_{\beta},
 \end{align*}
 where $N_{\alpha,\beta} = C_{\alpha,\beta} e_{\alpha+ \beta}$.  Note that $N_{\alpha,\beta} = -N_{\beta,\alpha}$ and $C_{\alpha,\beta} =0$ if $\alpha + \beta \notin \Phi .$ In this context, we can examine two cases: $\alpha + \beta = 0$ or $\alpha = \beta.$ Both of these cases result in $p = q,$ which means that $\{p,q\} = 0.$ Therefore, we can deduce that $\mathcal{G} \left( \{p, q\}\right) =   (0,2,1) \tilde{+} (0,1,2)  $. Hence, the expansion of the Poisson brackets of two degree $2$ generators contains only in the degree $3$ generating set. That is, $\{\textbf{q}_2,\textbf{q}_2\} \sim \textbf{q}_3.$

%Notice that for any $\beta_i \neq \beta_j,$ if $(\alpha,\beta_i + \beta_j)$ is connected, then $ $

We then consider the scenario where $r \geq 3$. We examine the connectivity of $ (\alpha,\beta_j)$ and $(-\alpha,\beta_j)$ for all $1 \leq j \leq r$. Given the symmetry property, we inspect the connectivity of pairs $(\alpha,\beta_j)$ for each $j$. Observe that if $(\alpha,\beta_j)$ are not connected for all $j$, then $\{p,q\} = 0$. Hence, we will assume that at least one pair of roots is connected. Furthermore, as established in Proposition $\ref{rootc},$ the number of distinct roots in the set $J_r$ connected to $\alpha$ does not exceed two. Consequently, we will proceed with the classification based on the number of roots in $J_r$ that are connected to $\alpha.$

We start our analysis by assuming that $J_r$ contains a unique root connected to $\alpha.$ Without loss of generality, let $\alpha = \epsilon_i - \epsilon_j$, and let $\beta_\ell = \epsilon_{i_\ell} - \epsilon_{i_{\ell+1}} \in J_r$ be a unique root such that the pair $(\alpha, \beta_\ell)$ is connected. By definition,
\begin{align}
    \alpha+ \beta_\ell = \epsilon_i - \epsilon_j + \epsilon_{i_\ell} - \epsilon_{i_{\ell+1}} = \left\{\begin{matrix}
         \epsilon_{i_\ell}  - \epsilon_j  & \text{ if } \epsilon_i = \epsilon_{i_{\ell+1}} \\
      \epsilon_i    - \epsilon_{i_{\ell+1}} & \hskip -0.3cm \text{ if } \epsilon_j = \epsilon_{i_\ell}
    \end{matrix}\right. \in \Phi. \label{eq:33}
\end{align}
Hence $\alpha = \epsilon_{i_{\ell+1}} - \epsilon_j$ or $\epsilon_i - \epsilon_{i_\ell}$ with $j \neq i_s, s\in \{1,\ldots,\ell\}$ and $i \neq i_t, t \in \{\ell+1,\ldots,r\}.$ Then using \eqref{eq:commu2r},
\begin{align}
    \{p,q\} = \left\{\begin{matrix}
         C_{\alpha,\beta_\ell} e_{\alpha+\beta_\ell} e_{-\alpha} \prod_{k\neq \ell}^r e_{\beta_k} +  C_{-\alpha,\beta_{\ell+1}} e_{\beta_{\ell+1}-\alpha} e_\alpha \prod_{k\neq \ell+1}^r e_{\beta_k} & \text{ if } \alpha = \epsilon_{i_{\ell+1}} - \epsilon_j \\
         \\
         C_{\alpha,\beta_\ell} e_{\alpha+\beta_\ell} e_{-\alpha} \prod_{k\neq \ell}^r e_{\beta_k} +  C_{-\alpha,\beta_{\ell-1}} e_{\beta_{\ell-1}-\alpha} e_\alpha \prod_{k\neq \ell-1}^r e_{\beta_k} & \hskip -0.3cm \text{ if } \alpha =\epsilon_i - \epsilon_{i_\ell}.
    \end{matrix}\right.  \label{eq:commu2rnc}
\end{align}
It is clear that all the monomials in $\eqref{eq:commu2rnc}$ are indecomposable. Hence, the explicit grading for the term $\{p,q\}$ is $(0,n_+,n_-),$ where $n_+ + n_- = \deg p + \deg q -1.$

 Now, we assume that there is more than one connected pair. Without loss of generality, assume that $(\alpha,\beta_\ell)$ and $(\alpha,\beta_l)$ are connected for fixed $l \neq \ell \in \{1,\ldots,r\}$.  Based on the previous discussion, $\alpha$ can be expressed as $\epsilon_{i_\ell} - \epsilon_{i_l}$ or $\epsilon_{i_{l+1}} - \epsilon_{i_\ell}$. Consequently,
\begin{align}
    \{p,q\} = \left\{\begin{matrix}
         C_{\alpha,\beta_\ell} e_{\alpha+\beta_{\ell-1}} e_{-\alpha} \prod_{k\neq \ell-1}^r e_{\beta_k}  +  C_{\alpha,\beta_l} e_{\alpha+\beta_l} e_{-\alpha} \prod_{k\neq l}^r e_{\beta_k} &  \hskip -0.9cm \text{ if } \alpha = \epsilon_{i_\ell} - \epsilon_{i_l} \\
      \hskip -0.5cm  + \  C_{-\alpha,\beta_{l-1}} e_{\beta_{l-1}-\alpha} e_\alpha \prod_{k\neq l-1}^r e_{\beta_k}  +  C_{-\alpha,\beta_{\ell}} e_{\beta_{\ell}-\alpha} e_\alpha \prod_{k\neq \ell}^r e_{\beta_k}   \\
        \\
       C_{\alpha,\beta_\ell} e_{\alpha+\beta_\ell} e_{-\alpha} \prod_{k\neq \ell}^r e_{\beta_k}  +  C_{\alpha,\beta_l} e_{\alpha+\beta_l} e_{-\alpha} \prod_{k\neq l}^r e_{\beta_k}& \hskip -0.5cm\text{ if }\alpha = \epsilon_{i_{l+1}} - \epsilon_{i_\ell} \\
        +  C_{-\alpha,\beta_{l+1}} e_{\beta_{l+1}-\alpha} e_\alpha \prod_{k\neq l+1}^r e_{\beta_k}+  C_{-\alpha,\beta_{\ell-1}} e_{\beta_{\ell-1}-\alpha} e_\alpha \prod_{k\neq \ell-1}^r e_{\beta_k}.
    \end{matrix}\right.   \label{eq:onep}
\end{align}
Note that each term in the right hand side of \eqref{eq:onep} is decomposable. Assume that $\ell < l$, and if $\alpha = \epsilon_{i_\ell} - \epsilon_{i_l}$, the components in the Poisson bracket $\eqref{eq:onep}$ can be further modified as follows
\begin{align*}
    \{p,q\} = & \, C_{\alpha,\beta_\ell} \left(e_{-\alpha} \prod_{j = \ell}^{l-1} e_{\beta_j} \right)  \left( e_{\alpha+\beta_{\ell-1}} \prod_{k\neq \ell-1,\ldots,l-1}^r e_{\beta_k} \right)  +  C_{\alpha,\beta_l} \left(e_{-\alpha} \prod_{j = \ell}^{l-1} e_{\beta_j} \right) \left(e_{\alpha+\beta_l}\prod_{k\neq \ell,\ldots,l }^r e_{\beta_k}\right)  \\
    & \,+  C_{-\alpha,\beta_{l-1}} \left(e_{\beta_{l-1}-\alpha} \prod_{j = \ell}^{l-2}e_{\beta_j}  \right) \left(e_\alpha \prod_{k\neq \ell,\ldots, l }^r e_{\beta_k} \right)+  C_{-\alpha,\beta_\ell} \left( e_{\beta_\ell-\alpha}  \prod_{j = l-1}^{\ell+1} e_j\right)\left( e_\alpha \prod_{k\neq \ell+1,\ldots,l-1}^r e_{\beta_k}\right) .
 \end{align*}
 On the other hand, if $\alpha = \epsilon_{i_{l+1}} - \epsilon_{i_\ell} $ with $l+1 < \ell$. The components in $\eqref{eq:onep}$ are decomposed into
 \begin{align*}
       \{p,q\} = & \, C_{\alpha,\beta_\ell}  \left( e_{-\alpha} \prod_{j = l+1}^{\ell-1} e_{\beta_j} \right) \left( e_{\alpha+\beta_\ell} \prod_{k\neq  l+1,\ldots,\ell}^r e_{\beta_k} \right) +  C_{\alpha,\beta_l} \left( e_{-\alpha} \prod_{j = l+1}^{\ell-1} e_{\beta_j} \right) \left(e_{\alpha+\beta_l} \prod_{k\neq l,\ldots,\ell-1}^r e_{\beta_k}\right) \\
       & \,+  C_{-\alpha,\beta_{l+1}} \left( e_{\beta_{l+1}-\alpha} \prod_{j = l+2}^{\ell-1} e_{\beta_j} \right) \left( e_\alpha \prod_{k\neq l+1,\ldots,\ell-1}^r e_{\beta_k}\right)+  C_{-\alpha,\beta_{\ell-1}}\left( e_{\beta_{\ell-1}-\alpha} \prod_{j = l+1}^{\ell-2}  e_{\beta_j} \right) \left(   e_\alpha \prod_{k\neq l+1,\ldots,\ell-2}^r e_{\beta_k}   \right).
 \end{align*}
  This determines all possible components in the non-trivial bracket $\{p,q\}$ from the grading in $\eqref{eq:commt}.$  

 In the context of the bracket $\{\textbf{q}_s,\textbf{q}_r\}$ where $3 \leq s, r \leq n+1$, the explicit polynomials within each non-trivial Poisson bracket can be divided into two distinct categories: one that involves Cartan elements and one that does not. This separation allows us to design an algorithmic approach to facilitate classifications. Specifically, given a bracket $\{p,q\}$ where $p \in \textbf{q}_s$ and $q \in \textbf{q}_r$, our initial step is to determine whether there exists a root $\beta_j \in J_r$ such that it satisfies the condition $\alpha = - \beta_j$. In cases where such a root does not exist, the subsequent task is to search for all connected roots. In a routine classification, we can deduce the possible decomposition of the components in the non-trivial brackets. In the Appendix \ref{app}, we report the classification of the case with $\{\textbf{q}_3,\textbf{q}_r\}$ for any $3 \leq r \leq n+1$.

\subsection{Cartan centralizer of $S( A_3)$}
\label{5.3}

 In Subsection \ref{5.2}, we propose a novel methodology that permits a more significant simplification of the grading terms present in the Poisson brackets of $\mathcal{Q}_{A_n}(n)$ by employing the properties of the root system. Within this section, we will focus on the application of these conceptual frameworks to a polynomial algebra of increased degree, which is intrinsically related to the rank-three Lie algebra $A_3$.  To present the calculation explicitly, we shall return to our previous indices notation to indicate polynomials and structure constants in the expansions of the non-trivial brackets, rather than relying on the roots in the generators.  In this case, $$\textbf{Q}_4 = \left\{h_1,h_2,h_3, p_{i,j}, p_{i,j,k},p_{i,j,k,l} : 1 \leq i \neq j \neq k \neq l \leq 4\right\} = \textbf{q}_1 \sqcup \textbf{q}_2 \sqcup \textbf{q}_3 \sqcup \textbf{q}_4  $$  with $\mathrm{Card}(\textbf{q}_1) =3, \mathrm{Card}(\textbf{q}_2) = \mathrm{Card}(\textbf{q}_4) = 6$ and $\mathrm{Card}(\textbf{q}_3) = 8.$ Here $\mathrm{Card}(\cdot)$ is the cardinality of the set. As each element in $\textbf{Q}_{n+1}$ can be realized as a $k$-cycle in the symmetric group $S_{n+1}$ with $k \leq n+1$, the cyclic symmetry of the indices in the monomial $p_{i_1,i_2,\ldots ,i_k }$ gives rise to the same element. That is,
\begin{align}
     p_{i_1,i_2,\ldots ,i_k } = p_{i_2,\ldots i_k,i_1 } = \ldots = p_{i_k,i_1,\ldots, i_{k-1} }. 
\end{align} 
 Therefore, for any $1 \leq i \neq j \neq k \neq l  \leq 4,$ the indices representative of the monomials in $\textbf{q}_3$ and $\textbf{q}_4$ are  \begin{align}
    p_{i,j,k} ,\text{ } p_{i,k,j} \quad  \text{ and } \quad  p_{i,j,k,l} , \text{ } p_{i,l,k,j} , \text{ } p_{i,l,j,k} , \text{ } p_{i,k,j,l} , \text{ } p_{i,j,l,k} , \text{ } p_{i,k,l,j} . \label{eq:qq34}
\end{align}

Using Proposition $\ref{5.1},$ we are able to determine the grading of each Poisson bracket as presented in Subsection \ref{4.3}. Here, we will present an example from $\{\textbf{q}_3,\textbf{q}_4\}.$ We obtain that $\mathcal{G} (\textbf{q}_3) = \{(0,1,2),(0,2,1)\}$ and $\mathcal{G}(\textbf{q}_4) = \{(0,1,3),(0,2,2),(0,3,1)\}.$ This implies that the grading in the non-trivial Poisson bracket $\{\textbf{q}_3,\textbf{q}_4\}$ will be different. For instance, for any $p \in \textbf{q}_3$ with $\mathcal{G}(p) = (0,1,2)$ and $ q\in\textbf{q}_4$ with  $\mathcal{G}(q) = (0,1,3)$, $\mathcal{G}(\{p,q\}) = (1,1,4) \tilde{+} (0,2,4) \tilde{+} (0,1,5).$ Then the allowed polynomials from each homogeneous generators are \begin{align*}
    (1,1,4) = (0,1,5) = \emptyset; \text{ } (0,2,4) = \{p_{1,4,3,2}p_{1,2},p_{1,4,3,2}p_{1,3},p_{1,4,3,2}p_{1,4},p_{1,4,3,2}p_{2,3},p_{1,4,3,2}p_{2,4},p_{1,4,3,2}p_{3,4}\}.
\end{align*} Therefore, we have:   \begin{align*}
    \{p,q\}=a_1p_{1,4,3,2}p_{1,2}+ a_2 p_{1,4,3,2}p_{1,3}+ a_3 p_{1,4,3,2}p_{1,4}+ a_4 p_{1,4,3,2}p_{2,3}+ a_5 p_{1,4,3,2}p_{2,4} + a_6 p_{1,4,3,2}p_{3,4},
\end{align*} for some coefficients $a_1,\ldots,a_6 \in \mathbb{R}$.  Now, take a $q' \in \textbf{q}_4$ with $\mathcal{G} (q') = (0,2,2)$. Then $\mathcal{G}(\{p,q'\}) = (1,2,3) \tilde{+} (0,2,4) \tilde{+} (0,3,3)$. In this case, there are $39$ permissible polynomials in $\mathcal{G}(\{p,q'\})$.  %the allowed generators in each homogeneous grading are \begin{align*}
 %   (1,2,3) = \{\}
%\end{align*}
Eventually, under the grading of the polynomial algebra,
the expected polynomial relations in reduced compact forms are given by \begin{align}
    \{\textbf{q}_2,\textbf{q}_2\} \sim & \,    \textbf{q}_3  \nonumber \\
    \{\textbf{q}_2,\textbf{q}_3\} \sim & \,   \textbf{q}_2^2 + \textbf{q}_1 \textbf{q}_3   \nonumber\\
     \{\textbf{q}_2,\textbf{q}_4\} \sim  & \,  \textbf{q}_1 \textbf{q}_2^2 +    \textbf{q}_1 \textbf{q}_4 + \textbf{q}_2 \textbf{q}_3 \label{eq:5.15}  \\
    \{\textbf{q}_3,\textbf{q}_4\} \sim  & \,     \textbf{q}_2 \textbf{q}_4 + \textbf{q}_1\textbf{q}_2 \textbf{q}_3 + \textbf{q}_2^3 \nonumber\\
  \{\textbf{q}_4,\textbf{q}_4\} \sim  & \,     \textbf{q}_1 \textbf{q}_2 \textbf{q}_4 +  \textbf{q}_1\textbf{q}_3^2  + \textbf{q}_2^2 \textbf{q}_3 + \textbf{q}_3\textbf{q}_4  . \nonumber
\end{align}  Building upon the initial presentation, we present in Table \ref{tab:my_label00} a comprehensive analysis that enables an in-depth comparison of the allowable polynomials found in the non-trivial Poisson brackets.
    \begin{table}[h]
     \centering
     \begin{tabular}{|c|c|c|c|}
     \hline
     \textsf{ Poisson brackets}    & $\begin{matrix}
            \textsf{   No. of polynomials without} \\
         \textsf{  using the grading method }
      \end{matrix}$  & $\begin{matrix}
          \textsf{Maximum No. of polynomials after applying} \\
          \textsf{ the grading method} 
      \end{matrix}$ & $\Delta$  \\
          \hline 
  $\{\textbf{q}_2,\textbf{q}_2\}$       & $ 31 $ & $8 $ & $23$  \\
   \hline 
  $\{\textbf{q}_2,\textbf{q}_3\}$       & $102 $ & $39 $ & $63$  \\
          \hline 
 $\{\textbf{q}_2,\textbf{q}_4\}$       & $261$ & $129 $ & $133 $ \\
          \hline  
    $\{\textbf{q}_3,\textbf{q}_4\}$       & $478 $ & $236 $ & $242 $  \\
           \hline 
  $\{\textbf{q}_4,\textbf{q}_4\}$       & $ 990 $ & $492 $ & $498$  \\
   \hline 
     \end{tabular}
     
     \quad
     
     \caption{Comparison of the number of polynomials }
     \label{tab:my_label00}
 \end{table} 
%From \eqref{eq:qq34}, for all $1 \leq i \neq j\neq k \neq l \leq 4,$ the degree $3$ generators in $\textbf{q}_3$ are distinguished by $p_{i,j,k}$ and $p_{i,k,j}$ and $p_{j,k,l}$ and $p_{j,l,k}$.

Here $\Delta$ means the difference of the allowed polynomials in the Poisson brackets through the two distinct approaches. To further reduce the number of the polynomial components after the application of the grading method, we will implement the tools provided in Section $\ref{5.2}$. The polynomials in the non-trivial Poisson brackets $\{\textbf{q}_s,\textbf{q}_t\}$ can be separated into its Cartan and non-Cartan components for $2 \leq s, t \leq 4$. Initially, we examine the case when $\{\textbf{q}_2,\textbf{q}_t\}$ includes the Cartan elements. In other words, we will consider the permissible polynomials from the homogeneous grading $(1,n_+,n_-).$  We start with the Poisson brackets in $\{\textbf{q}_2,\textbf{q}_3\}$.  Referring to the notation as set out in Section \ref{5.2}, consider $p = e_{\alpha}e_{-\alpha} \in \textbf{q}_2$, where we designate $\alpha = \epsilon_i - \epsilon_j$ in such a manner that $e_\alpha =   e_{ij}$, and $e_{-\alpha} = e_{ji}$. Our objective is to identify a generator $q \in \textbf{q}_3$ such that there exists a unique root $\beta$ within $R(q)$, which fulfills the condition $\alpha + \beta = 0$. We now proceed to find out the expansion in the Poisson brackets $\{p_{i,j},p_{i,j,k}\}$. In this context, the roots in $p_{i,j,k}$ are $J_3 = \{\beta_{ij},\beta_{jk},\beta_{ki}\}$, and $\beta_{ij} = -\alpha$.  \begin{remark}
In the following, we consider the root representation of the structure tensor $C_{\beta_\ell,\beta_{\ell+1}}$ in an indexed form. For example, if $\beta_\ell = \beta_{kl}: = \epsilon_k -\epsilon_l \in  \Phi_{A_3}$, we select $\beta_{\ell+1} = \beta_{lm} := \epsilon_l - \epsilon_m$. Consequently, the tensor $C_{\beta_\ell,\beta_{\ell+1}}$ is reformulated as $C_{kl,lm}$. 
\end{remark} Utilizing formula \eqref{eq:comp} in Proposition $\ref{5.5}$, we derive that \begin{align*}
      \{p_{i,j},p_{i,j,k}\} = & \, \left(\sum_{\ell=1}^3C_{ij,ji}^\ell  \,  h_\ell\right)   \, p_{i,j,k}  +  \left(C_{ij,jk} \,   p_{k,i} + C_{ki,ij} p_{k,j} \right)p_{i,j} .
      \end{align*} Here $C_{ij,jk}$ are the structure constants of $\{e_{ij},e_{jk}\}$, and same for $C_{ki,ij}$. Note that as $\{e_{ij},e_{ji}\} \in \mathfrak{h}^*$, the coefficients $C_{ij,ji}^\ell$ run through all the generators in $\mathfrak{h}^*$. Similarly, we deduce \begin{align*}
    \{p_{i,j},p_{i,k,j}\} = & \,\left(\sum_{\ell=1}^3C_{ij,ji}^\ell  \,  h_\ell\right)\, p_{i,k,j} + \left(C_{ji,ik} \, p_{j,k}   + C_{ji,kj} p_{i,k}\right) p_{i,j}.
\end{align*} We now turn to study the non-trivial relations in the Poisson brackets of $\{\textbf{q}_2,\textbf{q}_4\}$.  Using \eqref{eq:comp} again, it can be shown that the nontrivial brackets comprise higher-order monomials alongside either Cartan elements or quadratic monomials. Thus,  
\begin{align*}
    \{p_{i,j},p_{i,j,k,l}\} = & \,\left(\sum_{\ell=1}^3C_{ij,ji}^\ell  \,  h_\ell \right)\, p_{ i,j,k,l}+ \left( C_{li,ij} \,  p_{l,j, k}  + C_{ij,jk} \,   p_{ i,k  , l}\right) p_{i,j}   \\
    \{p_{i,j},p_{j,i,k,l}\} = & \,\left( \sum_{\ell=1}^3C_{ij,ji}^\ell  \,  h_\ell\right)\, p_{j,i,k,l} + \left(C_{lj,ji} \, p_{i,j} p_{i,k,l} + C_{ji,ik} \,  p_{j,k,l}\right)p_{i,j} \\
    \{p_{i,j},p_{i,j,l,k}\} = & \, \left(\sum_{\ell=1}^3C_{ij,ji}^\ell  \,  h_\ell\right)\, p_{i,j,l, k} + \left(C_{ij,ki} \, p_{j,i}p_{k ,j,l}+ C_{ij,jl} \,  p_{j, i,l ,k}\right)p_{i,j}   \\
    \{p_{i,j},p_{j,i,l,k}\} = & \,\left(\sum_{\ell=1}^3C_{ij,ji}^\ell  \,  h_\ell\right)\, p_{j, i,l,k} + \left(C_{ji,il} \, p_{i,j}p_{j,l,k} + C_{ji,kj} \,  p_{i,l,k}\right)p_{i,j}.
\end{align*}   On the other hand, consider the Poisson brackets without the Cartan elements involved.    As discussed in Subsection $\ref{5.2},$ we first assume that there exists only one root from the monomials of $\textbf{q}_3$ that is connected to the root $\alpha_{ij}$ in $p_{i,j}$. From $\eqref{eq:commu2rnc},$ we observe that
\begin{align*}
 \{p_{i,j},p_{j,k,l}\} = & \, C_{ij,jk} \, p_{j,i,k,l} + C_{lj,ji} \, p_{j,k,l,i} \\
       \{p_{i,j},p_{j,l,k}\} = & \, C_{ij,jl} \,p_{i,j,l,k}+ C_{kj,ji}\, p_{i,j,l,k}.
\end{align*}  Otherwise, from $\eqref{eq:onep}$, two connected roots will lead to the decomposition in each of the components as follows: \begin{align*}
      \{p_{i,j},p_{i,l,j,k}\} = & \, \left(C_{ij,jk} \,  p_{i,k} + C_{ij,ki} \,p_{k,j}\right) p_{j,i,l} + \left(C_{ji,il} \, p_{ j,l}  + C_{ji,lj} \, p_{l,i }\right) p_{j,k,i }  \\
     \{p_{i,j},p_{i,k,j,l}\} = & \, \left(  C_{ij,jl} \, p_{i,l}  + C_{ij,li} \,p_{l,j} \right) p_{i,k,j} + \left( C_{ji,ik} \, p_{j,k }  + C_{ji,kj} \, p_{k,i } \right) p_{j,l,i}.
\end{align*}

Next, we consider the non-trivial brackets $\{\textbf{q}_3,\textbf{q}_r\}.$ Here $r = 3$ and $4$.  In analogy to the previous discussion, we split the case into the Cartan involved part and non-Cartan involved part. We first focus on the Cartan-implied case, starting with $r = 3$. Without loss of generality, assume that $p = e_{\alpha_1} e_{\alpha_2} e_{\alpha_3} $ and $q = e_{\beta_1} e_{\beta_2} e_{\beta_3}  $.  Based on the analysis provided in the Appendix \ref{app}, the classification is based on the number of roots in $R(p)$, which permits a singular root in $R(q)$ and such that their sum is zero. In this case, we have only two possibilities:

(a) There exists a unique $\alpha_u$ in $\{\alpha_1,\alpha_2,\alpha_3: \alpha_1 + \alpha_2 + \alpha_3 = 0\}$ such that $\alpha_u + \beta_v = 0$ for a fixed $u$ and the rest of the roots are connected. Then using \eqref{eq:35} and \eqref{eq:36}, we obtain\begin{align*}
     \{p_{i,j,k},p_{j,l,k}\} = & \, \left(\sum_{\ell=1}^3C_{jk,kj}^\ell \, h_\ell \right)\, p_{i,j,l,k} + C_{ij,jl} \, p_{k,i,l} p_{k,j}     + C_{ki,lk} \,p_{l,i,j} p_{k,j}\\
    \{p_{i,k,j},p_{j,k,l}\} = & \left(\sum_{\ell=1}^3C_{jk,kj}^\ell \, h_\ell \right) \, p_{j,i,k,l} +C_{ik,kl} \, p_{k,j} p_{j,i,l}   + C_{ji,lj} \, p_{i,k,l} p_{j,k} .
\end{align*}

(b) If $\alpha_u + \beta_v = 0$ for all $u,v = 1,2,3,$ then from Proposition \ref{5.8}, we have  \begin{align*}
       \{p_{i,j,k},p_{i,k,j}\} = & \, \left(\sum_{\ell=1}^3C_{ij,ji}^\ell  \,  h_\ell\right) \, p_{j,k} p_{i,k} + \left(\sum_{\ell=1}^3 C_{jk,kj}^\ell h_\ell \right)\, p_{i,j} p_{i,k} + \left( \sum_{\ell=1}^3 C_{ki,ik}^\ell h_\ell\right) \, p_{i,j} p_{k,j}.
\end{align*}

On the other hand, the Cartan-free case is simply given by
\begin{align*}
   \{p_{i,j,k},p_{j,k,l}\} = & \, C_{ij,jk} \,  p_{i,k} p_{k,l,j} + C_{jk,kl} \, p_{j,l} p_{k,i,j} + C_{jk,lj} \, p_{l,k} p_{k,i,j} + C_{jk,ki} \, p_{i,j} p_{k,l,j} \\
     \{p_{i,k,j},p_{j,l,k}\} = & \, C_{ik,kj} \, p_{i,j} p_{l,k,j}   + C_{kj,jl} \, p_{k,l} p_{i,k,j} + C_{ji,kj} \, C_{ki,ik} p_{k,j,l}.
\end{align*}

We now consider the case with $r = 4.$ Starting with the Cartan-involved case. For any $p \in \textbf{q}_3$ and $q \in \textbf{q}_4,$ we can always find roots $\alpha_u \neq \alpha_w$ in $\{\alpha_1,\alpha_2,\alpha_3,\alpha_4 : \alpha_1 + \alpha_2 + \alpha_3 + \alpha_4 = 0\}$ (roots in $p$) such that $\alpha_u + \beta_v = 0$ and $\alpha_w + \beta_z =0$. Here $\beta_v \neq \beta_z \in J_4$ with $1 \leq u,v,w,z \leq 4.$ A direct computation shows that %The roots in $p$ is some sort of reordering of roots in $q.$ Then if $\alpha_i + \beta_j = 0,$ there must exist $\alpha_k \neq \alpha_i$ such that $\alpha_k + \beta_l = 0.$
\begin{align*}
       \{p_{j,k,l},p_{i,j,l,k}\}   = & \,\left( \left(\sum_{\ell=1}^3C_{kl,lk}^\ell h_\ell \right) \, p_{l,j}     + \left(\sum_{\ell =1}^3 C_{lj,jl}^\ell \, h_\ell \right) \, p_{k,i} \right) p_{k,i,j} + \left(C_{jk,ki} \, p_{j,i}  +  C_{jk,ij} \, p_{i,k}\right) p_{k,l} p_{j,l}   \\
     \{p_{j,k,l},p_{i,l,k,j}\}   = & \,\left(\left(\sum_{\ell=1}^3 C_{jk,kj}^\ell h_\ell \right)  \, p_{k,l}  + \left(\sum_{\ell =1}^3  C_{kl,lk}^\ell \, h_\ell \right) \,p_{j,k} \right) p_{l,j,i} + \left( C_{lj,ji} \, p_{i,l}   + C_{il,lj} \, p_{i,j}\right) p_{j,k} p_{k,l}\\
        \,\{p_{j,k,l},p_{l,i,k,j}\}   = & \, \left( \left(\sum_{\ell=1}^3 C_{jk,kj}^\ell h_\ell \right)  \, p_{l,j}   + \left(\sum_{\ell =1}^3  C_{lj,jl}^\ell \, h_\ell \right) \, p_{j,k}\right) p_{l,i,k} +\left(C_{kl,li} \, p_{k,i}   + C_{kl,ik} \, p_{i,l} \right)p_{j,k}p_{l,j}  \\
    \{p_{j,l,k},p_{i,j,k,l}\}  = & \,  \left( \left(\sum_{\ell=1}^3 C_{lk,kl}^\ell h_\ell \right) \, p_{k,j}  + \left(\sum_{\ell =1}^3  C_{kj,jk}^\ell \, h_\ell \right)\, p_{k,l} \right)  p_{i,j,l} +\left(C_{jl,li} \, p_{j,i}   +C_{jl,ij} \, p_{i,l}  \right)p_{k,l}p_{k,j}   \\
     \{p_{j,l,k},p_{j,i,k,l}\} = & \, \left(\left(\sum_{\ell=1}^3 C_{jl,lj}^\ell h_\ell \right) \, p_{k,l} + \left(\sum_{\ell =1}^3  C_{lk,kl}^\ell \, h_\ell \right) \, p_{l,j}\right)  p_{j,i,k} + \left( C_{kj,ji} \, p_{k,i} + C_{kj,ik} \, p_{i,j}\right) p_{l,k} p_{l,j}\\
       \{p_{j,l,k},p_{i,l,j,k}\}  = & \, \left(\left(\sum_{\ell=1}^3 C_{jl,lj}^\ell h_\ell \right)  \, p_{k,j}+ \left(\sum_{\ell =1}^3  C_{kj,jk}^\ell \, h_\ell \right) \,  p_{j,l} \right) p_{k,i,l} + \left(C_{lk,ki} \, p_{l,i} +C_{lk,il} \, p_{i,k} \right) p_{j,l}  p_{k,j}   \\
    \{p_{j,l,k},p_{i,k,l,j}\}  = & \, \left( \left(\sum_{\ell=1}^3C_{jl,lj}^\ell h_\ell \right)  \, p_{k,l}   + \left(\sum_{\ell =1}^3  C_{lk,kl}^\ell \, h_\ell \right)\, p_{j,l}\right) p_{j,i,k} + \left(C_{kj,ji} \, p_{k,i}  + C_{kj,ik} \,p_{i,j}\right) p_{k,l} p_{j,l} .
\end{align*} Moreover, suppose that there are no Cartan elements in the brackets $\{\textbf{q}_3,\textbf{q}_4\}$. We observe that for any $p \in \textbf{q}_3,$ each root in the set $\{\alpha_1,\alpha_2,\alpha_3: \alpha_1 + \alpha_2 + \alpha_3 = 0\}$ is connected to two distinct roots in $J_4.$ Taking into account the case (b2) in the Appendix \ref{app}, the explicit components in the expansions within the brackets $\{\textbf{q}_3,\textbf{q}_4\}$ are delineated by   \begin{align*}
       \left\{p_{j,k,l},p_{i,j,k,l}\right\} =  & \,\left( C_{ij,jk} \,  p_{k,l,i} + C_{kl,li} \, p_{k,i,j} \right)p_{j,k,l} + \left( C_{lj,jk} \, p_{l,i}   + C_{lj,kl} \,  p_{k,j} \right) p_{i,j,k,l}\\
    \{p_{j,k,l},p_{j,i,k,l}\} = & \,\left(C_{kl,ik} \,   p_{i,l,j}  + C_{lj,ji}    \, p_{k,l,i}\right) p_{k,l,j} +\left( C_{jk,kl} \, p_{j,l}   + C_{jk,lj} \, p_{l,k}\right) p_{i,k,l,j}    \\
     \{p_{j,k,l},p_{i,l,j,k}\}     = & \, \left(C_{jk,ki} \, p_{i,l,j}     + C_{lj,il } \, p_{i,j,k} \right) p_{j,k,l } + \left(C_{ kl,lj} \, p_{k,j}  + C_{kl,jk } \, p_{j,l} \right)p_{i,l,j,k}     \\
      \{p_{j,l,k},p_{i,j,l,k}\}  = & \,  \left( C_{jl,ij} \, p_{i,l,k}  +C_{lk,ki} \, p_{l,i,j} \right) p_{k,j,l}  +\left(C_{kj,jl} \, p_{l,k}  +C_{kj,lk} \, p_{l,j}\right) p_{i,j,l,k}   \\
       \{p_{j,l,k},p_{i,l,k,j}\}  = & \, \left(C_{lk,il} \, p_{k,j,i} + C_{kj,ji} \, p_{i,l,k}\right) p_{j,k,l}+ \left(C_{jl,lk} \, p_{j,k } + C_{jl,kj} \, p_{k,l}\right) p_{i,k,j,l} .
\end{align*}

Finally, we consider the components in the non-trivial brackets $\{\textbf{q}_4,\textbf{q}_4\}$. From the generators in $\eqref{eq:qq34}$,  we observe that the brackets of elements with four indices always contain the Cartan elements.  We first consider $p_{i,j,k,l}$ with the rest of the cyclic generators. Using Proposition \ref{5.8}, we first derive that 
\begin{align*}
  \left\{p_{i,j,k,l},p_{l,k,j,i}\right\} = & \left( \sum_{\ell =1}^3\, C_{ij,ji}^\ell \,h_\ell \right) \, p_{j,k} p_{k,l} p_{l,i}+  \left( \sum_{\ell =1}^3C_{jk,kj}^\ell \,h_\ell \right) \, \, p_{j,i} p_{k,l} p_{l,i}+  \left( \sum_{\ell =1}^3\,  C_{kl,lk}^\ell \,h_\ell \right)  \, p_{j,k} p_{i,j} p_{l,i} \\
  & \, + \left( \sum_{\ell =1}^3  C_{li,il}^\ell\, h_\ell \right) \, \, p_{j,k} p_{k,l} p_{j,i} .
  \end{align*} The rest of the Poisson brackets contain only one Cartan components. That is, 
  \begin{align*}
     \{p_{i,j,k,l},p_{i,j,l,k}\} = & \, \left(C_{jk,ki} \,  p_{i,j}   + C_{li,ij} \, p_{j,l}\right) p_{k,l}p_{i,j,k} + \left( C_{li,jl} \, p_{j,i} +C_{ij,jl} \,  p_{i,l}\right) p_{k,l}p_{j,k,i}\\
     & \, +C_{ij,ki} \, p_{k,j}p_{k,l} p_{l,i,j}  + C_{jk,ij} \,  p_{k,l,i}p_{k,i,j,l}+ \left( \sum_{\ell=1}^3  C_{kl,lk}^\ell\, h_\ell \right) \, p_{i,j,k} p_{l,i,j}  \\
      \{p_{i,j,k,l},p_{i,l,j,k}\} = & \, \left(C_{ij,jk}\,   p_{i,k} + C_{ij,ki} \, p_{k,j}\right) p_{i,l}p_{k,l,j} + C_{jk,lj} \, p_{k,l} p_{l,i } p_{j,k,i} + C_{jk,ki} \, p_{i,j} p_{j,l}p_{l,j,k}\\
      & \, +\left( C_{kl,lj} \, p_{k,j} + C_{kl,jk}\, p_{j,l} \right)p_{l,i} p_{i,j,k}  + \left( \sum_{\ell=1}^3  C_{li,il}^\ell\, h_\ell \right) \, p_{i,j,k} p_{k,l,j}  \\
     % ------------------------------------------------------------------------------;
   \left\{p_{i,j,k,l},p_{i,k,j,l}\right\} = & \, \left(C_{kl,ik} \, p_{i,l}   + C_{ij,li} \, p_{l,j} \right)+p_{j,k}p_{j,l,i} C_{ij,jl} \, p_{i,l} p_{j,k} p_{i,k,j} + C_{kl,li} \, p_{k,i } p_{k,j}p_{k,l,i} \\
   & \, +\left( C_{li,ik} \, p_{l,k } + C_{li,jl} \, p_{i,j}\right)p_{j,k} p_{k,l,i}+\left( \sum_{\ell=1}^3  C_{jk,kj}^\ell\, h_\ell \right)\, p_{j,l,i} p_{k,l,i}   \\
    \left\{p_{i,j,k,l},p_{i,k,l,j}\right\} = &  \, \left(   C_{jk,kl}  \, p_{j,l}  + C_{jk,lj} \, p_{l,k} \right) p_{j,i}p_{k,l,i} + C_{kl,ik} \, p_{i,l} p_{i,j}p_{j,k,l} + C_{li,ik} \, p_{l,k}p_{i,j}p_{k,l,j} \\
    & \,  + \left(C_{kl,lj} \, p_{k,j}  + C_{li,kl} \, p_{k,i }p_{i,j}\right)p_{k,l,i} +\,\left( \sum_{\ell=1}^3C_{li,il}^\ell  \, h_\ell \right) \, p_{j,k,l } p_{k,l,i}  \\
      \{p_{i,j,k,l},p_{i,l,j,k}\} = & \,  \left( C_{ij,jk}\, p_{i,k}   + C_{ij,ki} \,  p_{j,k} p_{l,i} \right) p_{l,j,k}+ \left(C_{kl,lj} \, p_{k,j}  + C_{kl,jk}\, p_{j,l} \right)p_{i,l}p_{i,j,k} \\
      & \,+ \left(C_{jk,ki} \, p_{ i,j}  p_{l,j,k} + C_{jk,lj} \, p_{l,k} p_{j,k,i}\right) p_{i,l} + \left( \sum_{\ell=1}^3C_{li,il}^\ell\, h_\ell \right) p_{i,j,k} p_{k,l,j} .
      \end{align*}
Moreover, the derivation of the remaining nontrivial generators is achieved based on $\textbf{q}_4$ as detailed in $\eqref{eq:qq34}$. By employing an analogous methodology to that previously described, we subsequently derive the comprehensive expansions for the rest of the Poisson brackets, outlined below:
      \begin{align*}
 %  \{p_{i,l,j,k},p_{i,j,k,l}\} = & \,C_{ki,jk}\, p_{i,j,k} p_{klj,jk}+ \left(C_{lj,jk} \, p_{k,l}  + C_{jk,kl} \, p_{j,l} \right)p_{i,l}p_{i,j,k}+ C_{lj,kl} \,   p_{k,j}p_{l,i}p_{k,i,j} \\
%   & \,+ \left(C_{jk,ij} \, p_{i,k}   + C_{ki,ij}\, p_{k,j}\right) p_{l,i} p_{j,k,l}  +  C_{il,li}\, \mathfrak{h} \, p_{j,k,i} p_{j,k,l}   \\
    \{p_{i,l,j,k},p_{i,j,l,k}\} = & \, \left(C_{il,lk}\, p_{i,k}   + C_{il,ki} \, p_{k,l}\right) p_{l,j}p_{j,k,i} + \left(C_{jk,ij} \, p_{i,k}  + C_{ki,ij}\, p_{k,j}\right) p_{j,l} p_{l,k,i} \\
    & \,+ \left(C_{jk,ki} \, p_{i,j}   p_{i,l,k} + C_{ki,lk}\,   p_{i,l}p_{i,j,k} \right) p_{l,j}  +\left(\sum_{\ell=1}^3 C_{lj,jl}^\ell \, h_\ell\right)  \, p_{j,k,i } p_{l,k,i} \\
   %   \{p_{iljk,ik},p_{iljk,ik}\} = & \, 0.
    % \{p_{i,l,j,k},p_{i,k,j,l}\} = & \, C_{il,ki} \, \mathfrak{h} \,  p_{} p_{}  +C_{lj,jl} \, \mathfrak{h} \,  p_{} p_{}  +C_{jk,kj} \, p_{} p_{} p_{}  +C_{} \, p_{} p_{} p_{}  +C_{} \, p_{} p_{} p_{}  +C_{} \, p_{} p_{} p_{}  +C_{} \, p_{} p_{} p_{}  +C_{} \, p_{} p_{} p_{}  +C_{} \, p_{} p_{} p_{}   ; \\
         \{p_{i,l,j,k},p_{i,k,l,j}\} = & \,   \left(C_{il,lj} \, p_{i,j}   +C_{il,ji} \, p_{l,j} \right) p_{k,i} p_{l,j,k}  +\left(C_{lj,ji} \, p_{l,i}  p_{k,l,j}  +C_{lj,kl} \, p_{k,j} p_{j,i,l}\right) p_{k,i} \\
        & \, +\left(C_{jk,kl} \, p_{j,l}    +C_{jk,lj} \, p_{k,l} \right)p_{k,i}  p_{j,i,l}  +\left(\sum_{\ell=1}^3C_{ki,ik}^\ell \, h_\ell\right)  \, p_{i,l,j} p_{j,k,l} \\
     \{p_{i,l,j,k},p_{j,i,k,l}\} = & \,  \left( C_{il,lj} \, p_{i,j}   +C_{il,ji} \, p_{j,l}\right)  p_{j,k} p_{l,j,k}  +\left(C_{lj,ji} \, p_{i,l}  +C_{lj,kl} \, p_{j,k}\right) p_{k,i} p_{k,l,j}    \\
     & \, +\left(C_{jk,kl} \, p_{j,l}  +C_{jk,lj} \, p_{l,k} \right)  p_{k,i} p_{i,l,j} +\left(\sum_{\ell=1}^3C_{ki,ik}^\ell \, h_\ell\right)  \, p_{l,j,k} p_{j,i,l}        \\
   \{p_{i,l,k,j},p_{i,k,j,l}\} = & \,  \left( C_{lk,kj} \, p_{l,j}   +C_{lk,jl} \, p_{j,k}\right) p_{i,l} p_{k,j,i}   +\left(C_{kj,jl} \, p_{k,l}  p_{i,k,j}  +C_{kj,ik} \, p_{i,j}p_{k,j,l}\right) p_{l,i}   \\
   & \, +\left(C_{ji,ik} \, p_{j,k}  +C_{ji,kj} \, p_{k,i} \right)  p_{i,l} p_{l,k,j}    +   \left(\sum_{\ell=1}^3 C_{il,li}^\ell \, h_\ell\right)  \, p_{j,i,k} p_{k,j,l}  \\
      \{p_{i,l,k,j},p_{j,i,k,l}\} = & \,  \left( C_{il,lj} \, p_{i,j}  +C_{il,ji} \, p_{j,l}\right) p_{k,l} p_{j,i,k}   +\left(C_{kj,ji} \, p_{k,i}  +C_{kj,ik} \, p_{i,j}\right)  p_{k,l} p_{j,i,l}   \\
      & \,+ \left( C_{ji,ik} \, p_{j,k}   p_{j,i,l}  +C_{ji,lj} \, p_{l,i}  p_{j,i,k} \right)p_{k,l}    + \left(\sum_{\ell=1}^3C_{lk,kl}^\ell \, h_\ell\right) \, p_{j,i,k} p_{j,i,l}  . %   ;\\
      % \{p_{ijlk,ji},p_{i,k,l,j}\} = & \,   C_{jl,lj} \, p_{} p_{} p_{}  +C_{} \, p_{} p_{} p_{}  +C_{} \, p_{} p_{} p_{}  +C_{} \, p_{} p_{} p_{}  +C_{} \, p_{} p_{} p_{}  +C_{} \, p_{} p_{} p_{} +  C_{ij,ji} \, \mathfrak{h} \, p_{} p_{}    .
\end{align*}

From the previously outlined construction, it can be inferred that $\mathcal{Q}_{A_3}(3)$ defines a cubic Poisson algebra.

\section{conclusion}

  In this paper, the procedure proposed in \cite{MR4660510,campoamor2024superintegrable} to determine the commutant in the enveloping algebra associated to a subalgebra chain of reductive Lie algebras has been reexamined, by considering an additional simplification based on gradings determined by the embedding.  Using such a grading of monomials in the centralizer of symmetric algebras, it is possible to considerably reduce the number of admissible polynomials, also leading to a compact presentation of the polynomial algebra and their commutators. Explicitly, three reduction chains related to the simple rank-two complex Lie algebra $\mathfrak{sl}(3,\mathbb{C})$ have been analyzed: the Elliott chain $\mathfrak{so}(3) \subset \mathfrak{su}(3)$, the reduction $\mathfrak{o}(3) \subset \mathfrak{sl}(3,\mathbb{C})$ and $\mathfrak{h} \subset \mathfrak{sl}(3,\mathbb{C})$, previously considered in \cite{MR4660510}, \cite{MR4710584} and \cite{campoamor2023algebraic}, respectively. A grading of the indecomposable polynomials has been presented, as well as a procedure to obtain the Poisson brackets under which the algebraic structure closes in the Poisson-Lie setting. In this context, a description of the main grading properties has been given and it has been illustrated how to use the root systems associated with a semisimple Lie algebra to completely characterize the polynomial algebra $\mathcal{Q}_{A_n}(n)$ that comes from the centralizer in $S(A_{n})$ with respect to the Cartan subalgebra $\mathfrak{h}$. This fact has been shown to be relevant in the theory of superintegrable systems in classical and quantum mechanics. In particular, the last example treated connects with the generic models on the $n$ sphere through the Marsden-Weinstein realizations \cite{campoamor2023algebraic},\cite{campoamor2023polynomial}. On the other hand, the Elliott chain, besides its relevance in nuclear physics, presents some special features, as the embedding is singular, a fact that requires some modifications of the method, as certain properties of Casimir invariants are broken down, making computations and the explicit analysis of the closure of the polynomial algebra harder. In this context, the grading method has been shown to be an effective ansatz to study the reduction chains in the case of singular embeddings. The latter type is particularly relevant in physical applications and labeling problems, where a direct approach through root systems is generally not possible. It should be observed that the proposed construction is completely independent of realizations of Lie algebras as vector fields, hence providing a generic universal character that may allow for a more detailed insight into the particular structure of centralizers in enveloping algebras as well as the associated missing label problems.

In further work, more physical models related to the subalgebra chains can be considered to involve the construction of polynomial algebras. For instance, the interacting boson-fermion model (IBFM) or the supermultiplet model. On the one hand, the IBFM involves Lie algebra chains such as \( \mathfrak{u}(6) \supset \mathfrak{u}(5) \), \(  \mathfrak{su}(3) \supset  \mathfrak{su}(2) \), and related hierarchical structures underpin the IBFM \cite{MR1081533,MR1152802}. Each subalgebra within the chain represents specific symmetries or conserved quantities associated with different physical behaviors of the nucleus, such as vibrational, rotational, or transitional dynamics. For example, the \( \mathfrak{u}(6) \supset  \mathfrak{su}(3) \) chain models rotational symmetries associated with deformed nuclei, whereas the \( \mathfrak{so}(6)\supset \mathfrak{u}(5)  \) chain models vibrational modes pertinent to spherical or near-spherical nuclei. These Lie algebraic chains provide a structured pathway from higher-symmetry groups, representing general nuclear behavior, to more specialized subgroups that describe specific symmetries and conserved quantities relevant to particular nuclear states. On the other hand, within the framework of the supermultiplet model, with the chain \( \mathfrak{su}(4) \supset \mathfrak{su}(2) \times \mathfrak{su}(2) \) \cite{Wig,Draayer,Brunet,Hecht}, further enriches this scheme by providing a unified treatment of both protons and neutrons in the nuclear shell model, encapsulating both their spin and isospin degrees of freedom within a single algebraic structure. The construction of polynomial algebras through the grading method in these directions is currently ongoing.

\section*{Acknowledgement}
IM was supported by the Australian Research Council Future Fellowship FT180100099. YZZ was supported by the Australian Research Council Discovery Project DP190101529. RCS acknowledges financial support by the research grant PID2019-106802GB-I00/AEI/10.13039/501100011033 (AEI/ FEDER, UE). The research of DL is partially funded by MUR - Dipartimento di Eccellenza 2023-2027, codice CUP \textsf{G43C22004580005} - codice progetto \textsf{DECC23\_012\_DIP} and partially supported by INFN-CSN4 (Commissione Scientifica Nazionale 4 - Fisica Teorica), MMNLP project. DL is a member of GNFM, INdAM.

\appendix
\section{Explicit polynomials in $\{\textbf{q}_3,\textbf{q}_r\}$}

\label{app}

In this Appendix \ref{app}, we derive the explicit generators in the non-trivial brackets $\{\textbf{q}_3,\textbf{q}_r\}$ for all $ 3 \leq  r \leq n+1.$  In the following, we will denote $p = e_{\alpha_1}e_{\alpha_2} e_{\alpha_3} \in \textbf{q}_3$ with $\alpha_1 + \alpha_2 + \alpha_3 = 0 $, and $q = e_{\beta_1} \cdots e_{\beta_r}  \in \textbf{q}_r$ with $R(q) = \beta_1 + \ldots + \beta_r = 0.$ Let $J_r : = \{\beta_1,\ldots,\beta_r\}$ such that $\beta_1 + \ldots + \beta_r = 0.$  Moreover, assume that $\alpha_1 = \epsilon_{j_1} - \epsilon_{j_2}$, $ \alpha_2 = \epsilon_{j_2} - \epsilon_{j_3}$ and $\alpha_3 = \epsilon_{j_3}  - \epsilon_{j_1}.$ Here $j_1,j_2,j_3 \in \{1,\ldots,n+1\}.$ Using Proposition $\ref{5.1}$, we deduce that
\begin{align*}
     \mathcal{G}\left(\{p,q\}\right) = \left\{\begin{matrix}
        (1,n_+  ,n_-  + 1) \tilde{+}   (0,n_+ +1,n_- + 1)\tilde{+}(0,n_+  ,n_-+ 2 )  & \text{ if } \mathcal{G}(p) = (0,1,2) \\
         \\
        (1,n_+   + 1,n_- ) \tilde{+}  (0,n_+  + 1,n_-   +1 )  \tilde{+}(0,n_+ +2,n_- ) & \text{ if } \mathcal{G}(p) = (0,2,1)
     \end{matrix}\right.
\end{align*} with $n_+ + n_- = r.$
A direct computation shows that
\begin{align}
    \left\{e_{\alpha_1}e_{\alpha_2} e_{\alpha_3}, \prod_{k \geq 1}^r e_{\beta_k} \right\}  = & \,\underbrace{\sum_{w=1}^r N_{\alpha_1,\beta_w} \prod_{k \neq w}^re_{\beta_k}e_{\alpha_2} e_{\alpha_3}}_{\text{: = $\mathcal{N}_1$}}  +\underbrace{\sum_{w=1}^r N_{\alpha_2,\beta_w}\prod_{k \neq w}^r e_{\beta_k}e_{\alpha_1} e_{\alpha_3}}_{\text{: = $\mathcal{N}_2$}}   +   \underbrace{\sum_{w=1}^r N_{\alpha_3,\beta_w}\prod_{k \neq w}^r e_{\beta_k} e_{\alpha_1} e_{\alpha_2}}_{\text{: = $\mathcal{N}_3$}}  . \label{eq:commur3}
\end{align}
Here $$ N_{\alpha_i,\beta_k} = \left\{e_{\alpha_i} , e_{\beta_k} \right\} = \left\{\begin{matrix}
    C_{\alpha_i,\beta_k} e_{\alpha_i + \beta_k} & \hskip -0.6cm \text{ if } (\alpha_i,\beta_k) \text{ is connected} \\
   0 & \text{ if } (\alpha_i,\beta_k) \text{ is not connected}
\end{matrix} \right.  $$ for all $ 1 \leq i \leq 3.$

%Depends on how many equal terms we may have!! If $\alpha_1 = \beta_j,$ we have one case. whereas $\alpha_1 = -\beta_j$ and $\alpha_2 = -\beta_k,$ we have another case. Notice that, in this case, we must have $j = k+1.$

We will once again examine the Cartan-free and Cartan parts, beginning with the Cartan-involved part. According to Proposition \ref{root}, for each $\alpha_i$, there is a unique $\beta_j \in J_r$ such that $\alpha_i = - \beta_j$. It is important to note that, if $3 < r \leq n+1$, not all $\alpha_1, \alpha_2$ and $\alpha_3$ correspond to non-adjacent roots in $J_r$. For example, if we assume that $\alpha_1 = -\beta_j$ and $\alpha_2 = -\beta_k$, then $\beta_j + \beta_k = \alpha_3 = - \beta_l$, which implies that $q$ is decomposable. Clearly, the Poisson bracket of these monomials results in Cartan-free monomials. To classify the Cartan-involved case, we consider the following cases:

(a) If $\alpha_i = -\beta_j$ for a unique $1 \leq i \leq 3$ and a fixed $j \in \{1,\ldots,r\}$;

(b) If $\alpha_i = -\beta_j$ and $\alpha_s = - \beta_l$ with $ 1 \leq  i \neq s \leq 3 $ and fixed $l \neq j \in \{1,\ldots,r\}$.

Consider case (a). Without loss of generality, assume further that $\alpha_1 = - \beta_j$ for a fixed $j$. The similar analysis holds for letting $\alpha_2$ or $\alpha_3$ equals to $-\beta_j$. From the argument in Proposition $\ref{5.5}$, we deduce that
\begin{align}
    \mathcal{N}_1 =  \left(\sum_{\ell=1}^n C_{-\beta_j,\beta_j}^\ell h_\ell \right)  \, e_{\alpha_2}e_{\alpha_3}\prod_{k \neq j}^r e_{\beta_k}. \label{eq:35}
\end{align} Here $\sum_{\ell=1}^n C_{-\beta_j,\beta_j}^\ell h_\ell \in \mathfrak{h}^*.$ Note that the monomial $ e_{\alpha_2}e_{\alpha_3}\prod_{k \neq j}^re_{\beta_k}$ is indecomposable as $R(e_{\alpha_2}e_{\alpha_3}) = -\alpha_1 = \beta_j.$ Moreover, as $\epsilon_{j_3}$ is an undetermined term, the explicit value of $\mathcal{N}_2$ and $\mathcal{N}_3$ depends on whether $\alpha_2$ and $\alpha_3$ are connected to some $\beta_j$ or not. By Proposition \ref{rootc}, for each $\alpha_i$, there are at most two different choices of roots in the set $J_r$ such that the pairs between $\alpha_i$ with these roots are connected. As $\alpha_1 = -\beta_j,$ we observe that $\alpha_2 + \beta_{j-1} ,\alpha_3 + \beta_{j+1} \in \Phi$. This implies that  \begin{align}
    \mathcal{N}_2 =   C_{\alpha_2,\beta_{j-1}} \left(e_{\alpha_2 + \beta_{j-1}}\prod_{k \neq j,j-1}^r e_{\beta_k}e_{\alpha_3}\right)\underbrace{e_{\beta_j}e_{\alpha_1}}_{\text{$ \in \textbf{q}_2$}} \, \text{ and } \, \mathcal{N}_3 = C_{\alpha_3,\beta_{j+1}} \left( e_{\alpha_3 + \beta_{j+1}}\prod_{k \neq j,j+1}^r e_{\beta_k} e_{\alpha_2}\right) \underbrace{e_{\alpha_1}e_{\beta_j}}_{\text{$ \in \textbf{q}_2$}}   . \label{eq:36}
\end{align}  We then consider the case where more than one root in $J_r$ is connected to $\alpha_2$. Suppose, without loss of generality, that there exists a $\beta_t \in J_r$ with $t\neq j-1$ such that $\alpha_2 + \beta_t \in \Phi$ for a fixed $t$. We further assume that $t <j$. Then
\begin{align*}
     \alpha_2+ \beta_t = & \, \epsilon_{i_j} - \epsilon_{j_3} + \epsilon_{i_t} - \epsilon_{i_{t+1}} = \left\{\begin{matrix}
         \epsilon_{i_l}  - \epsilon_{j_3}  & \text{ if } \epsilon_{i_j} = \epsilon_{i_{t+1}} \\
      \epsilon_{i_j}    - \epsilon_{i_{t+1}} & \hskip -0.3cm\text{ if } \epsilon_{j_3} = \epsilon_{i_t}.
    \end{matrix}\right.
\end{align*}
Note that if $t + 1 =j$, then the value of $\mathcal{N}_2$ and $\mathcal{N}_3$ take the form of $\eqref{eq:36}.$ On the other hand, if $j_3 = i_t,$ then $\alpha_2 = \epsilon_{i_j} - \epsilon_{i_t} = \beta_j + \ldots + \beta_{t-1}$ and $\alpha_3 = \epsilon_{i_t} - \epsilon_{i_{j+1}} = -\beta_t -\ldots - \beta_{j+1}.$ However, one can easily check that $\alpha_1 + \alpha_2 + \alpha_3 = -\beta_t - \beta_j \neq 0.$ Hence $t >j,$ and $\alpha_3 + \beta_{t-1} \in \Phi.$ Therefore, \begin{align*}
     \mathcal{N}_2& =      C_{\alpha_2,\beta_{j-1}} \left(e_{\alpha_2 + \beta_{j-1}}\prod_{k \neq j,j-1}^re_{\beta_k}e_{\alpha_3}\right)e_{\beta_j}e_{-\beta_j} + C_{\alpha_2,\beta_t} \left(e_{\alpha_2 + \beta_l}\prod_{k \neq j,t}^re_{\beta_k}e_{\alpha_3}\right)e_{\beta_j}e_{-\beta_j}  \\
      \mathcal{N}_3 & =   C_{\alpha_3,\beta_{j+1}} \left( e_{\alpha_3 + \beta_{j+1}}\prod_{k \neq j,j+1}^r e_{\beta_k} e_{\alpha_2}\right) e_{-\beta_j}e_{\beta_j} + C_{\alpha_3,\beta_{t-1}} \left(e_{\alpha_3 + \beta_{t-1}}\prod_{k \neq j,t-1}^r e_{\beta_k}e_{\alpha_2}\right)e_{\beta_j}e_{-\beta_j}.
\end{align*}

Now, consider case (b). Without loss of generality, assume that $\alpha_1 = -\beta_j$ and $\alpha_2 = -\beta_l.$ Then using the constraints  \begin{align*}
   0 = \alpha_1 + \alpha_2 + \alpha_3 = \epsilon_{i_{j+1}} - \epsilon_{i_j} + \epsilon_{i_{l+1}} - \epsilon_{i_l} +  \epsilon_{j_3} - \epsilon_{j_1}  ,
\end{align*} we have the following two different possibilities for the choice of the indices:

\begin{center}
    (i) $j = l+1,$ $j_3 = i_l$ and $ j_1 = i_{j+1} ;$ \quad  (ii) $l = j+1$, $j_3 = i_j$ and $ j_1 = i_{l+1} .$
\end{center}

We will only provide the classification for option (i), as the analysis in the remaining case is analogous. Given that $\alpha_3 = \epsilon_{i_l} - \epsilon_{i_{l+2}} = \beta_l + \beta_{l+1}$, it follows that $\alpha_3 + \beta_{l+2}$ and $\alpha_3 + \beta_{l-1}$ are elements of $\Phi$. By Proposition $\ref{rootc}$, $\alpha_3$ is maximally connected to two distinct roots in $J_r$. Returning to $\eqref{eq:commur3},$ we immediately find that
\begin{align*}
     & \,\mathcal{N}_1 =   \left(\sum_{\ell=1}^n C_{-\beta_j,\beta_j}^\ell h_\ell \right)  \left( e_{\alpha_3}\prod_{k \neq j,l}^r e_{\beta_k}\right)e_{-\beta_l} e_{\beta_l}, \text{ } \quad \mathcal{N}_2 =   \left(\sum_{\ell=1}^n C_{-\beta_j,\beta_j}^\ell h_\ell \right) \left( e_{\alpha_3}\prod_{k \neq j,l}e_{\beta_k}\right)e_{-\beta_j} e_{\beta_j} \\
    \mathcal{N}_3 = & \,    C_{\alpha_3,\beta_{l+2}} \left( e_{\alpha_3 + \beta_{l+2}}\prod_{k \neq l,l+1,l+2}^r e_{\beta_k} \right)e_{\alpha_2} e_{\beta_l} e_{\alpha_1}e_{\beta_j} + C_{\alpha_3,\beta_{l-1}} \left(e_{\alpha_3 + \beta_{l-1}}\prod_{k \neq l+1,l,l-1}^r e_{\beta_k}e_{\alpha_3}\right)e_{\beta_{l+1}}e_{\alpha_1}e_{\beta_l}e_{\alpha_2},
\end{align*}
where $  e_{\alpha_3 + \beta_{l+2}}\prod_{k \neq l,l+1,l+2}^r e_{\beta_k}  \in \textbf{q}_{r-2}$ is indecomposable.

Now, we look at the Cartan-free part. In other words, assume that $\alpha_i + \beta_k \neq 0$ for all $1 \leq i \leq 3$ and $ 1 \leq k \leq r$. By Proposition $\ref{rootc},$ for each $\alpha_i,$ there are maximal $2$ distinct choices of $\beta_k$ that are connected to it. From this fact, we will provide the classification using the number of the connected pairs for each $\alpha_i$ and $\beta_k$. Again, if $\alpha_i + \beta_k \notin \Phi$ for all $i,k$, we immediately conclude that $\{p,q\} = 0.$

(A) Suppose that only one of the roots in $\{\alpha_1,\alpha_2,\alpha_3\}$ is connected to some roots in $J_r.$

(a1) For a fixed $\alpha_{i_0}$ with $i_0 \in \{1,2,3\}$, we first assume that there exists only one root $\beta_j \in J_r$ such that $\alpha_{i_0} +\beta_j \in \Phi.$ Here $j$ is a fixed integer from $1$ to $r$. Without loss of generality, assume that $\alpha_1 + \beta_j \in \Phi$. It turns out that we either have 

\begin{center}
  (i) $\alpha_1 + \beta_j , \text{  } \alpha_3 + \beta_{j+1} \in \Phi$; \quad   (ii) $\alpha_1 + \beta_j , \text{  }\alpha_2 + \beta_{j-1} \in \Phi.$
\end{center}

Then
\begin{align}
    \{p,q\} = \left\{\begin{matrix}
        C_{\alpha_1,\beta_j} \left(e_{\alpha_1 + \beta_j} e_{\alpha_2}e_{\alpha_3} \prod_{w\neq j}^r e_{\beta_w} \right) +    C_{\alpha_3,\beta_{j+1}} \left(e_{\alpha_3 + \beta_{j+1}} e_{\alpha_1}e_{\alpha_2} \prod_{w\neq j+1}^r e_{\beta_w} \right)   & \text{ Case (i)}   \\
        \\
         C_{\alpha_1,\beta_j} \left(e_{\alpha_1 + \beta_j} e_{\alpha_2}e_{\alpha_3} \prod_{w\neq j}^r e_{\beta_w} \right) +    C_{\alpha_2,\beta_{j-1}} \left(e_{\alpha_2 + \beta_{j-1}} e_{\alpha_1}e_{\alpha_3} \prod_{w\neq j-1}^r e_{\beta_w} \right)  & \hskip 0.2cm\text{ Case (ii)}
    \end{matrix}\right.
\end{align}
In this case, the Poisson brackets contain only indecomposable monomials.

(a2) Under the assumption of (a1), suppose that there exists a $\beta_a \neq \beta_j$ such that $\alpha_1 + \beta_a \in \Phi$ for a fixed $ 1 \leq a \neq j \leq r$. Then from the relation below,
\begin{align*}
        \alpha_1+ \beta_j = & \, \epsilon_{j_1} - \epsilon_{j_2} + \epsilon_{i_j} - \epsilon_{i_{j+1}} = \left\{\begin{matrix}
         \epsilon_{i_j}  - \epsilon_{j_2}  & \text{ if } \epsilon_{j_1} = \epsilon_{i_{j+1}} \\
      \epsilon_{j_1}    - \epsilon_{i_{j+1}} & \hskip -0.2cm\text{ if } \epsilon_{j_2} = \epsilon_{i_j}
    \end{matrix}\right.  \\
     \alpha_1+ \beta_a = & \, \epsilon_{j_1} - \epsilon_{j_2} + \epsilon_{i_a} - \epsilon_{i_{a+1}} = \left\{\begin{matrix}
         \epsilon_{i_a}  - \epsilon_{j_2}  & \text{ if } \epsilon_{j_1} = \epsilon_{i_{a+1}} \\
      \epsilon_{j_1}    - \epsilon_{i_{a+1}} & \hskip -0.25cm \text{ if } \epsilon_{j_2} = \epsilon_{i_a}
    \end{matrix}\right.
\end{align*}
  we deduce that
\begin{align*}
     \text{(i)} \text{ } & \, \alpha_1 = \epsilon_{i_{j+1}} - \epsilon_{i_a}, \text{ }   \alpha_2 = \epsilon_{i_a} - \epsilon_{j_3}, \text{ }   \alpha_3 = \epsilon_{j_3} - \epsilon_{i_{j+1}}; \\
   \text{(ii)} \text{ } & \, \alpha_1 = \epsilon_{i_{a+1}} - \epsilon_{i_j}, \text{ }   \alpha_2 = \epsilon_{i_j} - \epsilon_{j_3}, \text{ }   \alpha_3 = \epsilon_{j_3} - \epsilon_{i_{a+1}}.
\end{align*}
Consider case (i). A direct computation shows that \begin{align}
    \mathcal{N}_1 = & \, C_{\alpha_1,\beta_j} \left(e_{\alpha_1 + \beta_j} \prod_{k=1}^{a+1} e_{\beta_k} \prod_{v=j+1}^r e_{\beta_v} \right) \left( e_{\alpha_2} e_{\alpha_3} \prod_{w = j}^{a-1} e_{\beta_w}\right)+ C_{\alpha_1,\beta_a} \left(e_{\alpha_1 + \beta_a} \prod_{k=1}^{a-1} e_{\beta_k} \prod_{v=j}^r e_{\beta_v} \right) \left( e_{\alpha_2} e_{\alpha_3} \prod_{w = j+1}^{a-2} e_{\beta_w}\right)   \nonumber \\
     \mathcal{N}_2 = & \, C_{\alpha_2,\beta_{a-1}} \left(e_{\alpha_2 + \beta_{a-1}} e_{\alpha_1}e_{\alpha_3} \prod_{w\neq a-1}^r e_{\beta_w} \right)  \text{ and }  \mathcal{N}_3 =      C_{\alpha_3,\beta_{j+1}} \left(e_{\alpha_3 + \beta_{j+1}} e_{\alpha_1}e_{\alpha_2} \prod_{w\neq j+1}^r e_{\beta_w} \right) . \label{eq:a5}
\end{align}
We will omit case (ii) as the values of $\mathcal{N}_1$, $\mathcal{N}_2$ and $\mathcal{N}_3$ admit the same decomposable monomials as provided in \eqref{eq:a5}.

\vskip 0.4cm

(B) Suppose that two of the roots in $\{\alpha_1,\alpha_2,\alpha_3\}$ are connected to some roots in $J_r.$

(b1) Assume that $\alpha_1,\alpha_2$ is connected to only one root in $J_r.$  Without loss of generality, for fixed $\beta_j \neq \beta_\ell \in J_r,$ suppose that $(\alpha_1,\beta_j)$ and $(\alpha_2,\beta_\ell)$ are connected. Then from the relation below, \begin{align*}
        \alpha_1+ \beta_j = & \, \epsilon_{j_1} - \epsilon_{j_2} + \epsilon_{i_j} - \epsilon_{i_{j+1}} = \left\{\begin{matrix}
         \epsilon_{i_j}  - \epsilon_{j_2}  & \hskip 0.1cm \text{ if } \epsilon_{j_1} = \epsilon_{i_{j+1}} \\
      \epsilon_{j_1}    - \epsilon_{i_{j+1}} & \hskip -0.2cm\text{ if } \epsilon_{j_2} = \epsilon_{i_j}
    \end{matrix}\right.  \\
\alpha_2+ \beta_\ell = & \, \epsilon_{j_2} - \epsilon_{j_3} + \epsilon_{i_\ell} - \epsilon_{i_{\ell+1}} = \left\{\begin{matrix}
         \epsilon_{i_\ell}  - \epsilon_{j_3}  & \hskip 0.1cm \text{ if } \epsilon_{j_2} = \epsilon_{i_{\ell+1}} \\
      \epsilon_{j_2}    - \epsilon_{i_{\ell+1}} &  \text{ if } \epsilon_{j_3} = \epsilon_{i_\ell} \, ,
    \end{matrix}\right.  
\end{align*}
 it follows that
\begin{align}
   \text{(i)} \text{ } & \, \alpha_1 = \epsilon_{i_{j+1}} - \epsilon_{i_{\ell+1}}, \text{ }   \alpha_2 = \epsilon_{i_{\ell+1}} - \epsilon_{j_3}, \text{ }   \alpha_3 =   \epsilon_{j_3}-\epsilon_{i_{j+1}}  ; \\
   \nonumber
    \text{(ii)} \text{ } & \, \alpha_1 = \epsilon_{i_{j+1}} - \epsilon_{i_{\ell+1}}, \text{ }   \alpha_2 = \epsilon_{j_2} - \epsilon_{i_\ell}, \text{ }   \alpha_3 =   \epsilon_{i_\ell}-\epsilon_{i_{j+1}} ;    \\
     \nonumber
    \text{(iii)} \text{ } & \, \alpha_1 = \epsilon_{j_1} - \epsilon_{i_j}, \text{ }   \alpha_2 = \epsilon_{i_j} - \epsilon_{i_\ell}, \text{ }   \alpha_3 =   \epsilon_{i_\ell}-\epsilon_{j_1}.
\end{align}
From now on, the values of $\mathcal{N}_1, \mathcal{N}_2$ and $\mathcal{N}_3$ will have similar decomposable monomials as seen in (a2). Consequently, we will present only one case to clarify the differences in the components of $\mathcal{N}_i$ with $ 1 \leq i \leq 3$. Consider case (iii). Direct computation shows that
\begin{align*}
   &  \mathcal{N}_1 =   C_{\alpha_1,\beta_j} \left( e_{\alpha_1 + \beta_j} e_{\alpha_2}e_{\alpha_3} \prod_{w \neq j }^r e_{\beta_w}\right) \text{ and }  \mathcal{N}_3 =  C_{\alpha_3,\beta_\ell}\left( e_{\alpha_3 + \beta_\ell} e_{\alpha_2}e_{\alpha_1} \prod_{w \neq \ell }^r e_{\beta_w}\right)  \\
    \mathcal{N}_2 = &\, C_{\alpha_2,\beta_{j-1}} \left(e_{\alpha_2 +\beta_{j-1}} \prod_{k=1}^{j-2} e_{\beta_k} \prod_{v = \ell}^r e_{\beta_v}\right) \left( e_{\alpha_1}e_{\alpha_3} \prod_{w = j}^{\ell-1} e_{\beta_w}\right)+ C_{\alpha_2,\beta_\ell} \left( e_{\alpha_2 +\beta_\ell}\prod_{k=1}^{j-1} e_{\beta_k} \prod_{v = \ell+1}^r e_{\beta_v} \right)\left( e_{\alpha_1}e_{\alpha_3} \prod_{w = j}^{\ell-1} e_{\beta_w}\right) .
\end{align*}

(b2) Under the assumption of (b1), we further assume that there exists a $\beta_a \neq \beta_j$ in the set $J_r$ such that $\alpha_1 + \beta_a \in \Phi.$ Then the following relations
\begin{align*}
        \alpha_1+ \beta_j = & \, \epsilon_{j_1} - \epsilon_{j_2} + \epsilon_{i_j} - \epsilon_{i_{j+1}} = \left\{\begin{matrix}
         \epsilon_{i_j}  - \epsilon_{j_2}  & \text{ if } \epsilon_{j_1} = \epsilon_{i_{j+1}} \\
      \epsilon_{j_1}    - \epsilon_{i_{j+1}} & \hskip -0.26cm \text{ if } \epsilon_{j_2} = \epsilon_{i_j}
    \end{matrix}\right.  \\
     \alpha_1+ \beta_a = & \, \epsilon_{k_1} - \epsilon_{j_2} + \epsilon_{i_a} - \epsilon_{i_{a+1}} = \left\{\begin{matrix}
         \epsilon_{i_a}  - \epsilon_{j_2}  & \hskip -0.1cm  \text{ if } \epsilon_{j_1} = \epsilon_{i_{a+1}} \\
      \epsilon_{j_1}    - \epsilon_{i_{a+1}} & \hskip -0.35cm\text{ if } \epsilon_{j_2} = \epsilon_{i_a}
    \end{matrix}\right.  \\
\alpha_2+ \beta_\ell = & \, \epsilon_{j_2} - \epsilon_{j_3} + \epsilon_{i_\ell} - \epsilon_{i_{\ell+1}} = \left\{\begin{matrix}
         \epsilon_{i_\ell}  - \epsilon_{j_3}  & \hskip  0.1cm\text{ if } \epsilon_{j_2} = \epsilon_{i_{\ell+1}} \\
      \epsilon_{j_2}    - \epsilon_{i_{\ell+1}} & \hskip -0.2cm\text{ if } \epsilon_{j_3} = \epsilon_{i_\ell}
    \end{matrix}\right.
\end{align*}
imply the following choices \begin{align*}
     \text{(i)} \text{ } & \, \alpha_1 = \epsilon_{i_{j+1}} - \epsilon_{i_a}, \text{ }   \alpha_2 = \epsilon_{i_a} - \epsilon_{i_{\ell}}, \text{ }   \alpha_3 = \epsilon_{i_\ell} - \epsilon_{i_{j+1}}; \\
   \text{(ii)} \text{ } & \, \alpha_1 = \epsilon_{i_{a+1}} - \epsilon_{i_j}, \text{ }   \alpha_2 = \epsilon_{i_j} - \epsilon_{i_\ell}, \text{ }   \alpha_3 = \epsilon_{i_\ell} - \epsilon_{i_{a+1}}; \\
   \text{(iii)} \text{ } & \, \alpha_1 = \epsilon_{i_{j+1}} - \epsilon_{i_{\ell+1}}, \text{ }   \alpha_2 = \epsilon_{i_a} - \epsilon_{j_3}, \text{ }   \alpha_3 = \epsilon_{j_3} - \epsilon_{i_{j+1}}.
\end{align*}
For each scenario mentioned previously, distinct values for $\mathcal{N}_1, \mathcal{N}_2$ and $\mathcal{N}_3$ will arise in expression $\eqref{eq:commur3}$. Given the analogous nature of the computations required for each situation, our focus will be solely directed towards analyzing case (i). A direct computation shows that $\{p,q\} = \mathcal{N}_1 + \mathcal{N}_2 + \mathcal{N}_3$, where
\begin{align*}
    \mathcal{N}_1 = & \, C_{\alpha_1,\beta_j} e_{\alpha_1+\beta_j} e_{\alpha_2} e_{\alpha_3}\prod_{w \neq    j}^re_{\beta_w} + C_{\alpha_1,\beta_a} e_{\alpha_1 + \beta_a} e_{\alpha_1 + \beta_a} e_{\alpha_2} e_{\alpha_3} \prod_{w \neq    a}^re_{\beta_w} \\
    =&\, C_{\alpha_1,\beta_j} \left(e_{\alpha_3} \prod_{w=j+1}^{\ell-1}e_{\beta_w}\right) \left(e_{\alpha_2}e_{\alpha_1+\beta_j} \prod_{k=1}^{a-1}e_{\beta_k}\prod_{v=\ell}^r e_{\beta_v}\right) + C_{\alpha_1,\beta_a} \left(e_{\alpha_3} \prod_{w=j+1}^{\ell-1}e_{\beta_w}\right) \left(e_{\alpha_2}e_{\alpha_1+\beta_a} \prod_{k=1}^j e_{\beta_k}\prod_{v=\ell}^r e_{\beta_v}\right) \\
     \mathcal{N}_2 =     &\, C_{\alpha_2,\beta_{k-1}} \left(e_{\alpha_3} \prod_{w=j+1}^{\ell-1}e_{\beta_w}\right) \left(e_{\alpha_1}e_{\alpha_2+\beta_{k-1}} \prod_{k=1}^je_{\beta_k}\prod_{v=\ell}^re_{\beta_v}\right) + C_{\alpha_2,\beta_\ell} \left(e_{\alpha_3} \prod_{w=j+1}^{\ell-1}e_{\beta_w}\right) \left(e_{\alpha_1}e_{\alpha_2+\beta_\ell}\prod_{k=1}^j e_{\beta_k}\prod_{v=\ell+1}^r e_{\beta_v}\right) \\
       \mathcal{N}_3 =     &\, C_{\alpha_3,\beta_{\ell-1}} \left(e_{\alpha_3 + \beta_{\ell-1}} \prod_{w=j+1}^{\ell-2}e_{\beta_w}\right) \left(e_{\alpha_2}e_{\alpha_1} \prod_{k=1}^je_{\beta_k}\prod_{v=\ell}^r e_{\beta_v}\right) + C_{\alpha_3,\beta_{j+1}} \left(e_{\alpha_3 + \beta_{j+1}} \prod_{w=j+2}^{\ell-1}e_{\beta_w}\right) \left(e_{\alpha_2}e_{\alpha_1} \prod_{k=1}^j e_{\beta_k}\prod_{v=\ell}^re_{\beta_v}\right).
\end{align*}
Here each monomial in the bracket $(\cdot)$ is indecomposable.   Additionally, for case (iii), since $\epsilon_{j_3}$ is undetermined, the expressions for all of $\mathcal{N}_1, \mathcal{N}_2$ and $\mathcal{N}_3$ will be expanded in the similar manner to those provided in case (b1).  

(b3) Under the assumption of (b1) and (b2), we further assume that there exist $\beta_l \neq \beta_\ell \in J_r$ such that $\alpha_2 + \beta_l, \text{ } \alpha_2 + \beta_\ell \in \Phi.$ Then by definition,
\begin{align*}
        \left.\begin{matrix}
           \alpha_1+ \beta_j =   \, \epsilon_{j_1} - \epsilon_{j_2} + \epsilon_{i_j} - \epsilon_{i_{j+1}} = \left\{\begin{matrix}
         \epsilon_{i_j}  - \epsilon_{j_2}  & \text{ if } \epsilon_{j_1} = \epsilon_{i_{j+1}} \\
      \epsilon_{j_1}    - \epsilon_{i_{j+1}} & \hskip -0.3cm\text{ if } \epsilon_{j_2} = \epsilon_{i_j}
    \end{matrix}\right.  \\
     \alpha_1+ \beta_a =   \, \epsilon_{k_1} - \epsilon_{j_2} + \epsilon_{i_a} - \epsilon_{i_{a+1}} = \left\{\begin{matrix}
         \epsilon_{i_a}  - \epsilon_{j_2}  & \text{ if } \epsilon_{j_1} = \epsilon_{i_{a+1}} \\
      \epsilon_{j_1}    - \epsilon_{i_{a+1}} & \hskip -0.3cm \text{ if } \epsilon_{j_2} = \epsilon_{i_a}
    \end{matrix}\right.
       \end{matrix} \right\} & \, \Longrightarrow \alpha_1 = \epsilon_{i_{j+1}} - \epsilon_{i_a} \text{ or } \epsilon_{i_{a+1}} - \epsilon_{i_j} \\
  \left. \begin{matrix}
    \alpha_2+ \beta_\ell =  \, \epsilon_{j_2} - \epsilon_{j_3} + \epsilon_{i_\ell} - \epsilon_{i_{\ell+1}} = \left\{\begin{matrix}
         \epsilon_{i_\ell}  - \epsilon_{j_3}  & \text{ if } \epsilon_{j_2} = \epsilon_{i_{\ell+1}} \\
      \epsilon_{j_2}    - \epsilon_{i_{\ell+1}} & \hskip -0.3cm \text{ if } \epsilon_{j_3} = \epsilon_{i_\ell}
    \end{matrix}\right.  \\
    \alpha_2 + \beta_l =   \epsilon_{j_2} - \epsilon_{j_3} + \epsilon_{i_l} - \epsilon_{i_{l+1}} = \left\{\begin{matrix}
         \epsilon_{i_l}  - \epsilon_{j_3}  &   \text{ if } \epsilon_{j_2} = \epsilon_{i_{l+1}} \\
      \epsilon_{j_2}    - \epsilon_{i_{l+1}} & \hskip -0.20cm \text{ if } \epsilon_{j_3} = \epsilon_{i_l}
    \end{matrix}\right.
\end{matrix} \right\} & \, \Longrightarrow \alpha_2 = \epsilon_{i_{\ell+1}} - \epsilon_{i_l} \text{ or } \epsilon_{i_{l+1}} - \epsilon_{i_\ell} .
\end{align*}
This gives
\begin{align*}
    \text{ (i) } \text{ } \alpha_1 = & \, \epsilon_{i_{j+1}} - \epsilon_{i_a} , \text{ } \alpha_2 =  \epsilon_{i_{\ell+1}} - \epsilon_{i_l}, \text{ } \alpha_3 = \epsilon_{i_{j+1}} -   \epsilon_{i_l} \text{ or } \epsilon_{i_a}   -\epsilon_{i_{\ell+1}}  ; \\
    \text{ (ii) } \text{ } \alpha_1 = & \, \epsilon_{i_{j+1}} - \epsilon_{i_a} , \text{ } \alpha_2 =  \epsilon_{i_{l+1}} - \epsilon_{i_\ell}, \text{ } \alpha_3 = \epsilon_{i_{j+1}} -   \epsilon_{i_\ell} \text{ or } \epsilon_{i_a} -\epsilon_{i_{l+1}}    ; \\
    \text{ (iii) } \text{ } \alpha_1 = & \, \epsilon_{i_{a+1}} - \epsilon_{i_j}, \text{ } \alpha_2 =  \epsilon_{i_{l+1}} - \epsilon_{i_\ell}, \text{ } \alpha_3 =   \epsilon_{i_\ell} - \epsilon_{i_{a+1}} \text{ or }   \epsilon_{i_j} - \epsilon_{i_{l+1}} ; \\
    \text{ (iv) } \text{ } \alpha_1 = & \, \epsilon_{i_{a+1}} - \epsilon_{i_j}, \text{ } \alpha_2 =  \epsilon_{i_{l+1}} - \epsilon_{i_\ell}, \text{ } \alpha_3 =     \epsilon_{i_\ell}-\epsilon_{i_{a+1}} \text{ or } \epsilon_{i_j} - \epsilon_{i_{l+1}}   .
\end{align*}
Given that \( \alpha_1, \alpha_2 \), and \( \alpha_3 \) are defined, the Poisson bracket decomposition for all the cases mentioned above matches the Poisson bracket decomposition in part (i) of case (b1), based on the computations in case (b1). Hence, we omit all the cases list above. %In the following, we will present only the potential options for the roots.

\vskip 0.4cm

(C) Consider that all $\alpha_1, \alpha_2$ and $\alpha_3$ are connected to some roots in $J_r$. Without loss of generality, assume that there exists $\beta_j \neq \beta_\ell \neq \beta_s \in J_r$ such that $\alpha_1 + \beta_j,\alpha_2 + \beta_\ell,\alpha_3 + \beta_s \in \Phi $. Then from the relation below, \begin{align}
\nonumber
        \alpha_1+ \beta_j = & \, \epsilon_{j_1} - \epsilon_{j_2} + \epsilon_{i_j} - \epsilon_{i_{j+1}} = \left\{\begin{matrix}
         \epsilon_{i_j}  - \epsilon_{j_2}  & \text{ if } \epsilon_{j_1} = \epsilon_{i_{j+1}} \\
      \epsilon_{j_1}    - \epsilon_{i_{j+1}} & \hskip -0.3cm\text{ if } \epsilon_{j_2} = \epsilon_{i_{j}}
    \end{matrix}\right.   \\
     \alpha_2+ \beta_\ell = & \, \epsilon_{j_2} - \epsilon_{j_3} + \epsilon_{i_\ell} - \epsilon_{i_{\ell+1}} = \left\{\begin{matrix}
         \epsilon_{i_\ell}  - \epsilon_{j_3}  & \hskip 0.05 cm\text{ if } \epsilon_{j_2} = \epsilon_{i_{\ell+1}} \\
      \epsilon_{j_2}    - \epsilon_{i_{\ell+1}} & \hskip -0.23cm\text{ if } \epsilon_{j_3} = \epsilon_{i_\ell}
    \end{matrix}\right.  \nonumber   \\
    \alpha_3 + \beta_s =  & \epsilon_{j_3} - \epsilon_{j_1} + \epsilon_{i_s} - \epsilon_{i_{s+1}} = \left\{\begin{matrix}
         \epsilon_{i_s}  - \epsilon_{j_1}  & \hskip 0.1cm  \text{ if } \epsilon_{j_3} = \epsilon_{i_{s+1}} \\
      \epsilon_{j_3}    - \epsilon_{i_{s+1}} & \hskip -0.05cm \text{ if } \epsilon_{j_1} = \epsilon_{i_s}.
    \end{matrix}\right. \label{eq:a16}
\end{align}
From this, we can infer the following potential choices \begin{align*}
   \text{(i)} \text{ } & \, \alpha_1 = \epsilon_{i_{j+1}} - \epsilon_{i_{\ell+1}}, \text{ }   \alpha_2 = \epsilon_{i_{\ell+1}} - \epsilon_{i_{s+1}}, \text{ }   \alpha_3 = \epsilon_{i_{s+1}} - \epsilon_{i_{j+1}}; \\
   \text{(ii)} \text{ } & \, \alpha_1 = \epsilon_{i_s} - \epsilon_{i_j}, \text{ }   \alpha_2 = \epsilon_{i_j} - \epsilon_{i_\ell}, \text{ }   \alpha_3 = \epsilon_{i_\ell} - \epsilon_{i_s}.
\end{align*}  Analogously to the reasoning in (A) and (B), we can additionally postulate the existence of an extra root within $J_r$ such that each $\alpha_i$ is connected to either one or two roots in $J_r$. Since the classification method closely resembles the one previously detailed, we will illustrate just one scenario here: Under the assumption of Case (C), we further assume that there exist some roots $\beta_a \neq \beta_l \neq \beta_t$ such that $\alpha_1 + \beta_a , \alpha_2 + \beta_l,\alpha_3 + \beta_t \in \Phi.$ Then together with \eqref{eq:a16}, we obtain that
\begin{align}
\nonumber
        \left.\begin{matrix}
           \alpha_1+ \beta_j =   \, \epsilon_{j_1} - \epsilon_{j_2} + \epsilon_{i_j} - \epsilon_{i_{j+1}} = \left\{\begin{matrix}
         \epsilon_{i_j}  - \epsilon_{j_2}  & \text{ if } \epsilon_{j_1} = \epsilon_{i_{j+1}} \\
      \epsilon_{j_1}    - \epsilon_{i_{j+1}} & \hskip -0.3 cm\text{ if } \epsilon_{j_2} = \epsilon_{i_j}
    \end{matrix}\right.  \\
    \\
     \alpha_1+ \beta_a =   \, \epsilon_{k_1} - \epsilon_{j_2} + \epsilon_{i_a} - \epsilon_{i_{a+1}} = \left\{\begin{matrix}
         \epsilon_{i_a}  - \epsilon_{j_2}  & \text{ if } \epsilon_{j_1} = \epsilon_{i_{a+1}} \\
      \epsilon_{j_1}    - \epsilon_{i_{a+1}} & \hskip -0.3 cm\text{ if } \epsilon_{j_2} = \epsilon_{i_a}
    \end{matrix}\right.
       \end{matrix} \right\} & \Longrightarrow \alpha_1 = \epsilon_{i_{j+1}} - \epsilon_{i_a} \text{ or } \epsilon_{i_{a+1}} - \epsilon_{i_j} ;\\
  \left. \begin{matrix}
    \alpha_2+ \beta_\ell =  \, \epsilon_{j_2} - \epsilon_{j_3} + \epsilon_{i_\ell} - \epsilon_{i_{\ell+1}} = \left\{\begin{matrix}
         \epsilon_{i_\ell}  - \epsilon_{j_3}  & \text{ if } \epsilon_{j_2} = \epsilon_{i_{\ell+1}} \\
      \epsilon_{j_2}    - \epsilon_{i_{\ell+1}} & \hskip -0.3 cm\text{ if } \epsilon_{j_3} = \epsilon_{i_\ell}
    \end{matrix}\right.   \\
    \\
    \alpha_2 + \beta_l =   \epsilon_{j_2} - \epsilon_{j_3} + \epsilon_{i_l} - \epsilon_{i_{l+1}} = \left\{\begin{matrix}
         \epsilon_{i_l}  - \epsilon_{j_1}  & \text{ if } \epsilon_{j_2} = \epsilon_{i_{l+1}} \\
      \epsilon_{j_3}    - \epsilon_{i_{l+1}} & \hskip -0.3 cm\text{ if } \epsilon_{j_3} = \epsilon_{i_l}
    \end{matrix}\right.
\end{matrix} \right\} & \Longrightarrow \alpha_2 = \epsilon_{i_{\ell+1}} - \epsilon_{i_l} \text{ or } \epsilon_{i_{l+1}} - \epsilon_{i_\ell} ;   \\
\nonumber
  \left. \begin{matrix}     \alpha_3 + \beta_s =    \epsilon_{j_3} - \epsilon_{j_1} + \epsilon_{i_s} - \epsilon_{i_{s+1}} = \left\{\begin{matrix}
         \epsilon_{i_s}  - \epsilon_{j_1}  & \text{ if } \epsilon_{j_3} = \epsilon_{i_{s+1}} \\
      \epsilon_{j_3}    - \epsilon_{i_{s+1}} & \hskip -0.3 cm\text{ if } \epsilon_{j_1} = \epsilon_{i_s}
    \end{matrix}\right. \\
    \\
   \alpha_3 + \beta_t =   \epsilon_{j_3} - \epsilon_{j_1} + \epsilon_{i_t} - \epsilon_{i_{t+1}} = \left\{\begin{matrix}
         \epsilon_{i_t}  - \epsilon_{j_1}  & \text{ if } \epsilon_{j_3} = \epsilon_{i_{t+1}} \\
      \epsilon_{j_3}    - \epsilon_{i_{t+1}} & \hskip -0.3 cm\text{ if } \epsilon_{j_1} = \epsilon_{i_t}
    \end{matrix}\right.  \end{matrix} \right\} & \Longrightarrow \alpha_3 = \epsilon_{i_{s+1}} - \epsilon_{i_t} \text{ or } \epsilon_{i_{t+1}} - \epsilon_{i_s} .
\end{align}
Thus, there are $6$ possible choices for $\alpha_i$. By the constraints $\alpha_1 + \alpha_2 + \alpha_3 = 0,$ we can deduce the relations between the indices $\{a,l,t\}$ and $\{j,\ell,s\}.$ Due to the classification being analogous to that of case (b3), we shall refrain from detailing each individual possibility in this context. For instance,
\begin{align*}
    \alpha_1 = \epsilon_{i_{j+1}} - \epsilon_{i_a}, \text{ }  \quad \alpha_2 =  \epsilon_{i_{\ell+1}} - \epsilon_{i_l}, \text{ } \quad  \alpha_3 = \epsilon_{i_{s+1}} - \epsilon_{i_t}.
\end{align*}
From the constraint, we deduce that
\begin{center}
   (i) $a = \ell+1, l = s+1$ and $t = j+1 $; \quad \quad (ii) $j+1 = l, a = s+1 $ and $ t = \ell + 1.$
\end{center}

From case (i), $ \alpha_1 = \epsilon_{i_{j+1}} - \epsilon_{i_{\ell+1}}, \text{ } \alpha_2 =  \epsilon_{i_{\ell+1}} - \epsilon_{i_{s+1}}, \text{ } \alpha_3 = \epsilon_{i_{s+1}} - \epsilon_{i_{j+1}}$. In this scenario, the expression for $\{p,q\}$ will match the one derived in case (i) of part (b2). Similar argument holds for case (ii).

\bibliographystyle{unsrt}
\bibliography{bibliography.bib}

\begin{thebibliography}{10}

\bibitem{MR1184379}
T.~Tjin.
\newblock Finite {$W$}-algebras.
\newblock {\em Phys. Lett. B}, 292(1-2):60--66, 1992.

\bibitem{MR1255424}
J.~de~Boer and T.~Tjin.
\newblock Quantization and representation theory of finite {$W$} algebras.
\newblock {\em Comm. Math. Phys.}, 158(3):485--516, 1993.

\bibitem{MR1173277}
Ya.\~I. Granovski\v{i}, I.~M. Lutzenko, and A.~S. Zhedanov.
\newblock Mutual integrability, quadratic algebras, and dynamical symmetry.
\newblock {\em Ann. Physics}, 217(1):1--20, 1992.

\bibitem{MR1306244}
D.~Bonatsos, C.~Daskaloyannis, and K.~Kokkotas.
\newblock Deformed oscillator algebras for two-dimensional quantum superintegrable systems.
\newblock {\em Phys. Rev. A (3)}, 50(5):3700--3709, 1994.

\bibitem{MR1806263}
C.~Daskaloyannis.
\newblock Polynomial {P}oisson algebras for two-dimensional classical superintegrable systems and polynomial associative algebras for quantum superintegrable systems.
\newblock volume~50, pages 1209--1214. 2000.
\newblock Quantum groups and integrable systems (Prague, 2000).

\bibitem{MR2804560}
S.~Post.
\newblock Models of quadratic algebras generated by superintegrable systems in 2{D}.
\newblock {\em SIGMA Symmetry Integrability Geom. Methods Appl.}, 7:Paper 036, 20, 2011.

\bibitem{MR3493688}
N.~Reshetikhin.
\newblock Degenerately integrable systems.
\newblock {\em Zap. Nauchn. Sem. S.-Peterburg. Otdel. Mat. Inst. Steklov. (POMI)}, 433:224--245, 2015.

\bibitem{MR3797912}
M.~F. Hoque, I.~Marquette, and Y.-Z. Zhang.
\newblock Recurrence approach and higher order polynomial algebras for superintegrable monopole systems.
\newblock {\em J. Math. Phys.}, 59(5):052101, 10, 2018.

\bibitem{Correa2020}
F.~Correa, M.~A. del Olmo, I.~Marquette, and J.~Negro.
\newblock Polynomial algebras from $\mathfrak{su}(3)$ and a quadratically superintegrable model on the two sphere.
\newblock {\em J. Phys. A}, 54(1):015205, 2020.

\bibitem{MR3227700}
Vincent~X. Genest, L.~Vinet, and A.~Zhedanov.
\newblock Superintegrability in two dimensions and the {R}acah-{W}ilson algebra.
\newblock {\em Lett. Math. Phys.}, 104(8):931--952, 2014.

\bibitem{MR3119484}
W.~Miller~Jr, S.~Post, and P.~Winternitz.
\newblock Classical and quantum superintegrability with applications.
\newblock {\em J. Phys. A}, 46(42):423001, 97, 2013.

\bibitem{MR2492581}
I.~Marquette.
\newblock Superintegrability with third order integrals of motion, cubic algebras, and supersymmetric quantum mechanics. {I}. {R}ational function potentials.
\newblock {\em J. Math. Phys.}, 50(1):012101, 23, 2009.

\bibitem{MR2566882}
I.~Marquette.
\newblock Superintegrability with third order integrals of motion, cubic algebras, and supersymmetric quantum mechanics. {II}. {P}ainlev\'{e} transcendent potentials.
\newblock {\em J. Math. Phys.}, 50(9):095202, 18, 2009.

\bibitem{MR4256391}
H.~De~Bie, P.~Iliev, W.~van~de Vijver, and L.~Vinet.
\newblock The {R}acah algebra: an overview and recent results.
\newblock In {\em Lie groups, {N}umber {T}heory, and {V}ertex {A}lgebras}, volume 768 of {\em Contemp. Math.}, pages 3--20. Amer. Math. Soc., [Providence], RI, [2021] \copyright 2021.

\bibitem{MR4131326}
G.~Bergeron, N.~Cramp\'{e}, S.~Tsujimoto, L.~Vinet, and A.~Zhedanov.
\newblock The {H}eun-{R}acah and {H}eun-{B}annai-{I}to algebras.
\newblock {\em J. Math. Phys.}, 61(8):081701, 15, 2020.

\bibitem{Latini_2019}
D.~Latini.
\newblock Universal chain structure of quadratic algebras for superintegrable systems with coalgebra symmetry.
\newblock {\em J. Phys. A}, 52(12):125202, feb 2019.

\bibitem{LATINI2021168397}
D.~Latini, I.~Marquette, and Y.-Z. Zhang.
\newblock Embedding of the racah algebra {$R(n)$} and superintegrability.
\newblock {\em Ann. Physics}, 426:168397, 2021.

\bibitem{Latini_2021}
D.~Latini, I.~Marquette, and Y.-Z. Zhang.
\newblock Racah algebra {$R(n)$} from coalgebraic structures and chains of {$R(3)$} substructures.
\newblock {\em J. Phys. A}, 54(39):395202, 2021.

\bibitem{marquette2023infinite}
I.~Marquette, J.~Zhang, and Y.-Z. Zhang.
\newblock Infinite-dimensional representations of cubic and quintic algebras and special functions.
\newblock {\em Eur. Phys. J. Plus}, 138(6):1--17, 2023.

\bibitem{bernard2023bethe}
P.-A. Bernard, G.n Carcone, N.~Cramp{\'e}, and L.~Vinet.
\newblock Bethe ansatz diagonalization of the heun--racah operator.
\newblock {\em Lett. Math. Phys.}, 113(1):8, 2023.

\bibitem{MR0411412}
A.~Peccia and R.~T. Sharp.
\newblock Number of independent missing label operators.
\newblock {\em J. Math. Phys.}, 17(7):1313--1314, 1976.

\bibitem{MR2276736}
R.~Campoamor-Stursberg.
\newblock Number of missing label operators and upper bounds for dimensions of maximal {L}ie subalgebras.
\newblock {\em Acta Phys. Polon. B}, 37(10):2745--2760, 2006.

\bibitem{MR4710584}
R.~Campoamor-Stursberg and I.~Marquette.
\newblock Decomposition of enveloping algebras of simple {L}ie algebras and their related polynomial algebras.
\newblock {\em J. Lie Theory}, 34(1):17--40, 2024.

\bibitem{campoamor2023algebraic}
R.~Campoamor-Stursberg, D.~Latini, I.~Marquette, and Y.-Z. Zhang.
\newblock Algebraic (super-) integrability from commutants of subalgebras in universal enveloping algebras.
\newblock {\em J. Phys. A}, 56(4):045202, 2023.

\bibitem{MR4411095}
N.~Cramp{\'e}, D.~Shaaban~Kabakibo, and L.~Vinet.
\newblock The {$SU(3)\supset SO(3)$} missing label problem and the analytical {B}ethe ansatz.
\newblock {\em Int. J. Mod. Phys. A}, 37(8):Paper No. 2250038, 16, 2022.

\bibitem{MR4660510}
R.~Campoamor-Stursberg, D.~Latini, I.~Marquette, and Y.-Z. Zhang.
\newblock Polynomial algebras from {L}ie algebra reduction chains {$\mathfrak{g}\supset\mathfrak{g}'$}.
\newblock {\em Ann. Physics}, 459:(19):169496, 2023.

\bibitem{MR0760556}
M.~Iosifescu and H.~Scutaru.
\newblock Poisson bracket realizations of {L}ie algebras and subrepresentations of {$({\rm ad}\sp{\otimes k})\sb{s}$}.
\newblock {\em J. Math. Phys.}, 25(10):2856--2862, 1984.

\bibitem{MR4355741}
R.~Campoamor-Stursberg and I.~Marquette.
\newblock Quadratic algebras as commutants of algebraic {H}amiltonians in the enveloping algebra of {S}chr\"{o}dinger algebras.
\newblock {\em Ann. Physics}, 437:Paper No. 168694, 16, 2022.

\bibitem{MR1451138}
J.~Dixmier.
\newblock {\em Alg\`ebres enveloppantes}.
\newblock Gauthier-Villars, Paris, 1974.

\bibitem{MR191995}
E.~G. Beltrametti and A.~Blasi.
\newblock On the number of {C}asimir operators associated with any {L}ie group.
\newblock {\em Phys. Lett.}, 20:62--64, 1966.

\bibitem{MR1520346}
L.~E. Dickson.
\newblock Recent {P}ublications: {R}eviews: {V}orlesungen \"uber die {T}heorie der algebraischen {Z}ahlen.
\newblock {\em Amer. Math. Monthly}, 31(1):45--46, 1924.

\bibitem{marquette2023algebraic}
I.~Marquette, J.~Zhang, and Y.-Z. Zhang.
\newblock Algebraic structures and {H}amiltonians from the equivalence classes of 2{D} conformal algebras.
\newblock {\em arXiv preprint arXiv:2309.11030}, 2023.

\bibitem{MR0204094}
M.~Pauri and G.~M. Prosperi.
\newblock Canonical realizations of {L}ie symmetry groups.
\newblock {\em J. Math. Phys.}, 7:366--375, 1966.

\bibitem{MR2515551}
L.~J. Boya and R.~Campoamor-Stursberg.
\newblock Commutativity of missing label operators in terms of {B}erezin brackets.
\newblock {\em J. Phys. A}, 42(23):235203, 12, 2009.

\bibitem{gtp}
R.~Campoamor-Stursberg and M.~Rausch de~Traubenberg.
\newblock {\em Group {T}heory in {P}hysics: {A} {P}ractitioner's {G}uide}.
\newblock World Scientific, Singapore, 2018.

\bibitem{MR1322960}
D.~Eisenbud.
\newblock {\em Commutative {A}lgebra}.
\newblock Springer-Verlag, New York, 1995.

\bibitem{MR92620}
J.~P. Elliott.
\newblock Collective motion in the nuclear shell model. {I}. {C}lassification schemes for states of mixed configurations.
\newblock {\em Proc. Roy. Soc. London Ser. A}, 245:128--145, 1958.

\bibitem{MR94178}
J.~P. Elliott.
\newblock Collective motion in the nuclear shell model. {II}. {T}he introduction of intrinsic wave-functions.
\newblock {\em Proc. Roy. Soc. London Ser. A}, 245:562--581, 1958.

\bibitem{MR1920389}
A.~W. Knapp.
\newblock {\em Lie {G}roups {B}eyond an {I}ntroduction}, volume 140 of {\em Progress in Mathematics}.
\newblock Birkh\"{a}user Boston, Inc., Boston, MA, second edition, 2002.

\bibitem{campoamor2024superintegrable}
R.~Campoamor-Stursberg, D.~Latini, I.~Marquette, J.~Zhang, and Y.-Z. Zhang.
\newblock Superintegrable systems associated to commutants of cartan subalgebras in enveloping algebras.
\newblock {\em arXiv preprint arXiv:2406.01958}, 2024.

\bibitem{campoamor2023polynomial}
R.~Campoamor-Stursberg, D.~Latini, I.~Marquette, and Y.-Z. Zhang.
\newblock Polynomial algebras from commutants: Classical and quantum aspects of $\mathcal{A}_3$.
\newblock In {\em Journal of Physics: Conference Series}, volume 2667, page 012037. IOP Publishing, 2023.

\bibitem{MR1081533}
D.~Bonatsos.
\newblock {\em Interacting boson models of nuclear structure}, volume~10 of {\em Oxford Studies in Nuclear Physics}.
\newblock The Clarendon Press, Oxford University Press, New York, 1988.
\newblock Oxford Science Publications.

\bibitem{MR1152802}
F.~Iachello and P.~Van~Isacker.
\newblock {\em The interacting boson-fermion model}.
\newblock Cambridge Monographs on Mathematical Physics. Cambridge University Press, Cambridge, 1991.

\bibitem{Wig}
E.~P. Wigner.
\newblock On the consequences of the symmetry of the nuclear {H}amiltonian on the spectroscopy of nuclei.
\newblock {\em Phys. Rev.}, 51:106--119, 1937.

\bibitem{Draayer}
J.~P. Draayer.
\newblock {$SU(4)\supset SU(2)\otimes SU(2)$} projection techniques.
\newblock {\em J. Math. Phys.}, 11:3225--3238, 1970.

\bibitem{Brunet}
M.~Brunet and M.~Resnikoff.
\newblock The representations {$U(4)\supset U(2)\otimes U(2)$}.
\newblock {\em J. Math. Phys.}, 11:1474--1481, 1970.

\bibitem{Hecht}
K.~T. Hecht and S.~Ch. Pang.
\newblock On the {W}igner supermultiplet scheme.
\newblock {\em J. Math. Phys.}, 10:1571--1616, 1969.

\end{thebibliography}

\end{document}